\documentclass[11pt,twoside]{article}

\usepackage[margin=1in]{geometry}
\usepackage[T1]{fontenc}
\usepackage{amsmath,amssymb,amsthm,mathtools,bm}
\usepackage{booktabs}
\usepackage{enumitem}
\usepackage[hidelinks]{hyperref}
\hypersetup{
    colorlinks=true,
    citecolor=blue,
    filecolor=black,
    linkcolor=red,
    urlcolor=blue,
    bookmarksopen=true,
    pdfstartview=FitH
}
\usepackage{comment}
\usepackage{natbib}
\usepackage{xcolor}
\usepackage{tikz,etoolbox,float,needspace,framed}
\usepackage[normalem]{ulem}
\usepackage{hyperref}
\usepackage{algorithm,algpseudocode}

\usepackage{setspace} 
\usepackage{fancyhdr}
\allowdisplaybreaks[4]
\setlist{nosep,leftmargin=2em}

\newtheorem{theorem}{Theorem}[section]
\newtheorem{proposition}[theorem]{Proposition}
\newtheorem{lemma}[theorem]{Lemma}
\newtheorem{example}[theorem]{Example}
\newtheorem{corollary}[theorem]{Corollary}
\theoremstyle{definition}
\newtheorem{assumption}[theorem]{Assumption}
\newtheorem{definition}[theorem]{Definition}
\theoremstyle{remark}
\newtheorem{remark}[theorem]{Remark}
\renewcommand{\(}{\left(}
\renewcommand{\)}{\right)}
\newcommand{\R}{\mathbb{R}}
\newcommand{\E}{\mathbb{E}}
\newcommand{\Pp}{\mathbb{P}}
\newcommand{\1}{\mathbf{1}}
\newcommand{\Sn}{\mathfrak{S}}
\newcommand{\Tcal}{\mathcal{T}}

\newcommand{\Var}{\operatorname{Var}}
\newcommand{\norm}[1]{\left\lVert #1\right\rVert}
\newcommand{\abs}[1]{\left\lvert #1\right\rvert}
\newcommand{\argmax}{\operatorname*{arg\,max}}

\newcommand{\dd}{\,\mathrm{d}}
\newcommand{\Un}{\mathbb{U}_n}
\newcommand{\Sym}{\mathsf{S}}
\newcommand{\Mcal}{\mathcal{M}}

\newcommand{\va}{\boldsymbol a}

\newcommand{\vl}{\boldsymbol\ell}
\newcommand{\vs}{\boldsymbol s}
\newcommand{\vi}{\boldsymbol i}
\newcommand{\vv}{\boldsymbol v}
\newcommand{\vu}{\boldsymbol u}

\providecommand{\conv}{\operatorname{conv}}
\providecommand{\Vol}{\operatorname{Vol}}
\providecommand{\cH}{\mathcal H}

\renewcommand{\qed}{\hfill \mbox{\raggedright \rule{0.08in}{0.08in}}} 
 
\renewenvironment{proof}[1][\proofname]{{\noindent\sc#1.}}{\qed\vspace{15pt}} 

\title{\bf\sc\Large Optimal Allocation and Volume under Surface}
\author{Kai Feng\thanks{This paper originated from ideas independently conceived by the first author.}\thanks{School of Finance, Renmin University of China. Email: kfeng@ruc.edu.cn.} \and 
Han Hong\thanks{Department of Economics, Stanford University. Email: doubleh@stanford.edu.} \and 
Jessie Li\thanks{Department of Economics, University of California Santa Cruz. Email: jeqli@ucsc.edu} \and 
Wenshi Wei\thanks{School of Economics and Management, Beijing Jiaotong University. Email: wswei@bjtu.edu.cn}}
\date{\today}

\begin{document}
\maketitle

\begin{abstract}
\noindent {\sc Abstract.}
This paper develops a framework for estimation and inference on the volumes of sets that are projections of critical function sets, focusing particularly on the convex body beneath the optimal receiver operating characteristic (ROC) surface. 
Specifically, we propose a volume calculation method that first uses an Aumann expectation representation and then applies Minkowski mixed volumes. 
Using this framework, we show that the population volume under the ROC surface (VUS) is proportional to the expectation of a symmetric U-statistic kernel. 
We then propose a double/debiased machine learning estimator of the VUS, derive its asymptotic properties, and develop an inference procedure. 
Further applications of this framework include an analysis of the feasible error set across pre-defined groups and a natural generalization of the Gini coefficient for measuring inequality. 

\vspace{15pt}
\noindent {\sc Keywords}: Volume under the Surface, ROC Surface, Polytope Valued Random Variables, Minkowski Mixed Volume, Double/Debiased machine learning U-statistics, Gini Coefficient 

\vspace{15pt}
\noindent {\sc JEL Codes}: C14, C21, C44
\end{abstract}

\section{Introduction}

Individual-level treatments allocation in light of heterogeneous characteristics is central to problems concerning policy welfare and fairness \citep{athey2021policy,liang2026algorithm}. 
Classification problems, in which outcomes are fully observed and normalized to binary values, constitute an importance special case of general optimal allocation problems. 
They are also of independent interest given the growing role of artificial intelligence in decision-making. 
Although statistical methods for binary classification are well established, 
there is still a lack of principled criteria for evaluating multi-class classifiers.

This paper provides an explicit formulation, estimators and inference procedures for the volume under ROC surfaces with $K \geq 2$ classes. 
The (optimal) ROC surface can be represented as the value function of a constrained linear programming problem over the space of critical functions allowing for randomized assignment rules. 
Our key observation is that the convex body beneath this surface is the Aumann expectation of a random polytope. 
Let $P_k$ denote the feature distribution in class $k$, $P$ the corresponding population mixture, and $\mathfrak S_K$ the set of permutations of ${1,\ldots,K}$. For independent draws $X_1,\ldots,X_K\sim P$, we use Minkowski mixed volumes to obtain
\begin{align}\begin{split}\label{VUS main formula}
\mathrm{VUS}=\frac1{K!}
\E\left[\max_{\pi\in\mathfrak S_K}
\prod_{r=1}^K q_{\pi(r)}(X_r)\right], 
\end{split}\end{align}
where $q_k=\frac{dP_k}{dP}$. 



Since the fundamental works by \cite{manski1989anatomy} and \cite{manski1990nonparametric}, 
set identified models has been largely advanced and popularized. 
\cite{chernozhukov2007estimation} proposes confidence sets based on criterion functions that cover the true identification region with probability 
converging to a given level. 
\cite{beresteanu2008asymptotic} study estimation and inference procedures for a population identification region that can be written as a transformation of 
the Aumann expectation of a properly defined set-valued random variable. 
Unlike in the aforementioned works, the parameter of interest in this paper is the volume functional of certain sets, rather than the sets themselves.

Building on the VUS formula \eqref{VUS main formula}, we then propose a double/debiased machine learning estimator of the population VUS \citep{chernozhukov2018double,chernozhukov2022locally}.
Using an approximately optimal plug-in policy similar to those in \cite{feng2026statistical} and a refined margin assumption inspired by \cite{escanciano2026debiased}, we show that our estimator converges to the population VUS at essentially the $1/\sqrt{n}$ rate, with an asymptotically linear representation.\footnote{The margin assumptions in this paper is different from, for instance, those used in \cite{kitagawa2018should} and \cite{feng2026statistical}. 
We intentionally employ the ``excluding center'' type of margin assumption from works such as \cite{audibert2007fast} and \cite{escanciano2026debiased}. See section \ref{refined margin assumption}.}
We also propose a cross-validation scheme for $K$th-order U-statistics to achieve estimation efficiency. 

The volume under the surface (VUS) is commonly used as the three-class generalization of the AUC, although the interpretation of the VUS has been too difficult to formalize. 
Several papers (see e.g. \cite{scurfield1996multiple}, \cite{scurfield1998generalization}, \cite{mossman1999three}, \cite{dreiseitl2000comparing}, \cite{nakas2004ordered}, \cite{li2008roc}, \cite{he2008meaning}, \cite{wu2013optimal}, and \cite{wu2016roc}) have studied the meaning of the VUS for ordinal classification problems where a ranking among the different classes is determined based on a decision rule using an ordered series of cutoffs. 
However, to our knowledge, no paper has examined the meaning of the VUS in more general classification problems where there isn't necessarily a ranking among the different classes.

Our volume calculation method can also be applied to other problems in which the sets of interest can be written as projections of critical function sets. 
We give two examples.
The first is the volume of the feasible error set in \cite{liang2026algorithm}, and the second is a generalization of the Gini coefficient closely related to the Lorenz zonoid in \cite{koshevoy1996lorenz}. 
In these two applications, the volumes of the sets serve as measures of informativeness and inequality, respectively.
Recently, \cite{heikkuri2026subgroup} derives a decomposition of the Gini coefficient into within- and between-group components using the Brunn–Minkowski theorem. Mathematically, Minkowski mixed volumes and the Brunn–Minkowski theorem are both part of the Brunn–Minkowski theory of convex bodies, see \cite{Schneider2014}.

The rest of the paper is organized as follows.
Section \ref{population VUS section} develops the main Aumann expectation--mixed volume method in detail.
Section \ref{sample VUS section} proposes the double/debiased estimator for the VUS and establishes its convergence rates and inference results.
Section \ref{more examples section} provides the feasible error set and the generalized Gini coefficient examples. 
Section \ref{conclusion section} offers conclusions.

\section{Volume under the Receiver Operating Characteristic Surface}
\label{sec:roc-geometric-representation}\label{population VUS section}

In this section, we introduce a geometric framework for computing the volume under the multi-class ROC frontier.
The central object is the ROC body $\cH$, the region below the frontier, which collects the accuracy requirements generated by allocation rules using the observed features.
Combining an Aumann representation of $\cH$ with mixed-volume methods yields the explicit volume formula in Theorem \ref{thm:volume-formula}.
This formula expresses $\Vol_K(\cH)$ in terms of an expected maximum over class assignments and permits its calculation without reconstructing the frontier.
We also compare this framework with the Lagrangian approach and explain how the volume reflects the information that the observed features provide about class labels.

\subsection{Allocation rules and the ROC body}
\label{subsec:allocation-roc-body}

Let $O=(X,C)$ denote an observation, where $C\in[K]\coloneqq\{1,\ldots,K\}$ is its class label and $K\geq2$ is the number of classes.
The observed feature vector $X$ takes values in the measurable space $\mathcal X$, and $x\in\mathcal X$ denotes a realization.
Write $\rho_k=\Pp(C=k)>0$ for the class probabilities, $P_k$ for the conditional distribution of $X$ given $C=k$, and $P$ for the marginal distribution of $X$.
The law of total probability gives
\begin{align}\begin{split}\nonumber
P=\sum_{k=1}^K\rho_kP_k,
\quad
\sum_{k=1}^K\rho_k=1.
\end{split}\end{align}
Because every $\rho_k$ is positive, a set of feature values with zero probability under $P$ also has zero probability under $P_k$.
Thus, $P_k$ is absolutely continuous with respect to $P$, and we define the \emph{class density ratio} by
\begin{align}\begin{split}\nonumber
q_k=\frac{\dd P_k}{\dd P}.
\end{split}\end{align}
Write $p_k(x)\coloneqq\Pp(C=k\mid X=x)$ for the conditional probability of class $k$.
Then $p_k(x)=\rho_kq_k(x)$, or equivalently $q_k(x)=p_k(x)/\rho_k$.
Let $p=(p_1,\ldots,p_K)$ denote the vector of conditional class probabilities and $q=(q_1,\ldots,q_K)$ the vector of class density ratios.
The Radon--Nikodym derivative $q_k$ expresses integration under $P_k$ as weighted integration under $P$.
For every measurable $g:\mathcal X\to\R$ satisfying $\E_P[q_k(X)|g(X)|]<\infty$,
\begin{align}\begin{split}\label{eq:class-change-of-measure}
\E[g(X)\mid C=k]
=\int_{\mathcal X}g(x)\dd P_k(x)
=\int_{\mathcal X}q_k(x)g(x)\dd P(x)
=\E_P[q_k(X)g(X)].
\end{split}\end{align}

The class density ratios $q_k$ satisfy the following identities and bounds: 
\begin{align}\begin{split}\label{eq:q-properties}
\E_P[q_k(X)]=1,
\quad
\sum_{k=1}^K\rho_kq_k(x)=1,
\quad
0\leq q_k(x)\leq\rho_k^{-1}.
\end{split}\end{align}
Unless otherwise stated, equalities and inequalities between functions of $x$ are understood to hold $P$-almost everywhere.

Write $\R_+^K=[0,\infty)^K$ and let $\Delta^{K-1}\coloneqq\{a\in\R_+^K:\sum_{k=1}^Ka_k=1\}$ denote the probability simplex.

\begin{definition}[Randomized allocation rule]
\label{def:allocation-rule}
A \emph{randomized allocation rule} is a measurable map $\phi=(\phi_1,\ldots,\phi_K):\mathcal X\to\Delta^{K-1}$.
The rule $\phi$ uses only the observed features $X$ and, conditional on $X=x$, assigns class $k$ with probability $\phi_k(x)$.
\end{definition}

\begin{definition}[True-positive rates]
\label{def:true-positive-rates}
For a randomized allocation rule $\phi$, the \emph{class-$k$ true-positive rate} is $T_k(\phi)\coloneqq\E[\phi_k(X)\mid C=k]$, the probability of assigning an observation to class $k$ conditional on its true class being $k$.
The corresponding true-positive rate vector is $T(\phi)=(T_1(\phi),\ldots,T_K(\phi))$.
\end{definition}

By the change of measure in \eqref{eq:class-change-of-measure}, the class-$k$ true-positive rate $T_k(\phi)$ admits the representation
\begin{align}\begin{split}\label{eq:tp-rate-representation}
T_k(\phi)
&=\E[\phi_k(X)\mid C=k]
=\int_{\mathcal X}\phi_k(x)\dd P_k(x) \\
&=\int_{\mathcal X}q_k(x)\phi_k(x)\dd P(x)
=\E_P[q_k(X)\phi_k(X)].
\end{split}\end{align}

To describe the tradeoffs among the rates in $T(\phi)$, we maximize the class-1 rate subject to minimum accuracy requirements for the other classes.
For $k=2,\ldots,K$, let $t_k\in[0,1]$ denote the minimum required true-positive rate for class $k$.
The feasible set $\Lambda$ consists of the requirement vectors $t=(t_2,\ldots,t_K)$ for which there exists an allocation rule meeting all the requirements.
\begin{align}\begin{split}\nonumber
\Lambda
\coloneqq\left\{
t\in[0,1]^{K-1}:
T_k(\phi)\geq t_k,\ k=2,\ldots,K,
\text{ for some allocation rule }\phi
\right\}.
\end{split}\end{align}
For each $t\in\Lambda$, the greatest class-1 rate achievable under these requirements is
\begin{align}\begin{split}\label{eq:constrained-frontier}
\beta_1(t)
\coloneqq\max_\phi\left\{T_1(\phi):T_k(\phi)\geq t_k,\ k=2,\ldots,K\right\}.
\end{split}\end{align}
For every $t\in\Lambda$, an optimal allocation rule with value $\beta_1(t)$ exists, which will be proved in Proposition \ref{prop:tpr-region}.
As $t$ varies over $\Lambda$, these optimal values define the \emph{upper receiver operating characteristic (ROC) frontier}, the graph $\{(\beta_1(t),t):t\in\Lambda\}$.

\begin{definition}[ROC body]
\label{def:roc-body}
The \emph{ROC body} is the region $\cH\coloneqq\{(u_1,t):t\in\Lambda,\ 0\leq u_1\leq\beta_1(t)\}$ below the upper ROC frontier.
Its \emph{volume under the ROC surface} is $\Vol_K(\cH)$, where $\Vol_K$ denotes $K$-dimensional Lebesgue volume.
\end{definition}

The ROC body $\cH$ can also be described through the true-positive rate vectors $T(\phi)$ achieved by allocation rules.
For $u=(u_1,t)\in\cH$, let $\phi$ be an optimal allocation rule for the problem in \eqref{eq:constrained-frontier}.
It satisfies
\begin{align}\begin{split}\nonumber
T_1(\phi)=\beta_1(t)\geq u_1,
\quad
T_k(\phi)\geq t_k=u_k,\quad k=2,\ldots,K.
\end{split}\end{align}
Conversely, suppose $u\geq0$ and an allocation rule satisfies $T_k(\phi)\geq u_k$ for every $k$.
Then $t=(u_2,\ldots,u_K)\in\Lambda$ and
\begin{align}\begin{split}\nonumber
0\leq u_1\leq T_1(\phi)\leq\beta_1(t),
\end{split}\end{align}
so $u\in\cH$.
The two implications identify the ROC body $\cH$ with the set of feasible nonnegative requirement vectors $u$, giving
\begin{align}\begin{split}\label{eq:roc-body-feasibility}
\cH=\left\{u\in\R_+^K\mid u\leq T(\phi)\text{ for some allocation rule }\phi\right\}.
\end{split}\end{align}
The inequality $u\leq T(\phi)$ is understood coordinate-wise.

\subsection{Aumann representations of the ROC body}
\label{subsec:aumann-tpr-region}

To analyze the volume of the ROC body $\cH$, we first introduce the true-positive-rate region $\Gamma$, defined as follows.

\begin{definition}[True-positive-rate region]
\label{def:true-positive-rate-region}
The \emph{true-positive-rate region} $\Gamma\subseteq[0,1]^K$ is the set of all true-positive rate vectors $T(\phi)$ obtained as the randomized allocation rule $\phi$ varies.
\end{definition}

\begin{remark}[Requirements and achieved rates]
\label{rem:requirements-achieved-rates}
By \eqref{eq:roc-body-feasibility}, the condition $u\in\cH$ requires an allocation rule satisfying $T_k(\phi)\geq u_k$ for every class $k$, whereas $u\in\Gamma$ requires an allocation rule satisfying $T(\phi)=u$.
A feasible requirement vector $u\in\cH$ therefore need not equal $T(\phi)$ for any allocation rule $\phi$.
The ROC body $\cH$ is downward closed because, whenever $u\in\cH$ and $0\leq u'\leq u$, an allocation rule $\phi$ with $u\leq T(\phi)$ also satisfies $u'\leq T(\phi)$, and hence $u'\in\cH$.
\end{remark}

For each $t\in\Lambda$, the frontier value $\beta_1(t)$ in \eqref{eq:constrained-frontier} is the largest first coordinate $T_1(\phi)$ among the true-positive rate vectors $T(\phi)\in\Gamma$ satisfying $T_k(\phi)\geq t_k$ for $k=2,\ldots,K$.
Figure \ref{fig:roc-construction} illustrates this construction for $K=2$.

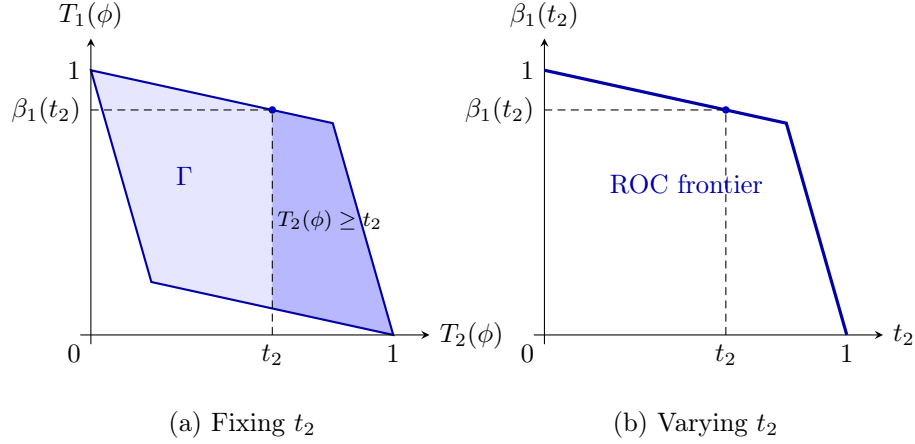
\begin{figure}[htbp!]
\centering
\begin{tikzpicture}[x=4cm,y=3.5cm,>=stealth,font=\small]
\begin{scope}
\fill[blue!10] (0,1) -- (0.8,0.8) -- (1,0) -- (0.2,0.2) -- cycle;
\fill[blue!28] (0.6,0.85) -- (0.8,0.8) -- (1,0) -- (0.6,0.1) -- cycle;
\draw[blue!65!black,thick] (0,1) -- (0.8,0.8) -- (1,0) -- (0.2,0.2) -- cycle;
\draw[->] (-0.035,0) -- (1.12,0) node[right] {$T_2(\phi)$};
\draw[->] (0,-0.035) -- (0,1.12) node[above] {$T_1(\phi)$};
\draw[densely dashed] (0.6,0) -- (0.6,0.85) -- (0,0.85);
\fill[blue!75!black] (0.6,0.85) circle (1.4pt);
\node[below left] at (0,0) {$0$};
\node[below] at (1,0) {$1$};
\node[left] at (0,1) {$1$};
\node[below] at (0.6,0) {$t_2$};
\node[left] at (0,0.85) {$\beta_1(t_2)$};
\node[blue!65!black] at (0.31,0.60) {$\Gamma$};
\node[font=\scriptsize,align=center] at (0.79,0.43) {$T_2(\phi)\geq t_2$};
\node[anchor=north] at (0.5,-0.25) {(a) Fixing $t_2$};
\end{scope}
\begin{scope}[xshift=6cm]
\draw[->] (-0.035,0) -- (1.12,0) node[right] {$t_2$};
\draw[->] (0,-0.035) -- (0,1.12) node[above] {$\beta_1(t_2)$};
\draw[blue!65!black,very thick] (0,1) -- (0.8,0.8) -- (1,0);
\draw[densely dashed] (0.6,0) -- (0.6,0.85) -- (0,0.85);
\fill[blue!75!black] (0.6,0.85) circle (1.4pt);
\node[below left] at (0,0) {$0$};
\node[below] at (1,0) {$1$};
\node[left] at (0,1) {$1$};
\node[below] at (0.6,0) {$t_2$};
\node[left] at (0,0.85) {$\beta_1(t_2)$};
\node[blue!65!black] at (0.47,0.57) {ROC frontier};
\node[anchor=north] at (0.5,-0.25) {(b) Varying $t_2$};
\end{scope}
\end{tikzpicture}
\caption{Construction of the upper ROC frontier}
\label{fig:roc-construction}
\begin{minipage}{0.94\linewidth}
\footnotesize
\textit{Note.} The feature takes two values with class-conditional probabilities $(0.8,0.2)$ and $(0.2,0.8)$.
The true-positive-rate region $\Gamma$ in Panel (a) has vertices $(0,1),(0.8,0.8),(1,0),(0.2,0.2)$, with coordinates ordered as $(T_2,T_1)$.
Its darker part satisfies $T_2(\phi)\geq t_2$.
The highest point in this part gives $\beta_1(t_2)$, and Panel (b) records these maxima as $t_2$ varies over $\Lambda=[0,1]$.
\end{minipage}
\end{figure}


To describe the true-positive-rate region $\Gamma$ underlying this frontier construction, we return to the integral representation of each rate $T_k(\phi)$ in \eqref{eq:tp-rate-representation}.
For each allocation rule $\phi$, collect the integrands with respect to $P$ into the vector 
\begin{align}\begin{split}\nonumber
f_\phi(x)\coloneqq(q_1(x)\phi_1(x),\ldots,q_K(x)\phi_K(x)). 
\end{split}\end{align}
Taking expectations coordinatewise gives
\begin{align}\begin{split}\label{eq:weighted-allocation-selection}
T(\phi)=\bigl(\E_P[q_1(X)\phi_1(X)],\ldots,\E_P[q_K(X)\phi_K(X)]\bigr)=\E_P[f_\phi(X)].
\end{split}\end{align}
Hence Definition \ref{def:true-positive-rate-region} gives
\begin{align}\begin{split}\nonumber
\Gamma=\left\{\E_P[f_\phi(X)]:\phi\text{ is a randomized allocation rule}\right\}.
\end{split}\end{align}

At each $x$, the possible values of the integrand $f_\phi(x)$ form the set
\begin{align}\begin{split}\label{eq:pointwise-selection-set}
F(x)
\coloneqq\left\{(q_1(x)a_1,\ldots,q_K(x)a_K):a\in\Delta^{K-1}\right\}.
\end{split}\end{align}

Each $F(x)$ is compact and convex, being a linear image of the probability simplex $\Delta^{K-1}$.
Aumann integration allows us to represent the true-positive-rate region $\Gamma$ as an integral of $F$ and transfer the compactness and convexity of the sets $F(x)$ to $\Gamma$.
To obtain this representation, we consider measurable selections of $F$.
A measurable function $f:\mathcal X\to\R^K$ satisfying $f(x)\in F(x)$ for $P$-almost every $x$ is called a \emph{measurable selection} of $F$.
For each allocation rule $\phi$, the function $f_\phi$ is therefore a measurable selection of $F$.

A measurable selection $f$ is \emph{integrable} if $\E_P[\lVert f(X)\rVert]<\infty$.
For an integrable selection $f=(f_1,\ldots,f_K)$, integration is coordinate-wise,
\begin{align}\begin{split}\label{eq:selection-vector-integral}
\E_P[f(X)] = \int_{\mathcal X}f(x)\dd P(x)
=\left(\int_{\mathcal X}f_1(x)\dd P(x),\ldots,\int_{\mathcal X}f_K(x)\dd P(x)\right).
\end{split}\end{align}

Collecting the expectations $\E_P[f(X)]$ of all integrable measurable selections $f$ of $F$ gives the set-valued integral introduced by \citet{Aumann1965}. 

\begin{definition}[Aumann integral and Aumann expectation]
\label{def:aumann-expectation}
For a set-valued function $F$ with $F(x)\subseteq\R^K$ for every $x\in\mathcal X$, the \emph{Aumann integral} $\int_{\mathcal X}F(x)\dd P(x)$ is the set of vectors $\E_P[f(X)]$ obtained as $f$ ranges over all integrable measurable selections of $F$.
Its \emph{Aumann expectation} is $\E_A[F(X)]\coloneqq\operatorname{cl}(\int_{\mathcal X}F(x)\dd P(x))$, where $\operatorname{cl}$ denotes closure.
\end{definition}

A set-valued function $F$ is \emph{measurable} if $\{x\in\mathcal X:F(x)\cap O\neq\varnothing\}$ is measurable for every open set $O\subseteq\R^K$.
The set-valued function $F$ is \emph{integrably bounded} if $\E_P[\sup_{y\in F(X)}\lVert y\rVert]<\infty$.
For an integrably bounded $F$, every measurable selection is integrable because
\begin{align}\begin{split}\nonumber
\E_P[\lVert f(X)\rVert]
\leq\E_P\!\left[\sup_{y\in F(X)}\lVert y\rVert\right]
<\infty.
\end{split}\end{align}

The following theorem shows that compactness and convexity of the sets $F(x)$ carry over to the Aumann integral of $F$, provided $F$ is measurable and integrably bounded.
\begin{theorem}
\label{thm:aumann-compactness}
Let $F$ be a measurable, integrably bounded set-valued function with nonempty compact values in $\R^K$.
Then its Aumann integral is compact and
\begin{align}\begin{split}\nonumber
\int_{\mathcal X}F(x)\dd P(x)=\E_A[F(X)].
\end{split}\end{align}
If $F(x)$ is also convex almost everywhere, this set is convex.
\end{theorem}

The compactness in Theorem \ref{thm:aumann-compactness} follows from \citet[Theorem 2.1.38, p. 250]{Molchanov2017}.
The equality between the Aumann integral and its Aumann expectation $\E_A[F(X)]$ follows because a compact subset of $\R^K$ is closed.
If $F(x)$ is convex for $P$-almost every $x$, any integrable measurable selections $f,g$ of $F$ give another such selection $h=\lambda f+(1-\lambda)g$ for every $\lambda\in[0,1]$.
The identity $\E_P[h(X)]=\lambda\E_P[f(X)]+(1-\lambda)\E_P[g(X)]$ proves convexity of the Aumann integral.



The following proposition establishes the Aumann representation of the true-positive-rate region $\Gamma$.
We show that the functions $f_\phi$ generated by allocation rules $\phi$ are exactly the integrable measurable selections of $F$.
Taking expectations gives $\Gamma=\E_A[F(X)]$, and Theorem \ref{thm:aumann-compactness} then yields compactness and convexity of $\Gamma$.

\begin{proposition}
\label{prop:tpr-region}
For the set-valued function $F$ in \eqref{eq:pointwise-selection-set},
\begin{align}\begin{split}\label{eq:tpr-aumann-representation}
\Gamma
  &=\left\{\E_P[f_\phi(X)]: \ \text{where}\ \phi:\mathcal X\to\Delta^{K-1}\text{ is measurable}\right\}\\
&=\int_{\mathcal X}F(x)\dd P(x)
=\E_A[F(X)].
\end{split}\end{align}
The set $\Gamma$ is compact and convex and contains $\Delta^{K-1}$.
For every $t\in\Lambda$, the constrained maximization problem in \eqref{eq:constrained-frontier} has a solution $\phi$ satisfying $T_1(\phi)=\beta_1(t)$.
\end{proposition}

\begin{proof}
\emph{Equality of the sets.}
For every allocation rule $\phi$, the function $f_\phi$ is a measurable selection of $F$.
Since $0\leq\phi_k\leq1$,
\begin{align}\begin{split}\nonumber
\E_P[\lVert f_\phi(X)\rVert]
\leq\sum_{k=1}^K\E_P[q_k(X)\phi_k(X)] 
\leq\sum_{k=1}^K\E_P[q_k(X)]
=K.
\end{split}\end{align}
Thus, $f_\phi$ is integrable and Definition \ref{def:aumann-expectation} gives
\begin{align}\begin{split}\nonumber
\Gamma\subseteq\int_{\mathcal X}F(x)\dd P(x).
\end{split}\end{align}

For the reverse inclusion, let $f=(f_1,\ldots,f_K)$ be an integrable measurable selection of $F$.
We construct an allocation rule $\phi$ satisfying $f=f_\phi$, so that its integral belongs to $\Gamma$.
Fix $x$ for which $f(x)\in F(x)$.
By \eqref{eq:pointwise-selection-set}, there is an $a\in\Delta^{K-1}$ with
\begin{align}\begin{split}\nonumber
f_k(x)=q_k(x)a_k,
\quad k=1,\ldots,K.
\end{split}\end{align}
Consequently,
\begin{align}\begin{split}\nonumber
q_k(x)>0 \to \frac{f_k(x)}{q_k(x)}=a_k\in[0,1], \quad 
q_k(x)=0 \to f_k(x)=0.
\end{split}\end{align}
Define
\begin{align}\begin{split}\nonumber
\phi_k(x)
\coloneqq\begin{cases}
\dfrac{f_k(x)}{q_k(x)},&q_k(x)>0,\\[7pt]
\dfrac{1-\displaystyle\sum_{j:q_j(x)>0}f_j(x)/q_j(x)}
{\bigl|\{j:q_j(x)=0\}\bigr|},&q_k(x)=0.
\end{cases}
\end{split}\end{align}
The denominator counts the zero-ratio coordinates, and the second case distributes the remaining probability equally among them.
Its numerator satisfies
\begin{align}\begin{split}\nonumber
1-\sum_{j:q_j(x)>0}\frac{f_j(x)}{q_j(x)}
=1-\sum_{j:q_j(x)>0}a_j
=\sum_{j:q_j(x)=0}a_j
\geq0.
\end{split}\end{align}
If every $q_k(x)$ is positive, the first case applies to all coordinates and $\sum_k\phi_k(x)=\sum_ka_k=1$.
If at least one ratio is zero, the second case has a positive denominator and
\begin{align}\begin{split}\nonumber
\sum_{k=1}^K\phi_k(x)
=\sum_{k:q_k(x)>0}\frac{f_k(x)}{q_k(x)}
+1-\sum_{k:q_k(x)>0}\frac{f_k(x)}{q_k(x)}
=1.
\end{split}\end{align}
All coordinates are nonnegative, so $\phi(x)\in\Delta^{K-1}$.
Set $\phi(x)=(1,0,\ldots,0)$ on the exceptional null set where $f(x)\notin F(x)$.
The ratios on $\{q_k>0\}$, the finite sums, and the case events are measurable, so the resulting $\phi$ is measurable and is an allocation rule as specified in Definition \ref{def:allocation-rule}.
For every $k$, the construction gives $q_k\phi_k=f_k$, including the zero-ratio case.
Equations \eqref{eq:selection-vector-integral} and \eqref{eq:weighted-allocation-selection} imply
\begin{align}\begin{split}\nonumber
\int_{\mathcal X}f(x)\dd P(x)
=\E_P[f_\phi(X)]
=T(\phi)\in\Gamma.
\end{split}\end{align}
Since $f$ was arbitrary, this proves $\int_{\mathcal X}F(x)\dd P(x)\subseteq\Gamma$.

\emph{Compactness.}
For each $x$, the set $F(x)$ is the image of $\Delta^{K-1}$ under the continuous linear map $a\mapsto(q_k(x)a_k)_{k=1}^K$.
It is therefore nonempty, compact, and convex.
The rational points in $\Delta^{K-1}$ form a countable dense subset, and the image of each fixed rational point is measurable in $x$.
For every open $O\subseteq\R^K$, continuity in $a$ gives
\begin{align}\begin{split}\nonumber
&\{x\in\mathcal X:F(x)\cap O\neq\varnothing\}\\
&\quad=\bigcup_{a\in\Delta^{K-1}\cap\mathbb Q^K}
\{x\in\mathcal X:(q_1(x)a_1,\ldots,q_K(x)a_K)\in O\}.
\end{split}\end{align}
Every event in this countable union is measurable, which verifies measurability of $F$.
For $a\in\Delta^{K-1}$,
\begin{align}\begin{split}\nonumber
\lVert(q_k(x)a_k)_{k=1}^K\rVert
\leq\sum_{k=1}^Kq_k(x)a_k
\leq\max_{1\leq k\leq K}q_k(x).
\end{split}\end{align}
Hence \eqref{eq:q-properties} yields
\begin{align}\begin{split}\nonumber
\E_P\!\left[\sup_{y\in F(X)}\lVert y\rVert\right]
\leq\E_P\!\left[\max_{1\leq k\leq K}q_k(X)\right]
\leq\sum_{k=1}^K\E_P[q_k(X)]
=K.
\end{split}\end{align}
Theorem \ref{thm:aumann-compactness} applies and gives compactness and convexity of $\Gamma$, together with the last equality in \eqref{eq:tpr-aumann-representation}.

\emph{Convexity and constant rules.}
For allocation rules $\phi,\psi$ and $\lambda\in[0,1]$, the function $\lambda\phi+(1-\lambda)\psi$ is measurable and has nonnegative coordinates summing to one.
For each $k$,
\begin{align}\begin{split}\nonumber
T_k\bigl(\lambda\phi+(1-\lambda)\psi\bigr)
&=\E_P\!\left[q_k(X)\bigl(\lambda\phi_k(X)+(1-\lambda)\psi_k(X)\bigr)\right]\\
&=\lambda T_k(\phi)+(1-\lambda)T_k(\psi).
\end{split}\end{align}
Thus, $\lambda T(\phi)+(1-\lambda)T(\psi)\in\Gamma$, proving convexity.
For any $a\in\Delta^{K-1}$, the constant rule $\phi(x)=a$ gives
\begin{align}\begin{split}\nonumber
T_k(\phi)=\E_P[q_k(X)a_k]=a_k,
\quad k=1,\ldots,K,
\end{split}\end{align}
so $\Delta^{K-1}\subseteq\Gamma$.

\emph{Existence of an optimal allocation rule.}
For each $t\in\Lambda$, the feasible true-positive rate vectors form the nonempty compact set $\{v\in\Gamma:v_k\geq t_k,\ k=2,\ldots,K\}$.
The continuous map $v\mapsto v_1$ has a maximizing point in this set.
By Definition \ref{def:true-positive-rate-region}, a maximizing vector equals $T(\phi)$ for some allocation rule $\phi$, which is therefore an optimal allocation rule with value $\beta_1(t)$.
\end{proof}


Proposition \ref{prop:tpr-region} shows that an optimal allocation rule $\phi$ exists for each $t\in\Lambda$, so $\beta_1(t)$ in \eqref{eq:constrained-frontier} and the upper ROC frontier are well defined.
We can also extend the Aumann representation of the true-positive-rate region $\Gamma$ to the ROC body $\cH$, whose volume we aim to compute.
The characterization in \eqref{eq:roc-body-feasibility} gives an equivalent description of $\cH$ in terms of nonnegative measurable functions $\alpha$.
\begin{align}\begin{split}\label{eq:roc-body-alpha-representation}
\cH
=\left\{
\bigl(\E_P[q_k(X)\alpha_k(X)]\bigr)_{k=1}^K:
\alpha:\mathcal X\to\R_+^K\text{ is measurable},
\displaystyle\sum_{k=1}^K\alpha_k(X)\leq1
\right\}.
\end{split}\end{align}

To verify this representation, take any $u\in\cH$ and use \eqref{eq:roc-body-feasibility} to choose an allocation rule $\phi$ with $u\leq T(\phi)$.
Setting 
$\alpha_k(x)=u_k\phi_k(x)/T_k(\phi)$ when $T_k(\phi)>0$, and $\alpha_k(x)=0$ otherwise, gives $0\leq\alpha_k\leq\phi_k$ and $\E_P[q_k(X)\alpha_k(X)]=u_k$.
Conversely, any nonnegative measurable $\alpha$ with $\sum_k\alpha_k(X)\leq1$ can be extended to an allocation rule $\phi\geq\alpha$ by assigning the unused probability to class 1.
The vector $(\E_P[q_k(X)\alpha_k(X)])_{k=1}^K$ lies below $T(\phi)$ and hence belongs to $\cH$.

\begin{proof}[Derivation of \eqref{eq:roc-body-alpha-representation}]
Take a measurable $\alpha:\mathcal X\to\R_+^K$ with $\sum_k\alpha_k(X)\leq1$ almost surely, and set $u_k=\E_P[q_k(X)\alpha_k(X)]$.
After setting $\alpha=0$ on its exceptional null set, we may assume $\alpha_k(x)\geq0$ and $\sum_k\alpha_k(x)\leq1$ for every $x$.
Assign the unused probability to class 1,
\begin{align}\begin{split}\nonumber
\phi_1(x)=\alpha_1(x)+1-\sum_{j=1}^K\alpha_j(x),
\quad
\phi_k(x)=\alpha_k(x),\quad k=2,\ldots,K.
\end{split}\end{align}
These measurable functions satisfy $\phi_k(x)\geq\alpha_k(x)\geq0$ and $\sum_k\phi_k(x)=1$.
Thus, $\phi$ is an allocation rule, and \eqref{eq:tp-rate-representation} gives
\begin{align}\begin{split}\nonumber
T_1(\phi)
&=u_1+\E_P\!\left[q_1(X)\left(1-\sum_{j=1}^K\alpha_j(X)\right)\right]
\geq u_1,\\
T_k(\phi)
&=\E_P[q_k(X)\alpha_k(X)]
=u_k,\quad k=2,\ldots,K.
\end{split}\end{align}
Equation \eqref{eq:roc-body-feasibility} implies $u\in\cH$.

Conversely, take $u\in\cH$ and choose $\phi$ with $0\leq u_k\leq T_k(\phi)$ for every $k$.
Set
\begin{align}\begin{split}\nonumber
\alpha_k(x)
\coloneqq\begin{cases}
\dfrac{u_k}{T_k(\phi)}\phi_k(x),&T_k(\phi)>0,\\[5pt]
0,&T_k(\phi)=0.
\end{cases}
\end{split}\end{align}
Since $u_k=0$ whenever $T_k(\phi)=0$,
\begin{align}
0\leq\alpha_k(x)&\leq\phi_k(x),\notag\\
\sum_{k=1}^K\alpha_k(x)
&\leq\sum_{k=1}^K\phi_k(x)=1,\notag\\
\E_P[q_k(X)\alpha_k(X)]
&=\begin{cases}
\dfrac{u_k}{T_k(\phi)}\E_P[q_k(X)\phi_k(X)]=u_k,&T_k(\phi)>0,\\[5pt]
0=u_k,&T_k(\phi)=0.\notag
\end{cases}
\end{align}
This proves \eqref{eq:roc-body-alpha-representation}.
\end{proof}


To construct an Aumann representation of the ROC body $\cH$ from \eqref{eq:roc-body-alpha-representation}, we first define coordinate simplices.
Let $e_k$ denote the $k$-th coordinate vector in $\R^K$ and let $\conv$ denote convex hull.

\begin{definition}[Coordinate simplex]
\label{def:coordinate-simplex}
A \emph{coordinate simplex} in $\R^K$ is a set of the form $D(a)=D(a_1,\ldots,a_K)\coloneqq\conv\{0,a_1e_1,\ldots,a_Ke_K\}$, where $a=(a_1,\ldots,a_K)\in\R_+^K$.
Equivalently,
\begin{align}\begin{split}\label{eq:coordinate-simplex-representation}
D(a)=\left\{\sum_{k=1}^K a_k\alpha_ke_k:
\alpha\in\R_+^K,\ \sum_{k=1}^K\alpha_k\leq1\right\}.
\end{split}\end{align}
\end{definition}

For every $a\in\R_+^K$, we have $D(a)=\operatorname{diag}(a_1,\ldots,a_K)D(1,\ldots,1)$, and hence
\begin{align}
\Vol_K(D(a))
&=\frac{1}{K!}\prod_{k=1}^K a_k.
\label{eq:coordinate-simplex-volume}
\end{align}
In particular, if $a_k=0$ for some $k$, then $\Vol_K(D(a))=0$.

Taking $a_k=q_k(x)$ in Definition \ref{def:coordinate-simplex} gives
\begin{align}\begin{split}\label{eq:random-coordinate-simplex}
D(x)\coloneqq D(q_1(x),\ldots,q_K(x)),
\quad
D(X)=\conv\{0,q_1(X)e_1,\ldots,q_K(X)e_K\}.
\end{split}\end{align}
Thus the coordinate simplex $D(x)$ equals $\conv(\{0\}\cup F(x))$, with the origin accounting for unused allocation probability.
Lemma \ref{lem:random-simplex-regularity} verifies the conditions needed to apply Theorem \ref{thm:aumann-compactness} to the set-valued function $x\mapsto D(x)$.

\begin{lemma}
\label{lem:random-simplex-regularity}
The set-valued function $x\mapsto D(x)$ in \eqref{eq:random-coordinate-simplex} is measurable and integrably bounded, with nonempty compact convex values.
It satisfies
\begin{align}\begin{split}\label{eq:random-simplex-integrability}
\sup_{y\in D(x)}\lVert y\rVert=\max_{1\leq k\leq K}q_k(x),
\quad
\E_P\!\left[\sup_{y\in D(X)}\lVert y\rVert\right]
\leq\sum_{k=1}^K\E_P[q_k(X)]=K.
\end{split}\end{align}
The simplices $D(x)$ are all contained in the same fixed simplex,
\begin{align}\begin{split}\label{eq:random-simplex-bound}
D(x)\subseteq D(\rho_1^{-1},\ldots,\rho_K^{-1}).
\end{split}\end{align}
Its Aumann integral is compact and convex, and
\begin{align}\begin{split}\label{eq:aumann-coordinate-integral}
\int_{\mathcal X}D(x)\dd P(x)
=\int_{\mathcal X}\conv\{0,q_1(x)e_1,\ldots,q_K(x)e_K\}\dd P(x)
=\E_A[D(X)].
\end{split}\end{align}
\end{lemma}

\begin{proof}
Each $D(x)$ is nonempty, compact, and convex as the convex hull of finitely many points.
For every open $O\subseteq\R^K$, continuity in the coefficients $\alpha$ gives
\begin{align}\begin{split}\nonumber
\{x:D(x)\cap O\neq\varnothing\}
=\bigcup_{\substack{\alpha\in\mathbb Q_+^K\\\sum_{k=1}^K\alpha_k\leq1}}
\left\{x:\sum_{k=1}^Kq_k(x)\alpha_ke_k\in O\right\}.
\end{split}\end{align}
The union is countable, and each set on the right is measurable because the functions $q_k$ are measurable.
Thus, $x\mapsto D(x)$ is a measurable set-valued function.

For every admissible $\alpha$ in \eqref{eq:coordinate-simplex-representation},
\begin{align}\begin{split}\nonumber
\left\|\sum_{k=1}^K q_k(x)\alpha_ke_k\right\|
& \leq\sum_{k=1}^Kq_k(x)\alpha_k\\
& \leq\left(\max_{1\leq k\leq K}q_k(x)\right)\sum_{k=1}^K\alpha_k 
\leq\max_{1\leq k\leq K}q_k(x).
\end{split}\end{align}
Choosing a vertex with the largest coordinate length gives equality in the norm bound.
Taking expectations and using $\max_k q_k(X)\leq\sum_k q_k(X)$ together with \eqref{eq:q-properties} proves \eqref{eq:random-simplex-integrability}.

To prove \eqref{eq:random-simplex-bound} using \eqref{eq:q-properties}, a point $\sum_kq_k(x)\alpha_ke_k\in D(x)$ can be written as
\begin{align}\begin{split}\nonumber
\sum_{k=1}^Kq_k(x)\alpha_ke_k
=\sum_{k=1}^K\rho_k^{-1}\bigl(\rho_kq_k(x)\alpha_k\bigr)e_k,
\quad
\sum_{k=1}^K\rho_kq_k(x)\alpha_k\leq\sum_{k=1}^K\alpha_k\leq1.
\end{split}\end{align}
Theorem \ref{thm:aumann-compactness} now applies to $x\mapsto D(x)$, giving compactness and convexity of its Aumann integral and the equality \eqref{eq:aumann-coordinate-integral}.
\end{proof}

To interpret \eqref{eq:roc-body-alpha-representation} as an Aumann integral, we identify its integrands with the measurable selections of $x\mapsto D(x)$.

\begin{lemma}
\label{lem:coordinate-simplex-selections}
The integrable measurable selections of $x\mapsto D(x)$ in \eqref{eq:random-coordinate-simplex} are the functions of the form
\begin{align}\begin{split}\label{eq:coordinate-simplex-selection}
f_\alpha(x)\coloneqq\sum_{k=1}^Kq_k(x)\alpha_k(x)e_k,
\end{split}\end{align}
where $\alpha:\mathcal X\to\R_+^K$ is measurable and $\sum_k\alpha_k(X)\leq1$.
\end{lemma}

\begin{proof}
For every admissible $\alpha$, the function $f_\alpha$ in \eqref{eq:coordinate-simplex-selection} is measurable because the functions $q_k$ and $\alpha_k$ are measurable.
Equation \eqref{eq:coordinate-simplex-representation} gives $f_\alpha(x)\in D(x)$, and Lemma \ref{lem:random-simplex-regularity} gives its integrability.
Conversely, let $f=(f_1,\ldots,f_K)$ be an integrable measurable selection of $D(x)$.
Set $f(x)=0$ on the $P$-null set where $f(x)\notin D(x)$, which preserves its integral.
Membership in $D(x)$ implies
\begin{align}\begin{split}\nonumber
f_k(x)=0\quad\text{if }q_k(x)=0,
\quad
f_k(x)\geq0,
\quad
\sum_{k:q_k(x)>0}\frac{f_k(x)}{q_k(x)}\leq1.
\end{split}\end{align}
Define the measurable functions
\begin{align}\begin{split}\nonumber
\alpha_k(x)
\coloneqq\begin{cases}
f_k(x)/q_k(x),&q_k(x)>0,\\
0,&q_k(x)=0.
\end{cases}
\end{split}\end{align}
Then
\begin{align}\begin{split}\nonumber
\sum_{k=1}^K\alpha_k(x)
=\sum_{k:q_k(x)>0}\frac{f_k(x)}{q_k(x)}\leq1,
\quad
q_k(x)\alpha_k(x)=f_k(x),
\end{split}\end{align}
and hence $f=f_\alpha$ almost surely.

For nonnegative measurable functions $b_k$ satisfying $\E_P[\max_{1\leq k\leq K}b_k(X)]<\infty$, the measurability argument in the proof of Lemma \ref{lem:random-simplex-regularity} also applies to $x\mapsto D(b_1(x),\ldots,b_K(x))$.
These sets are nonempty, compact, and convex, and
\begin{align}\begin{split}\nonumber
\sup_{y\in D(b_1(x),\ldots,b_K(x))}\lVert y\rVert=\max_{1\leq k\leq K}b_k(x).
\end{split}\end{align}
The assumed integrability of the maximum makes this set-valued function integrably bounded, so Theorem \ref{thm:aumann-compactness} gives a compact Aumann integral.
The same construction, with $b_k$ in place of $q_k$, gives the selection characterization in this setting as well.
\end{proof}

Combining \eqref{eq:roc-body-alpha-representation} with Lemma \ref{lem:coordinate-simplex-selections} gives the following Aumann representation of the ROC body $\cH$.
The compactness and convexity of $\cH$ then follow from Lemma \ref{lem:random-simplex-regularity}.

\begin{proposition}
\label{prop:aumann-roc}
The ROC body $\cH$ satisfies
\begin{align}\begin{split}\label{eq:roc-aumann-representation}
\cH=\int_{\mathcal X}D(x)\dd P(x)=\E_A[D(X)].
\end{split}\end{align}
The set $\cH$ is compact and convex with nonempty interior in $\R^K$, and satisfies
\begin{align}\begin{split}\label{eq:roc-body-containments}
D(1,\ldots,1)\subseteq\cH\subseteq[0,1]^K.
\end{split}\end{align}
\end{proposition}

\begin{proof}
Lemma \ref{lem:coordinate-simplex-selections} identifies the integrable measurable selections of $x\mapsto D(x)$ with the functions $f_\alpha$.
Integration is coordinatewise, so
\begin{align}\begin{split}\nonumber
\int_{\mathcal X}f_\alpha(x)\dd P(x)
=\bigl(\E_P[q_k(X)\alpha_k(X)]\bigr)_{k=1}^K.
\end{split}\end{align}
Equation \eqref{eq:roc-body-alpha-representation} therefore gives $\cH=\int_{\mathcal X}D(x)\dd P(x)$.
Lemma \ref{lem:random-simplex-regularity} identifies this integral with $\E_A[D(X)]$, proving \eqref{eq:roc-aumann-representation}.

\emph{Geometry of $\cH$.}
Compactness follows from \eqref{eq:roc-aumann-representation} and Lemma \ref{lem:random-simplex-regularity}.
For $u,v\in\cH$, choose allocation rules $\phi,\psi$ such that $u\leq T(\phi)$ and $v\leq T(\psi)$.
For every $\lambda\in[0,1]$, Proposition \ref{prop:tpr-region} gives
\begin{align}\begin{split}\nonumber
\lambda u+(1-\lambda)v
\leq\lambda T(\phi)+(1-\lambda)T(\psi)
=T\bigl(\lambda\phi+(1-\lambda)\psi\bigr),
\end{split}\end{align}
so $\lambda u+(1-\lambda)v\in\cH$.
For any $u\in D(1,\ldots,1)$, the constant choice $\alpha_k(x)=u_k$ satisfies
\begin{align}\begin{split}\nonumber
\E_P[q_k(X)\alpha_k(X)]=u_k\E_P[q_k(X)]=u_k,
\end{split}\end{align}
which proves the lower containment in \eqref{eq:roc-body-containments}.
The upper containment follows from $T_k(\phi)\leq1$ and \eqref{eq:roc-body-feasibility}.
In particular, $\cH$ contains the nonempty open set $\{u\in\R^K:u_k>0\text{ for all }k,\ \sum_ku_k<1\}$.
\end{proof}

The ROC body $\cH$ is therefore a \emph{convex body}, meaning a compact convex set with nonempty interior in $\R^K$.
The containments in \eqref{eq:roc-body-containments} and the simplex volume in \eqref{eq:coordinate-simplex-volume} give
\begin{align}\begin{split}\label{eq:geometric-roc-volume-bounds}
\frac1{K!}=\Vol_K(D(1,\ldots,1))\leq\Vol_K(\cH)\leq\Vol_K([0,1]^K)=1.
\end{split}\end{align}

The uninformative and perfectly informative cases give equality at the lower and upper bounds, respectively.
In the uninformative case $P_1=\cdots=P_K$, the class density ratios satisfy $q_k=1$ for every $k$, so $D(X)=D(1,\ldots,1)$.
Proposition \ref{prop:aumann-roc} then gives $\cH=D(1,\ldots,1)$ and $\Vol_K(\cH)=1/K!$.
In the perfectly informative case, an allocation rule $\phi$ satisfies $T(\phi)=(1,\ldots,1)$.
The characterization in \eqref{eq:roc-body-feasibility} then gives $\cH=[0,1]^K$ and $\Vol_K(\cH)=1$.



\subsection{Minkowski sums, mixed volume and Hausdorff distance}
\label{subsec:mixed-volume-tools}

In this part, we introduce Minkowski sums and mixed volume, together with the continuity properties needed to approximate volumes of convex sets.
As in Definition \ref{def:roc-body}, $\Vol_K(A)$ denotes the $K$-dimensional Lebesgue measure of $A$ \citep[Conventions and notation, p. xxi]{Schneider2014}.
Every compact set $A\subseteq\R^K$ is closed and bounded, hence Borel measurable with finite volume.
If $A$ lies in an affine subspace of dimension less than $K$, then $\Vol_K(A)=0$.


For nonempty compact convex sets $A,B\subseteq\R^K$ and $\lambda\geq0$, their \emph{Minkowski sum} and \emph{scaling} are
\begin{align}\begin{split}\nonumber
A+B\coloneqq\{a+b:a\in A,\ b\in B\},
\quad
\lambda A\coloneqq\{\lambda a:a\in A\}.
\end{split}\end{align}
For nonempty compact convex sets $A_1,\ldots,A_M$ and weights $\lambda_1,\ldots,\lambda_M\geq0$, the weighted Minkowski sum is
\begin{align}\begin{split}\nonumber
\sum_{m=1}^M\lambda_mA_m
=\left\{\sum_{m=1}^M\lambda_ma_m:a_m\in A_m\text{ for every }m\right\}.
\end{split}\end{align}
The weighted Minkowski sum $\sum_{m=1}^M\lambda_mA_m$ is nonempty, compact, and convex.

\begin{definition}[Mixed volume (equation (5.20), p. 280 of \cite{Schneider2014})]
\label{def:mixed-volume}
For nonempty compact convex sets $A_1,\ldots,A_K\subseteq\R^K$, their \emph{mixed volume} $V(A_1,\ldots,A_K)$ is defined by the \emph{polarization formula}.
\begin{align}\begin{split}\label{eq:polytopal-mixed-volume}
V(A_1,\ldots,A_K)
\coloneqq\frac{1}{K!}
\sum_{\varnothing\neq I\subseteq\{1,\ldots,K\}}
(-1)^{K-\lvert I\rvert}
\Vol_K\!\left(\sum_{r\in I}A_r\right).
\end{split}\end{align}
\end{definition}
The subset $I$ selects positions in the ordered list $(A_1,\ldots,A_K)$, and $\lvert I\rvert$ counts the selected positions.
For example, when $K=3$, $I=\{1,3\}$ selects $A_1+A_3$, with $\lvert I\rvert=2$.
The same position labels apply when sets repeat.
If $A_1=A_2=A$ and $A_3=B$, then $I=\{1,2,3\}$ has $\lvert I\rvert=3$ and selects $A_1+A_2+A_3=A+A+B=2A+B$.

The mixed volume $V$ in Definition \ref{def:mixed-volume} determines the volume of weighted Minkowski sums through the following theorem of \citet[Theorem 5.1.7, equation (5.17), p. 280]{Schneider2014}.
\begin{theorem}[Minkowski volume theorem]
\label{thm:minkowski-volume}
Let $A_1,\ldots,A_M\subseteq\R^K$ be nonempty compact convex sets, where $M\geq1$.
For all nonnegative weights $\lambda_1,\ldots,\lambda_M$,
\begin{align}\begin{split}\label{eq:mixed-volume-polynomial}
\Vol_K\!\left(\sum_{m=1}^M\lambda_mA_m\right)
=\sum_{i_1,\ldots,i_K=1}^M
\lambda_{i_1}\cdots\lambda_{i_K}
V(A_{i_1},\ldots,A_{i_K}).
\end{split}\end{align}
The indices $(i_1,\ldots,i_K)$ describe $K$ successive choices from the $M$ sets, with repetitions allowed.
The sum therefore contains $M^K$ terms, one for each ordered tuple $(i_1,\ldots,i_K)\in\{1,\ldots,M\}^K$.
\end{theorem}

We illustrate Definition \ref{def:mixed-volume} and Theorem \ref{thm:minkowski-volume} with \emph{polytopes}, which are convex hulls of finitely many points and include the coordinate simplices $D(x)$ defined in \eqref{eq:random-coordinate-simplex}.

\begin{example}[Weighted Minkowski sum]
\label{ex:minkowski-area}
\leavevmode\par
Consider the two-dimensional coordinate simplices
\begin{align}\begin{split}\nonumber
A_1=\conv\{0,\tfrac32e_1,\tfrac12e_2\},
\quad
A_2=\conv\{0,\tfrac12e_1,\tfrac32e_2\}.
\end{split}\end{align}
For $\lambda_1,\lambda_2\geq0$, their weighted Minkowski sum decomposes into translated copies of $\lambda_1A_1$ and $\lambda_2A_2$, together with a rectangle of side lengths $3\lambda_1/2$ and $3\lambda_2/2$, as shown in Figure \ref{fig:mixed-area-decomposition}.
Adding the three areas gives
\begin{align}\begin{split}\nonumber
\Vol_2(\lambda_1A_1+\lambda_2A_2)
=\frac38\lambda_1^2+\frac94\lambda_1\lambda_2+\frac38\lambda_2^2.
\end{split}\end{align}
The rectangle uses one coordinate length from each simplex and accounts for the term involving both weights.
Theorem \ref{thm:minkowski-volume} expresses this quadratic polynomial in terms of mixed volumes and gives the same form for any two nonempty polytopes $A_1,A_2\subseteq\R^2$, including segments and points.
\begin{align}\begin{split}\nonumber
\Vol_2(\lambda_1A_1+\lambda_2A_2)
&=\sum_{i,j=1}^2\lambda_i\lambda_jV(A_i,A_j)\\
&=\lambda_1^2V(A_1,A_1)
+2\lambda_1\lambda_2V(A_1,A_2)
+\lambda_2^2V(A_2,A_2).
\end{split}\end{align}
For these two-dimensional arguments, Definition \ref{def:mixed-volume} gives
\begin{align}\begin{split}\nonumber
V(A_1,A_1)&=\Vol_2(A_1),
\quad V(A_2,A_2)=\Vol_2(A_2),\\
V(A_1,A_2)&=V(A_2,A_1)
=\frac{\Vol_2(A_1+A_2)-\Vol_2(A_1)-\Vol_2(A_2)}{2}.
\end{split}\end{align}
In Figure \ref{fig:mixed-area-decomposition}(b), the blue and orange triangles have areas $\lambda_1^2V(A_1,A_1)$ and $\lambda_2^2V(A_2,A_2)$, while the gray rectangle has area $2\lambda_1\lambda_2V(A_1,A_2)$.
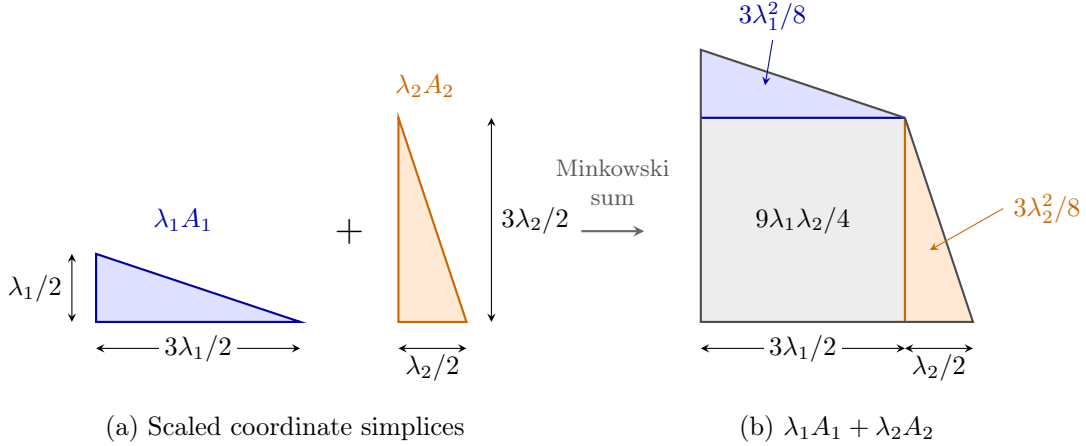
\begin{figure}[htbp]
\centering
\begin{tikzpicture}[x=1cm,y=1cm,>=stealth,font=\small]
\begin{scope}
  \filldraw[fill=blue!12,draw=blue!65!black,line width=0.8pt]
    (0,0) -- (2.7,0) -- (0,0.9) -- cycle;
  \node[text=blue!65!black] at (1.15,1.35) {$\lambda_1A_1$};
  \draw[<->] (0,-0.35) -- (2.7,-0.35)
    node[midway,fill=white,inner sep=2pt] {$3\lambda_1/2$};
  \draw[<->] (-0.32,0) -- (-0.32,0.9)
    node[midway,left] {$\lambda_1/2$};
  \node[font=\Large] at (3.35,1.20) {$+$};
  \filldraw[fill=orange!17,draw=orange!80!black,line width=0.8pt]
    (4,0) -- (4.9,0) -- (4,2.7) -- cycle;
  \node[text=orange!80!black] at (4.35,3.13) {$\lambda_2A_2$};
  \draw[<->] (4,-0.35) -- (4.9,-0.35)
    node[midway,below,inner sep=2pt] {$\lambda_2/2$};
  \draw[<->] (5.23,0) -- (5.23,2.7)
    node[midway,right] {$3\lambda_2/2$};
  \node[anchor=north,align=center] at (2.5,-1.10) {(a) Scaled coordinate simplices};
\end{scope}
\draw[->,thick,black!60] (6.42,1.20) -- (7.25,1.20);
\node[align=center,font=\footnotesize,text=black!65] at (6.835,1.85) {Minkowski\\sum};
\begin{scope}[xshift=8.0cm]
  \coordinate (mixed-area-corner) at (2.7,2.7);
  \fill[black!7] (0,0) rectangle (mixed-area-corner);
  \fill[blue!12] (0,2.7) -- (mixed-area-corner) -- (0,3.6) -- cycle;
  \fill[orange!17] (2.7,0) -- (3.6,0) -- (mixed-area-corner) -- cycle;
  \draw[blue!65!black,line width=0.8pt,line cap=butt]
    (0,2.7) -- (mixed-area-corner);
  \draw[orange!80!black,line width=0.8pt,line cap=butt]
    (2.7,0) -- (mixed-area-corner);
  \draw[black!70,line width=0.8pt]
    (0,0) -- (3.6,0) -- (mixed-area-corner) -- (0,3.6) -- cycle;
  \node at (1.35,1.35) {$9\lambda_1\lambda_2/4$};
  \node[text=blue!65!black] (firstarea) at (0.95,4.05) {$3\lambda_1^2/8$};
  \draw[->,blue!65!black,thin] (firstarea.south) -- (0.85,3.13);
  \node[text=orange!80!black,right] (secondarea) at (4.00,1.50) {$3\lambda_2^2/8$};
  \draw[->,orange!80!black,thin] (secondarea.west) -- (3.03,0.95);
  \draw[<->] (0,-0.35) -- (2.7,-0.35)
    node[midway,fill=white,inner sep=2pt] {$3\lambda_1/2$};
  \draw[<->] (2.7,-0.35) -- (3.6,-0.35)
    node[midway,below,inner sep=2pt] {$\lambda_2/2$};
  \node[anchor=north,align=center] at (1.80,-1.10) {(b) $\lambda_1A_1+\lambda_2A_2$};
\end{scope}
\end{tikzpicture}
\caption{Mixed-volume contributions to the area of a Minkowski sum}
\label{fig:mixed-area-decomposition}
\begin{minipage}{0.94\linewidth}
\footnotesize
\textit{Note.} Here $A_1=\conv\{0,\tfrac32e_1,\tfrac12e_2\}$ and $A_2=\conv\{0,\tfrac12e_1,\tfrac32e_2\}$.
The blue and orange regions in panel (b) are translated copies of the scaled simplices in panel (a).
The mixed volumes are $V(A_1,A_1)=V(A_2,A_2)=3/8$ and $V(A_1,A_2)=V(A_2,A_1)=9/8$, so the gray rectangle has area $2\lambda_1\lambda_2V(A_1,A_2)$.
The drawing uses $\lambda_1=\lambda_2=1$, while the labels apply to arbitrary positive weights.
\end{minipage}
\end{figure}
\end{example}

We next consider the continuity of volume $\Vol_K$ and mixed volume $V$ with respect to the \emph{Hausdorff distance}.
\begin{definition}[Hausdorff distance]
\label{def:hausdorff-distance}
For nonempty compact sets $A,B\subseteq\R^K$, their \emph{Hausdorff distance} is
\begin{align}\begin{split}\nonumber
d_H(A,B)
\coloneqq\max\left\{
\sup_{a\in A}\inf_{b\in B}\lVert a-b\rVert,
\sup_{b\in B}\inf_{a\in A}\lVert b-a\rVert
\right\}.
\end{split}\end{align}
\end{definition}
The Hausdorff distance $d_H(A,B)$ bounds the distance from every point of either set to its nearest point in the other set.
For nonempty compact sets $A,B,C\subseteq\R^K$, $d_H$ satisfies the metric properties \citep[Section 1.8, p. 61]{Schneider2014}
\begin{align}\begin{split}\nonumber
d_H(A,B)=d_H(B,A)\geq0,
d_H(A,B)=0\ \Leftrightarrow\ A=B,
d_H(A,C)\leq d_H(A,B)+d_H(B,C).
\end{split}\end{align}

For compact convex sets, Hausdorff convergence implies convergence of volume $\Vol_K$.
\begin{proposition}[Theorem 1.8.20 of \cite{Schneider2014}]
\label{prop:hausdorff-continuity}
For nonempty compact convex sets $A^{(n)},A\subseteq\R^K$, Hausdorff convergence $d_H(A^{(n)},A)\to0$ implies convergence of volume.
\begin{align}\begin{split}\nonumber
\Vol_K(A^{(n)})\to\Vol_K(A).
\end{split}\end{align}
\end{proposition}

For nonempty compact sets $A_m,B_m\subseteq\R^K$, $m=1,\ldots,M$, the Hausdorff distance between their weighted Minkowski sums satisfies the bound \citep[Section 1.8, p. 64]{Schneider2014}.
\begin{align}\begin{split}\label{eq:hausdorff-minkowski}
d_H\!\left(\sum_{m=1}^M\lambda_mA_m,\sum_{m=1}^M\lambda_mB_m\right)
\leq\sum_{m=1}^M\lambda_m d_H(A_m,B_m),
\quad \lambda_m\geq0.
\end{split}\end{align}
This bound shows that, for fixed weights $\lambda_m$, the two Minkowski sums become arbitrarily close as each $B_m$ approaches $A_m$ in Hausdorff distance.

The polarization formula in Definition \ref{def:mixed-volume} expresses the mixed volume $V$ as a finite linear combination of volumes.
Combining \eqref{eq:hausdorff-minkowski} with Proposition \ref{prop:hausdorff-continuity} therefore gives continuity of $V$.
\begin{proposition}
\label{prop:mixed-volume-continuity}
For $r=1,\ldots,K$, let $A_r^{(n)}$ and $A_r$ be nonempty compact convex subsets of $\R^K$.
If $d_H(A_r^{(n)},A_r)\to0$ for every $r$, then
\begin{align}\begin{split}\label{eq:mixed-continuity}
V(A_1^{(n)},\ldots,A_K^{(n)})
\to V(A_1,\ldots,A_K).
\end{split}\end{align}
\end{proposition}

\begin{proof}
For every nonempty $I\subseteq\{1,\ldots,K\}$, the bound in \eqref{eq:hausdorff-minkowski} gives
\begin{align}\begin{split}\nonumber
d_H\!\left(\sum_{r\in I}A_r^{(n)},\sum_{r\in I}A_r\right)
\leq\sum_{r\in I}d_H(A_r^{(n)},A_r)
\to0.
\end{split}\end{align}
These sums are compact and convex, so Proposition \ref{prop:hausdorff-continuity} gives convergence of their volumes.
The polarization formula in Definition \ref{def:mixed-volume} is a finite linear combination of these volumes, so taking limits term by term gives \eqref{eq:mixed-continuity}.
\end{proof}

Polytopal approximation provides one application of Proposition \ref{prop:mixed-volume-continuity}.
Every nonempty compact convex set can be approximated by nonempty polytopes \citep[Theorem 1.8.16, p. 67]{Schneider2014}.
That is, for $A_1,\ldots,A_K$, we can choose nonempty polytopes $A_r^{(n)}$ satisfying $d_H(A_r^{(n)},A_r)\to0$ for every $r$.
Proposition \ref{prop:mixed-volume-continuity} ensures that $V(A_1^{(n)},\ldots,A_K^{(n)})$ converges to $V(A_1,\ldots,A_K)$.

In addition to continuity, the mixed-volume function $V$ has the following properties \citep[Section 5.1, pp. 280--282]{Schneider2014}.
All sets in these statements are nonempty compact convex subsets of $\R^K$.
\begin{enumerate}
\item \emph{Symmetry and nonnegativity.}
The function $V$ is symmetric in its arguments and satisfies
\begin{align}\begin{split}\label{eq:mixed-normalization}
V(A_1,\ldots,A_K)\geq0,
\quad V(A,\ldots,A)=\Vol_K(A).
\end{split}\end{align}
\item \emph{Minkowski multilinearity.}
For $\lambda,\mu\geq0$,
\begin{align}\begin{split}\nonumber
V(\lambda A+\mu B,A_2,\ldots,A_K)
=\lambda V(A,A_2,\ldots,A_K)+\mu V(B,A_2,\ldots,A_K).
\end{split}\end{align}
\item \emph{Monotonicity.}
If $A_r\subseteq B_r$ for $r=1,\ldots,K$,
\begin{align}\begin{split}\label{eq:mixed-monotonicity}
0\leq V(A_1,\ldots,A_K)\leq V(B_1,\ldots,B_K).
\end{split}\end{align}
\end{enumerate}


\subsection{Computing the volume of the ROC body}
\label{subsec:roc-volume-formula}
\label{subsec:aumann-mixed-volume}
\label{subsec:coordinate-assignment}

We apply the tools from subsection \ref{subsec:mixed-volume-tools} to the ROC body $\cH$.
Its Aumann representation in Proposition \ref{prop:aumann-roc} gives
\begin{align}\begin{split}\label{eq:volume-aumann-body}
\Vol_K(\cH)
=\Vol_K\!\left(\int_{\mathcal X}D(x)\dd P(x)\right)
=\Vol_K\bigl(\E_A[D(X)]\bigr).
\end{split}\end{align}
Although $D(x)$ has explicit vertices given in \eqref{eq:random-coordinate-simplex}, the set $\E_A[D(X)]$ of expectations of all its measurable selections generally has a more complicated shape.


Let $X_1,\ldots,X_K$ be independent draws from $P$, with joint distribution $P^K$.
We first consider the case in which the random coordinate simplex $D(X)$ defined in \eqref{eq:random-coordinate-simplex} takes only finitely many values.

\begin{lemma}
\label{lem:finite-coordinate-simplices}
Suppose $D(X)$ takes only the distinct values $D_1,\ldots,D_M$, each with positive probability.
The Aumann expectation $\E_A[D(X)]$ is the weighted Minkowski sum of $D_1,\ldots,D_M$, with their probabilities as weights.
\begin{align}\begin{split}\label{eq:finite-aumann-minkowski}
\E_A[D(X)]=\int_{\mathcal X}D(x)\dd P(x)=\sum_{m=1}^M \Pp(D(X)=D_m)D_m.
\end{split}\end{align}
Theorem \ref{thm:minkowski-volume} and the independence of $X_1,\ldots,X_K$ give
\begin{align}\begin{split}\nonumber
\Vol_K\bigl(\E_A[D(X)]\bigr)
&=\sum_{i_1,\ldots,i_K=1}^M
\left(\prod_{r=1}^K\Pp(D(X)=D_{i_r})\right)V(D_{i_1},\ldots,D_{i_K})\\
&=\sum_{i_1,\ldots,i_K=1}^M
\Pp\!\left(D(X_1)=D_{i_1},\ldots,D(X_K)=D_{i_K}\right)V(D_{i_1},\ldots,D_{i_K})\\
&=\E_{P^K}\!\left[V(D(X_1),\ldots,D(X_K))\right].
\end{split}\end{align}
\end{lemma}

\begin{proof}
Writing $D_m=D(a_{m1},\ldots,a_{mK})$, equality of coordinate simplices is equivalent to equality of their axis lengths,
\begin{align}\begin{split}\label{eq:finite-simplex-events}
\{x:D(x)=D_m\}
=\bigcap_{k=1}^K\{x:q_k(x)=a_{mk}\}.
\end{split}\end{align}
These events are measurable, disjoint, and cover a set of $P$-measure one.

Take an integrable measurable selection $f$ of $x\mapsto D(x)$.
Set $f(x)=0$ on the measurable $P$-null set where $f(x)\notin D(x)$, which preserves its integral.
For $m=1,\ldots,M$, define
\begin{align}\begin{split}\nonumber
y_m\coloneqq\E_P[f(X)\mid D(X)=D_m]
=\frac{\E_P[f(X)\mathbf 1_{\{D(X)=D_m\}}]}{\Pp(D(X)=D_m)}.
\end{split}\end{align}
The coordinate description \eqref{eq:coordinate-simplex-representation} verifies directly that $y_m\in D_m$.
Indeed, on $\{D(X)=D_m\}$,
$f_k(X)\geq0$,
$f_k(X)=0$
if $a_{mk}=0$, and 
$\sum_{k:a_{mk}>0}\frac{f_k(X)}{a_{mk}} \leq 1$.
Taking conditional expectations gives $(y_m)_k \geq 0$, $(y_m)_k=0$ if $a_{mk}=0$, and 
\begin{align}\begin{split}\nonumber
\sum_{k:a_{mk}>0}\frac{(y_m)_k}{a_{mk}}
&=\E_P\!\left[\left.\sum_{k:a_{mk}>0}\frac{f_k(X)}{a_{mk}}\right|D(X)=D_m\right]
\leq1.
\end{split}\end{align}
Thus, the law of total expectation yields
\begin{align}\begin{split}\nonumber
\E_P[f(X)]
&=\sum_{m=1}^M \Pp(D(X)=D_m)\E_P[f(X)\mid D(X)=D_m] 
\in\sum_{m=1}^M \Pp(D(X)=D_m)D_m,
\end{split}\end{align}
and hence
\begin{align}\begin{split}\nonumber
\int_{\mathcal X}D(x)\dd P(x)\subseteq\sum_{m=1}^M \Pp(D(X)=D_m)D_m.
\end{split}\end{align}

Conversely, take $u\in\sum_m \Pp(D(X)=D_m)D_m$, so $u=\sum_m \Pp(D(X)=D_m)y_m$ for some $y_m\in D_m$.
Define
\begin{align}\begin{split}\nonumber
f(x)\coloneqq\begin{cases}
y_m,&D(x)=D_m,\quad m=1,\ldots,M,\\
0,&D(x)\notin\{D_1,\ldots,D_M\}.
\end{cases}
\end{split}\end{align}
Equation \eqref{eq:finite-simplex-events} makes $f$ measurable, and its finite range makes it integrable.
For every $x$,
\begin{align}\begin{split}\nonumber
f(x)\in D(x),
\quad
\begin{cases}
f(x)=y_m\in D_m=D(x),&D(x)=D_m,\\
f(x)=0\in\conv\{0,q_1(x)e_1,\ldots,q_K(x)e_K\}=D(x),&D(x)\notin\{D_m\}_{m=1}^M.
\end{cases}
\end{split}\end{align}
Therefore,
\begin{align}\begin{split}\nonumber
u=\sum_{m=1}^M \Pp(D(X)=D_m)y_m=\E_P[f(X)]\in\int_{\mathcal X}D(x)\dd P(x),
\end{split}\end{align}
which proves the reverse inclusion.
The weighted Minkowski sum is compact, so Definition \ref{def:aumann-expectation} gives
\begin{align}\begin{split}\nonumber
\E_A[D(X)]
&=\int_{\mathcal X}D(x)\dd P(x) \\ 
&=\sum_{m=1}^M \Pp(D(X)=D_m)D_m\\
&=\left\{\sum_{m=1}^M \Pp(D(X)=D_m)y_m:y_m\in D_m,\ m=1,\ldots,M\right\}.
\end{split}\end{align}

Apply \eqref{eq:mixed-volume-polynomial} to the sum in \eqref{eq:finite-aumann-minkowski}, with $\lambda_m=\Pp(D(X)=D_m)$, to obtain
\begin{align}\begin{split}\nonumber
\Vol_K\bigl(\E_A[D(X)]\bigr)
&=\Vol_K\!\left(\sum_{m=1}^M \Pp(D(X)=D_m)D_m\right)\\
&=\sum_{i_1,\ldots,i_K=1}^M
\left(\prod_{r=1}^K\Pp(D(X)=D_{i_r})\right)V(D_{i_1},\ldots,D_{i_K}).
\end{split}\end{align}
For any ordered tuple $(i_1,\ldots,i_K)\in\{1,\ldots,M\}^K$, independence gives
\begin{align}\begin{split}\nonumber
\Pp(D(X_1)=D_{i_1},\ldots,D(X_K)=D_{i_K})
&=\prod_{r=1}^K\Pp(D(X_r)=D_{i_r})\\
&=\prod_{r=1}^K\Pp(D(X)=D_{i_r}).
\end{split}\end{align}
On that event,
\begin{align}\begin{split}\nonumber
V(D(X_1),\ldots,D(X_K))=V(D_{i_1},\ldots,D_{i_K}).
\end{split}\end{align}
Summing over the disjoint events, including tuples with repeated indices, yields
\begin{align}\begin{split}\nonumber
\E_{P^K}[V(D(X_1),\ldots,D(X_K))]
&=\sum_{i_1,\ldots,i_K=1}^M
V(D_{i_1},\ldots,D_{i_K})
\Pp(D(X_r)=D_{i_r}\text{ for all }r)\\
&=\sum_{i_1,\ldots,i_K=1}^M
\left(\prod_{r=1}^K\Pp(D(X)=D_{i_r})\right)V(D_{i_1},\ldots,D_{i_K})\\
&=\Vol_K\bigl(\E_A[D(X)]\bigr),
\end{split}\end{align}
where the last equality follows from Theorem \ref{thm:minkowski-volume}.
The argument uses only the nonnegativity and measurability of the coordinate lengths and the finitely many possible simplex values.
\end{proof}

We extend the calculation in Lemma \ref{lem:finite-coordinate-simplices} to the general case by approximating $D(X)$ with finite-valued random coordinate simplices in Hausdorff distance.
The continuity results in Propositions \ref{prop:hausdorff-continuity} and \ref{prop:mixed-volume-continuity}, together with the bound in \eqref{eq:random-simplex-bound}, give the corresponding volume limits.

\begin{lemma}
\label{lem:finite-simplex-approximation}
\label{lem:coordinate-volume-limits}
Consider the random coordinate simplex $D(X)$ defined in \eqref{eq:random-coordinate-simplex}.
There exists a sequence of measurable random coordinate simplices $D^{(n)}(X)$, $n\geq1$.
Each $D^{(n)}(X)$ takes only finitely many values, such that 
$d_H(D^{(n)}(x),D(x))\to 0$  and 
\begin{align}
\Vol_K\bigl(\E_A[D^{(n)}(X)]\bigr)
& \to\Vol_K\bigl(\E_A[D(X)]\bigr), \label{eq:aumann-volume-limit}\\
\E_{P^K}[V(D^{(n)}(X_1),\ldots,D^{(n)}(X_K))]
& \to\E_{P^K}[V(D(X_1),\ldots,D(X_K))]. \label{eq:expected-mixed-volume-limit}
\end{align}
\end{lemma}

\begin{proof}
For $n\geq1$, define
\begin{align}\begin{split}\label{eq:finite-simplex-approximation}
q_k^{(n)}(x)\coloneqq\frac{\lfloor nq_k(x)\rfloor}{n},
\quad
D^{(n)}(x)\coloneqq D(q_1^{(n)}(x),\ldots,q_K^{(n)}(x)).
\end{split}\end{align}
Equation \eqref{eq:q-properties} gives $0\leq nq_k(x)\leq n\rho_k^{-1}$, $0\leq\lfloor nq_k(x)\rfloor \leq\lfloor n\rho_k^{-1}\rfloor$, and 
\begin{align}\begin{split}\nonumber
q_k^{(n)}(x)&\in
\left\{\frac{j}{n}:j=0,1,\ldots,\lfloor n\rho_k^{-1}\rfloor\right\}.
\end{split}\end{align}
Thus, the vector $(q_1^{(n)}(X),\ldots,q_K^{(n)}(X))$ takes values in a set of cardinality at most
$\prod_{k=1}^K\bigl(\lfloor n\rho_k^{-1}\rfloor+1\bigr)<\infty$.
The simplex $D^{(n)}(X)$ is determined by this vector, so for each fixed $n$ it equals one of finitely many possible coordinate simplices.
The floor inequalities also imply
\begin{align}\begin{split}\nonumber
\lfloor nq_k(x)\rfloor&\leq nq_k(x)<\lfloor nq_k(x)\rfloor+1,\\
0\leq q_k(x)-q_k^{(n)}(x)
&=\frac{nq_k(x)-\lfloor nq_k(x)\rfloor}{n}<\frac1n.
\end{split}\end{align}

The inequalities $0\leq q_k^{(n)}\leq q_k$ also give $D^{(n)}(x)\subseteq D(x)$.

For each selection, set the selection and its representing coefficients to zero on the measurable $P$-null set where the selection conditions fail.
Using Lemma \ref{lem:coordinate-simplex-selections}, pair the selection $f_\alpha$ in \eqref{eq:coordinate-simplex-selection} with an approximating selection using the same measurable coefficients $\alpha_k\geq0$, where $\sum_k\alpha_k\leq1$,
\begin{align}\begin{split}\nonumber
f_\alpha(x) = \sum_{k=1}^Kq_k(x)\alpha_k(x)e_k\in D(x),
f_\alpha^{(n)}(x) = \sum_{k=1}^Kq_k^{(n)}(x)\alpha_k(x)e_k\in D^{(n)}(x).
\end{split}\end{align}
By the triangle inequality and $\lVert e_k\rVert=1$,
\begin{align}\begin{split}\label{eq:paired-selection-bound}
\lVert f_\alpha(x)-f_\alpha^{(n)}(x)\rVert
&=\left\|\sum_{k=1}^K(q_k(x)-q_k^{(n)}(x))\alpha_k(x)e_k\right\|\\
&\leq\sum_{k=1}^K(q_k(x)-q_k^{(n)}(x))\alpha_k(x) 
\leq\frac1n\sum_{k=1}^K\alpha_k(x)
\leq\frac1n.
\end{split}\end{align}
Every point of either simplex admits such a coefficient representation, so
\begin{align}\begin{split}\label{eq:simplex-hausdorff-approximation}
d_H(D^{(n)}(x),D(x))\leq\frac1n.
\end{split}\end{align}

The construction in the proof of Lemma \ref{lem:coordinate-simplex-selections} also applies with $q_k^{(n)}$ in place of $q_k$.
These functions are measurable and satisfy $0\leq q_k^{(n)}\leq q_k$, so $x\mapsto D^{(n)}(x)$ is measurable and integrably bounded, with nonempty compact convex values.
Theorem \ref{thm:aumann-compactness} therefore identifies its Aumann expectation with its Aumann integral.
In particular, any measurable selection $f$ of $D^{(n)}(x)$ is represented by
\begin{align}\begin{split}\nonumber
\alpha_k(x)
=\begin{cases}
f_k(x)/q_k^{(n)}(x),&q_k^{(n)}(x)>0,\\
0,&q_k^{(n)}(x)=0,
\end{cases}
\quad
\sum_{k=1}^K\alpha_k(x)\leq1.
\end{split}\end{align}
Hence the paired expectation vectors
\begin{align}\begin{split}\nonumber
u&\coloneqq\E_P[f_\alpha(X)]
=\bigl(\E_P[q_k(X)\alpha_k(X)]\bigr)_{k=1}^K\in\E_A[D(X)],\\
u^{(n)}&\coloneqq\E_P[f_\alpha^{(n)}(X)]
=\bigl(\E_P[q_k^{(n)}(X)\alpha_k(X)]\bigr)_{k=1}^K\in\E_A[D^{(n)}(X)]
\end{split}\end{align}
exhaust the respective Aumann expectations as $\alpha$ varies.
Using \eqref{eq:paired-selection-bound},
\begin{align}\begin{split}\nonumber
\lVert u-u^{(n)}\rVert
&=\left\|\E_P\left[\sum_{k=1}^K(q_k(X)-q_k^{(n)}(X))\alpha_k(X)e_k\right]\right\|\\
&\leq\E_P\!\left[\left\|\sum_{k=1}^K(q_k(X)-q_k^{(n)}(X))\alpha_k(X)e_k\right\|\right]\\
&\leq\E_P[1/n]=1/n.
\end{split}\end{align}
For every $u\in\E_A[D(X)]$, the same $\alpha$ supplies a corresponding $u^{(n)}$, and conversely for every $u^{(n)}\in\E_A[D^{(n)}(X)]$.
Therefore,
\begin{align}\begin{split}\nonumber
\sup_{u\in\E_A[D(X)]}
\inf_{u^{(n)}\in\E_A[D^{(n)}(X)]}
\lVert u-u^{(n)}\rVert&\leq\frac1n,\\
\sup_{u^{(n)}\in\E_A[D^{(n)}(X)]}
\inf_{u\in\E_A[D(X)]}
\lVert u-u^{(n)}\rVert&\leq\frac1n.
\end{split}\end{align}
Taking their maximum gives
\begin{align}\begin{split}\label{eq:aumann-hausdorff-approximation}
d_H\bigl(\E_A[D^{(n)}(X)],\E_A[D(X)]\bigr)\leq\frac1n.
\end{split}\end{align}

By Proposition \ref{prop:hausdorff-continuity} and \eqref{eq:aumann-hausdorff-approximation},
\begin{align}\begin{split}\nonumber
\Vol_K\bigl(\E_A[D^{(n)}(X)]\bigr) \to \Vol_K\bigl(\E_A[D(X)]\bigr).
\end{split}\end{align}
The versions of $q_k$ chosen after \eqref{eq:q-properties} satisfy the coordinate bounds everywhere, so \eqref{eq:simplex-hausdorff-approximation} holds for every $x$.
Proposition \ref{prop:mixed-volume-continuity} therefore gives, for every $(x_1,\ldots,x_K)\in\mathcal X^K$,
\begin{align}\begin{split}\label{eq:mixed-volume-pointwise-limit}
V(D^{(n)}(x_1),\ldots,D^{(n)}(x_K))
\to V(D(x_1),\ldots,D(x_K)).
\end{split}\end{align}
For each $n$, the mixed volume $V(D^{(n)}(X_1),\ldots,D^{(n)}(X_K))$ is a measurable finite-valued function.
Its pointwise limit $V(D(X_1),\ldots,D(X_K))$ is therefore measurable.
The containment in \eqref{eq:random-simplex-bound}, together with $0\leq q_k^{(n)}\leq q_k$, gives
\begin{align}\begin{split}\nonumber
D^{(n)}(X_r)\subseteq D(X_r)\subseteq D(\rho_1^{-1},\ldots,\rho_K^{-1}),
\quad r=1,\ldots,K.
\end{split}\end{align}
Monotonicity \eqref{eq:mixed-monotonicity}, normalization \eqref{eq:mixed-normalization}, and \eqref{eq:coordinate-simplex-volume} yield
\begin{align}\begin{split}\label{eq:mixed-volume-dominating-bound}
0\leq V(D^{(n)}(X_1),\ldots,D^{(n)}(X_K))
&\leq V\bigl(D(\rho_1^{-1},\ldots,\rho_K^{-1}),\ldots,D(\rho_1^{-1},\ldots,\rho_K^{-1})\bigr)\\
&=\Vol_K\bigl(D(\rho_1^{-1},\ldots,\rho_K^{-1})\bigr)\\
&=\frac{1}{K!\prod_{k=1}^K\rho_k}<\infty.
\end{split}\end{align}
Equations \eqref{eq:mixed-volume-pointwise-limit} and \eqref{eq:mixed-volume-dominating-bound} allow dominated convergence,
\begin{align}\begin{split}\nonumber
\E_{P^K}[V(D^{(n)}(X_1),\ldots,D^{(n)}(X_K))]
\to\E_{P^K}[V(D(X_1),\ldots,D(X_K))].
\end{split}\end{align}
\end{proof}

Lemma \ref{lem:finite-coordinate-simplices} expresses $\Vol_K(\E_A[D(X)])$ as an expected mixed volume when $D(X)$ takes finitely many values.
For the approximations $D^{(n)}(X)$ in Lemma \ref{lem:finite-simplex-approximation}, this identity and the limits in \eqref{eq:aumann-volume-limit} and \eqref{eq:expected-mixed-volume-limit} give
\begin{align}\begin{split}\nonumber
\Vol_K\bigl(\E_A[D(X)]\bigr)
&=\lim_{n\to\infty}\Vol_K\bigl(\E_A[D^{(n)}(X)]\bigr)\\
&=\lim_{n\to\infty}\E_{P^K}[V(D^{(n)}(X_1),\ldots,D^{(n)}(X_K))]\\
&=\E_{P^K}[V(D(X_1),\ldots,D(X_K))].
\end{split}\end{align}
Equation \eqref{eq:volume-aumann-body} then identifies $\Vol_K(\E_A[D(X)])$ with the ROC volume $\Vol_K(\cH)$.
Theorem \ref{prop:polarization} summarizes the resulting formula.

\begin{theorem}
\label{prop:polarization}
The ROC volume $\Vol_K(\cH)$ equals the expected mixed volume of $K$ independent copies of the random coordinate simplex $D(X)$ defined in \eqref{eq:random-coordinate-simplex}.
\begin{align}\begin{split}\label{eq:expected-mixed-volume}
\Vol_K(\cH)
=\Vol_K\bigl(\E_A[D(X)]\bigr)
=\E_{P^K}\!\left[V(D(X_1),\ldots,D(X_K))\right].
\end{split}\end{align}
\end{theorem}

\begin{proof}
The first equality is \eqref{eq:volume-aumann-body}.
Apply the finite-valued calculation in Lemma \ref{lem:finite-coordinate-simplices} to the approximations $D^{(n)}(X)$ supplied by Lemma \ref{lem:finite-simplex-approximation}.
The limits in \eqref{eq:aumann-volume-limit} and \eqref{eq:expected-mixed-volume-limit} then give
\begin{align}\begin{split}\nonumber
\Vol_K\bigl(\E_A[D(X)]\bigr)
&=\lim_{n\to\infty}\Vol_K\bigl(\E_A[D^{(n)}(X)]\bigr)\\
&=\lim_{n\to\infty}\E_{P^K}[V(D^{(n)}(X_1),\ldots,D^{(n)}(X_K))]\\
&=\E_{P^K}[V(D(X_1),\ldots,D(X_K))].
\end{split}\end{align}
This proves \eqref{eq:expected-mixed-volume}.
\end{proof}

Equation \eqref{eq:expected-mixed-volume} turns the volume calculation for the generally complicated ROC body $\cH$ into an expectation of the mixed volume $V(D(X_1),\ldots,D(X_K))$.
To evaluate this mixed volume, fix feature values $x_1,\ldots,x_K\in\mathcal X$.
By Definition \ref{def:mixed-volume}, the mixed volume satisfies
\begin{align}\begin{split}\nonumber
V(D(x_1),\ldots,D(x_K))
=\frac{1}{K!}
\sum_{\varnothing\neq I\subseteq\{1,\ldots,K\}}
(-1)^{K-\lvert I\rvert}
\Vol_K\!\left(\sum_{r\in I}D(x_r)\right).
\end{split}\end{align}
Each coordinate simplex $D(x_r)$ in this expression has length $q_k(x_r)$ along coordinate axis $k$, as defined in \eqref{eq:random-coordinate-simplex}.
Writing the coordinate lengths of $D(x_r)$ as row $r$ gives the matrix
\begin{align}\begin{split}\label{eq:roc-assignment-matrix}
\bigl(q_k(x_r)\bigr)_{r,k=1}^K
=\begin{pmatrix}
q_1(x_1)&\cdots&q_k(x_1)&\cdots&q_K(x_1)\\
\vdots&\ddots&\vdots&\ddots&\vdots\\
q_1(x_K)&\cdots&q_k(x_K)&\cdots&q_K(x_K)
\end{pmatrix}.
\end{split}\end{align}
Column $k$ records the length $q_k(x_r)$ of each simplex $D(x_r)$ along coordinate axis $k$.
%
%
Theorem \ref{thm:mixed-assignment} states that the mixed volume $V(D(x_1),\ldots,D(x_K))$ is obtained by maximizing this product over all $\pi\in\Sn_K$ and dividing the result by $K!$.

\begin{theorem}
\label{thm:mixed-assignment}
\label{thm:volume-formula}
For any feature values $x_1,\ldots,x_K\in\mathcal X$, the mixed volume of the coordinate simplices $D(x_1),\ldots,D(x_K)$ defined in \eqref{eq:random-coordinate-simplex} is
\begin{align}\begin{split}\label{eq:coordinate-mixed-volume}
V(D(x_1),\ldots,D(x_K))
=\frac{1}{K!}
\max_{\pi\in\Sn_K}
\prod_{r=1}^K q_{\pi(r)}(x_r).
\end{split}\end{align}
Let $X_1,\ldots,X_K$ be independent with distribution $P$.
Substituting \eqref{eq:coordinate-mixed-volume} into \eqref{eq:expected-mixed-volume} gives the volume of the ROC body $\cH$,
\begin{align}\begin{split}\label{eq:main-volume-formula}
\Vol_K(\cH)
& = \frac{1}{K!}\E_{P^K}\!\left[
\max_{\pi\in\Sn_K}\prod_{r=1}^Kq_{\pi(r)}(X_r)\right] \\ 
& = \frac{1}{K!\prod_{k=1}^K\rho_k}\E_{P^K}\!\left[
\max_{\pi\in\Sn_K}\prod_{r=1}^Kp_{\pi(r)}(X_r)\right].
\end{split}\end{align}
\end{theorem}
The second equality in \eqref{eq:main-volume-formula} uses $q_k=p_k/\rho_k$ and the fact that each permutation uses every class exactly once.

\begin{proof}
We prove the formula for arbitrary nonnegative coordinate lengths $a_{rk}$, writing $D_r=D(a_{r1},\ldots,a_{rK})$.
Taking $a_{rk}=q_k(x_r)$ gives \eqref{eq:coordinate-mixed-volume}, and averaging over independent observations then gives \eqref{eq:main-volume-formula}.

\noindent\emph{Step 1. A lower bound from coordinate segments.}
Fix a permutation $\pi\in\Sn_K$ and choose the segment
\begin{align}\begin{split}\nonumber
L_r=[0,a_{r,\pi(r)}e_{\pi(r)}]\subseteq D_r,
\quad r=1,\ldots,K.
\end{split}\end{align}
Let $I$ be a nonempty subset of $\{1,\ldots,K\}$, recording which segments enter the Minkowski sum.
Since $\pi$ uses every coordinate axis exactly once, adding points from these segments gives the following coordinate box.
\begin{align}\begin{split}\nonumber
\sum_{r\in I}L_r
=\left\{u\in\R^K:
\begin{array}{ll}
0\leq u_{\pi(r)}\leq a_{r,\pi(r)} & \text{for }r\in I,\\
u_{\pi(r)}=0 & \text{for }r\notin I
\end{array}
\right\}.
\end{split}\end{align}
When $I=\{1,\ldots,K\}$, the box has side length $a_{r,\pi(r)}$ along coordinate axis $\pi(r)$, so
\begin{align}\begin{split}\nonumber
\Vol_K(L_1+\cdots+L_K)
=\prod_{r=1}^K a_{r,\pi(r)}.
\end{split}\end{align}
If $I\ne\{1,\ldots,K\}$, it omits some index $s\in\{1,\ldots,K\}$.
Every point $u\in\sum_{r\in I}L_r$ then satisfies $u_{\pi(s)}=0$, so the box has a side of length zero along axis $\pi(s)$.
Its $K$-dimensional volume is therefore
\begin{align}\begin{split}\nonumber
\Vol_K\!\left(\sum_{r\in I}L_r\right)=0
\text{ for every nonempty } I \neq \{1,\ldots,K\}.
\end{split}\end{align}
Thus, in \eqref{eq:polytopal-mixed-volume}, every term corresponding to a subset that omits an index is zero.
For $I=\{1,\ldots,K\}$, the factor $(-1)^{K-\lvert I\rvert}$ equals $(-1)^{K-K}=1$.
Substituting these volumes into that formula gives
\begin{align}\begin{split}\nonumber
V(L_1,\ldots,L_K)
&=\frac{1}{K!}
\sum_{\varnothing\neq I\subseteq\{1,\ldots,K\}}
(-1)^{K-\lvert I\rvert}
\Vol_K\!\left(\sum_{r\in I}L_r\right)\\
&=\frac{1}{K!}\Vol_K(L_1+\cdots+L_K)\\
&=\frac{1}{K!}\prod_{r=1}^K a_{r,\pi(r)}.
\end{split}\end{align}
For the fixed permutation $\pi$, the inclusions $L_r\subseteq D_r$ and monotonicity of mixed volume give
\begin{align}\begin{split}\nonumber
V(D_1,\ldots,D_K)
\geq V(L_1,\ldots,L_K)
=\frac{1}{K!}\prod_{r=1}^K a_{r,\pi(r)}.
\end{split}\end{align}
This inequality holds for every $\pi\in\Sn_K$, and $V(D_1,\ldots,D_K)$ is independent of $\pi$.
Taking the maximum of the right-hand side over all permutations therefore gives
\begin{align}\begin{split}\label{eq:assignment-lower}
V(D_1,\ldots,D_K)
\geq\frac{1}{K!}
\max_{\pi\in\Sn_K}\prod_{r=1}^K a_{r,\pi(r)}.
\end{split}\end{align}

\noindent\emph{Step 2. An upper bound when all coordinate lengths are positive.}
Assume $a_{rk}>0$ for every $r,k$.
Taking logarithms turns the product associated with a permutation into a sum.
\begin{align}\begin{split}\nonumber
\log\!\left(\prod_{r=1}^K a_{r,\pi(r)}\right)
=\sum_{r=1}^K\log a_{r,\pi(r)}.
\end{split}\end{align}
We first express the maximization of this sum as a linear program and then use its dual to construct an upper bound on mixed volume.

\noindent\emph{(a) The Birkhoff polytope.}
For a permutation $\pi\in\Sn_K$, the selected vertex of simplex $D_r$ is $a_{r,\pi(r)}e_{\pi(r)}$.
Its \emph{permutation matrix} $b^\pi=(b^\pi_{rk})$ records these choices, with row $r$ corresponding to $D_r$ and column $k$ corresponding to the coordinate direction $e_k$.
The entry $b^\pi_{rk}$ equals one when the vertex $a_{rk}e_k$ is selected and zero otherwise.
\begin{align}\begin{split}\nonumber
b^\pi_{rk}
=\begin{cases}
1,&k=\pi(r),\\
0,&k\ne\pi(r).
\end{cases}
\end{split}\end{align}
Multiplication by $b^\pi_{rk}$ therefore retains exactly the selected log lengths in the following sum.
\begin{align}\begin{split}\nonumber
\sum_{r=1}^K\sum_{k=1}^K b^\pi_{rk}\log a_{rk}
=\sum_{r=1}^K\log a_{r,\pi(r)}.
\end{split}\end{align}
The \emph{Birkhoff polytope} allows the entries of these matrices to be nonnegative real numbers while retaining the row and column sums.
\begin{align}\begin{split}\label{eq:birkhoff-polytope}
\mathcal B_K
\coloneqq\left\{b=(b_{rk})\in\R_+^{K\times K}:
\begin{array}{ll}
\displaystyle\sum_{k=1}^K b_{rk}=1&\text{for }r=1,\ldots,K,\\[4pt]
\displaystyle\sum_{r=1}^K b_{rk}=1&\text{for }k=1,\ldots,K
\end{array}
\right\}.
\end{split}\end{align}
A matrix satisfying these conditions is called \emph{doubly stochastic}.
The \emph{Birkhoff--von Neumann theorem} states that $\mathcal B_K=\conv\{b^\pi:\pi\in\Sn_K\}$ \citep[Theorem 17, p. 167]{Gerards1995}.
Thus, for every $b\in\mathcal B_K$, there are weights $\theta_\pi\geq0$ such that
\begin{align}\begin{split}\nonumber
b=\sum_{\pi\in\Sn_K}\theta_\pi b^\pi,
\quad
\sum_{\pi\in\Sn_K}\theta_\pi=1.
\end{split}\end{align}
The value of the linear objective at $b$ is the same weighted average of its values at the permutation matrices $b^\pi$.
\begin{align}\begin{split}\nonumber
\sum_{r=1}^K\sum_{k=1}^K b_{rk}\log a_{rk}
&=\sum_{\pi\in\Sn_K}\theta_\pi
\sum_{r=1}^K\log a_{r,\pi(r)}
\leq\max_{\pi\in\Sn_K}
\sum_{r=1}^K\log a_{r,\pi(r)}.
\end{split}\end{align}
Conversely, every permutation matrix belongs to $\mathcal B_K$.
The two maximization problems therefore have the same value.
\begin{align}\begin{split}\label{eq:birkhoff-assignment}
\max_{\pi\in\Sn_K}\sum_{r=1}^K\log a_{r,\pi(r)}
=\max_{b\in\mathcal B_K}\sum_{r=1}^K\sum_{k=1}^K b_{rk}\log a_{rk}.
\end{split}\end{align}

\noindent\emph{(b) Linear programming duality.}
Introduce a dual variable $u_r\in\R$ for each row constraint and $v_k\in\R$ for each column constraint in \eqref{eq:birkhoff-polytope}.
These variables range over all real numbers because the row and column constraints are equalities.
For completeness, the Lagrangian of this maximization problem is
\begin{align}\begin{split}\nonumber
\mathcal L_{\mathrm{ass}}(b;u,v)
&\coloneqq\sum_{r=1}^K\sum_{k=1}^K b_{rk}\log a_{rk}
 +\sum_{r=1}^K u_r\left(1-\sum_{k=1}^K b_{rk}\right) +\sum_{k=1}^K v_k\left(1-\sum_{r=1}^K b_{rk}\right)\\
&=\sum_{r=1}^K u_r+\sum_{k=1}^K v_k
 +\sum_{r=1}^K\sum_{k=1}^K b_{rk}(\log a_{rk}-u_r-v_k).
\end{split}\end{align}
In forming the dual, only the constraints $b_{rk}\geq0$ remain inside the supremum.
If one coefficient $\log a_{rk}-u_r-v_k$ is positive, sending that entry of $b$ to infinity makes this supremum infinite.
If every coefficient is nonpositive, $b=0$ maximizes the Lagrangian, with value $\sum_ru_r+\sum_kv_k$.
Thus the dual objective is finite exactly when the following constraints hold.
\begin{align}\begin{split}\nonumber
u_r+v_k\geq\log a_{rk},
\quad r,k=1,\ldots,K.
\end{split}\end{align}
To see why these constraints give an upper bound, take any $b\in\mathcal B_K$ and any $u,v$ satisfying them.
Nonnegativity of $b_{rk}$ and the row and column sums in \eqref{eq:birkhoff-polytope} give
\begin{align}\begin{split}\nonumber
\sum_{r=1}^K\sum_{k=1}^K b_{rk}\log a_{rk}
&\leq\sum_{r=1}^K\sum_{k=1}^K b_{rk}(u_r+v_k)\\
&=\sum_{r=1}^K u_r\left(\sum_{k=1}^K b_{rk}\right)
+\sum_{k=1}^K v_k\left(\sum_{r=1}^K b_{rk}\right)\\
&=\sum_{r=1}^K u_r+\sum_{k=1}^K v_k.
\end{split}\end{align}
Minimizing this upper bound subject to $u_r+v_k\geq\log a_{rk}$ is the dual linear program.
The \emph{linear programming strong-duality theorem} makes this bound exact and guarantees the existence of a dual minimizer.\footnote{
For a feasible finite-dimensional linear program with finite optimal value, the primal and dual both have optimal solutions and equal optimal values \citep[equation (34), p. 164]{Gerards1995}.
Here the primal maximizes $\sum_{r=1}^K\sum_{k=1}^K b_{rk}\log a_{rk}$ over $b\in\mathcal B_K$, and the dual minimizes $\sum_{r=1}^K u_r+\sum_{k=1}^K v_k$ over $u,v\in\R^K$ with $u_r+v_k\geq\log a_{rk}$.
The corresponding minimum-weight assignment formula is given in \citet[Section 6.1, equation (53), p. 173]{Gerards1995}.}
The primal is feasible because $\mathcal B_K$ contains every permutation matrix.
The set $\mathcal B_K$ is compact because its constraints are closed and imply $0\leq b_{rk}\leq1$.
The objective is continuous because all $\log a_{rk}$ are finite.
The primal therefore has an optimal solution with finite objective value, so the theorem applies.
Combining strong duality with \eqref{eq:birkhoff-assignment} gives
\begin{align}\begin{split}\label{eq:assignment-dual}
\max_{\pi\in\Sn_K}\sum_{r=1}^K\log a_{r,\pi(r)}
&=\max_{b\in\mathcal B_K}
\sum_{r=1}^K\sum_{k=1}^K b_{rk}\log a_{rk}\\
&=
\min_{\substack{u,v\in\R^K\\
u_r+v_k\geq\log a_{rk},\ r,k=1,\ldots,K}}
\left\{\sum_{r=1}^K u_r+\sum_{k=1}^K v_k\right\}.
\end{split}\end{align}

\noindent\emph{(c) Containing the simplices and bounding mixed volume.}
Choose minimizing vectors $u,v$ and define
\begin{align}\begin{split}\nonumber
\lambda_r=e^{u_r},
\quad
\mu_k=e^{v_k},
\quad
\mu=(\mu_1,\ldots,\mu_K).
\end{split}\end{align}
Exponentiating the constraints in \eqref{eq:assignment-dual} gives, for every $r,k$,
\begin{align}\begin{split}\nonumber
a_{rk}=e^{\log a_{rk}}
\leq e^{u_r+v_k}
=\lambda_r\mu_k.
\end{split}\end{align}
The vector $\mu$ specifies one coordinate simplex $D(\mu)$, and $\lambda_r$ scales it to contain $D_r$.
To verify this containment, Definition \ref{def:coordinate-simplex} gives
\begin{align}\begin{split}\nonumber
\lambda_rD(\mu)
=\conv\{0,\lambda_r\mu_1e_1,\ldots,\lambda_r\mu_Ke_K\}.
\end{split}\end{align}
Since $0<a_{rk}/(\lambda_r\mu_k)\leq1$, each vertex $a_{rk}e_k$ of $D_r$ is the following convex combination of two points of $\lambda_rD(\mu)$.
\begin{align}\begin{split}\nonumber
a_{rk}e_k
=\frac{a_{rk}}{\lambda_r\mu_k}(\lambda_r\mu_ke_k)
+\left(1-\frac{a_{rk}}{\lambda_r\mu_k}\right)0
\in\lambda_rD(\mu).
\end{split}\end{align}
The origin also belongs to $\lambda_rD(\mu)$, so convexity gives
\begin{align}\begin{split}\nonumber
D_r=\conv\{0,a_{r1}e_1,\ldots,a_{rK}e_K\}
\subseteq\lambda_rD(\mu),
\quad r=1,\ldots,K.
\end{split}\end{align}
Monotonicity of mixed volume applies to these $K$ inclusions.
Homogeneity in each argument, \eqref{eq:mixed-normalization}, and the simplex volume formula \eqref{eq:coordinate-simplex-volume} then give
\begin{align}\begin{split}\nonumber
V(D_1,\ldots,D_K)
&\leq V(\lambda_1D(\mu),\ldots,\lambda_KD(\mu))\\
&=\left(\prod_{r=1}^K\lambda_r\right)V(D(\mu),\ldots,D(\mu))\\
&=\left(\prod_{r=1}^K\lambda_r\right)\Vol_K(D(\mu)) 
=\frac{1}{K!}
\left(\prod_{r=1}^K\lambda_r\right)
\left(\prod_{k=1}^K\mu_k\right).
\end{split}\end{align}
Because $u,v$ solve the minimization problem in \eqref{eq:assignment-dual}, the product in this upper bound satisfies
\begin{align}\begin{split}\nonumber
\left(\prod_{r=1}^K\lambda_r\right)
\left(\prod_{k=1}^K\mu_k\right)
&=\exp\!\left(\sum_{r=1}^K u_r+\sum_{k=1}^K v_k\right)\\
&=\exp\!\left(\max_{\pi\in\Sn_K}
\sum_{r=1}^K\log a_{r,\pi(r)}\right) 
=\max_{\pi\in\Sn_K}\prod_{r=1}^K a_{r,\pi(r)}.
\end{split}\end{align}
This proves the upper bound matching \eqref{eq:assignment-lower}.

\noindent\emph{Step 3. Coordinate lengths equal to zero.}
For arbitrary $a_{rk}\geq0$, replace each $a_{rk}$ by $a_{rk}+\varepsilon$, where $\varepsilon>0$.
Steps 1 and 2 prove the formula for the resulting positive coordinate lengths.
Write $D_r^{(\varepsilon)}=D(a_{r1}+\varepsilon,\ldots,a_{rK}+\varepsilon)$.
Pairing points represented by the same nonnegative coefficients $\alpha_k$ with $\sum_k\alpha_k\leq1$ in \eqref{eq:coordinate-simplex-representation} gives
\begin{align}\begin{split}\nonumber
\left\|\sum_{k=1}^K\alpha_k(a_{rk}+\varepsilon)e_k
      -\sum_{k=1}^K\alpha_ka_{rk}e_k\right\|
\leq\varepsilon\sum_{k=1}^K\alpha_k\leq\varepsilon.
\end{split}\end{align}
The pairing works in both directions, so $d_H(D_r^{(\varepsilon)},D_r)\leq\varepsilon$.
The corresponding simplices therefore converge to $D_1,\ldots,D_K$ in Hausdorff distance as $\varepsilon\downarrow0$.
Their mixed volume converges by \eqref{eq:mixed-continuity}, which also covers lower-dimensional limiting simplices.
Each product $\prod_{r=1}^K(a_{r,\pi(r)}+\varepsilon)$ also converges to $\prod_{r=1}^K a_{r,\pi(r)}$.
There are finitely many permutations, so the maximum of these products converges to $\max_{\pi\in\Sn_K}\prod_{r=1}^K a_{r,\pi(r)}$.
Taking limits proves the formula for all nonnegative coordinate lengths.

\noindent\emph{Step 4. The ROC volume formula.}
Proposition \ref{prop:aumann-roc} gives the Aumann representation of $\cH$.
Theorem \ref{prop:polarization} then expresses its volume as an expected mixed volume.
Applying \eqref{eq:coordinate-mixed-volume} with $x_r=X_r$ evaluates $V(D(X_1),\ldots,D(X_K))$ for every realization of $(X_1,\ldots,X_K)$.
Combining these identities gives
\begin{align}\begin{split}\label{eq:roc-volume-calculation-chain}
\Vol_K(\cH)
&=\Vol_K\bigl(\E_A[D(X)]\bigr)\\
&=\E_{P^K}\!\left[V(D(X_1),\ldots,D(X_K))\right]\\
&=\frac{1}{K!}\E_{P^K}\!\left[
\max_{\pi\in\Sn_K}\prod_{r=1}^Kq_{\pi(r)}(X_r)
\right].
\end{split}\end{align}
The maximum is measurable because it is the maximum of finitely many products of measurable nonnegative functions.
Equation \eqref{eq:q-properties} also gives the pointwise bound
\begin{align}\begin{split}\nonumber
0\leq\max_{\pi\in\Sn_K}\prod_{r=1}^Kq_{\pi(r)}(X_r)
\leq\prod_{k=1}^K\rho_k^{-1}<\infty,
\end{split}\end{align}
because each permutation uses each class once.
Hence the expectations in \eqref{eq:roc-volume-calculation-chain} are finite, as also follows from \eqref{eq:mixed-volume-dominating-bound}.
\end{proof}


\begin{example}[Permutations and mixed volume]
\label{ex:two-permutations}
\leavevmode\par
Figure \ref{fig:matrix-volume-contributions} illustrates the two permutations for $K=2$ using the coordinate simplices $A_1,A_2$ from Example \ref{ex:minkowski-area} and Figure \ref{fig:mixed-area-decomposition}, now written as $D(x_1),D(x_2)$.
With $\lambda_1=\lambda_2=1$, the area decomposition in Figure \ref{fig:mixed-area-decomposition}(b) gives a gray rectangle of area 
\begin{align}\begin{split}\nonumber
9/4=2V(D(x_1),D(x_2)).
\end{split}\end{align}
The same rectangle appears in Figure \ref{fig:matrix-volume-contributions}(a), while the alternative assignment in panel (b) gives a rectangle of area $1/4$.
The larger rectangle area $9/4$ equals the maximum in \eqref{eq:coordinate-mixed-volume}, so
\begin{align}\begin{split}\nonumber
V(D(x_1),D(x_2))=\frac1{2!}\max\left\{\frac94,\frac14\right\}=\frac98.
\end{split}\end{align}

\begin{figure}[htbp]
\centering
\begin{tikzpicture}[x=1cm,y=1cm,>=stealth,font=\small]
\node[font=\small\bfseries] at (1.65,7.15) {Matrix entries};
\node[font=\small\bfseries] at (6.45,7.15) {Coordinate segments};
\node[font=\small\bfseries,align=center] at (11.4,7.15) {Sum of selected\\segments};
\foreach \selectioncase/\verticalshift in {1/4.0,2/0} {
  \begin{scope}[yshift=\verticalshift cm]
  \node[font=\scriptsize] at (1.25,2.18) {class 1};
  \node[font=\scriptsize] at (2.50,2.18) {class 2};
  \node[blue!65!black] at (0.15,1.50) {$x_1$};
  \node[orange!80!black] at (0.15,0.75) {$x_2$};
  \draw[black!65] (0.75,0.35) -- (0.60,0.35) -- (0.60,1.90) -- (0.75,1.90);
  \draw[black!65] (3.00,0.35) -- (3.15,0.35) -- (3.15,1.90) -- (3.00,1.90);
  \ifnum\selectioncase=1
    \node[draw=blue!65!black,fill=blue!12,line width=0.8pt,inner sep=5pt] at (1.25,1.50) {$1.5$};
    \node at (2.50,1.50) {$0.5$};
    \node at (1.25,0.75) {$0.5$};
    \node[draw=orange!80!black,fill=orange!17,line width=0.8pt,inner sep=5pt] at (2.50,0.75) {$1.5$};
  \else
    \node at (1.25,1.50) {$1.5$};
    \node[draw=blue!65!black,fill=blue!12,line width=0.8pt,inner sep=5pt] at (2.50,1.50) {$0.5$};
    \node[draw=orange!80!black,fill=orange!17,line width=0.8pt,inner sep=5pt] at (1.25,0.75) {$0.5$};
    \node at (2.50,0.75) {$1.5$};
  \fi
  \draw[->,black!60] (3.45,1.12) -- (4.30,1.12);
  \begin{scope}[xshift=5.0cm,x=1.35cm,y=1.35cm]
    \fill[blue!12] (0,0) -- (1.5,0) -- (0,0.5) -- cycle;
    \fill[orange!17] (0,0) -- (0.5,0) -- (0,1.5) -- cycle;
    \draw[blue!65!black,line width=0.8pt] (0,0) -- (1.5,0) -- (0,0.5) -- cycle;
    \draw[orange!80!black,line width=0.8pt] (0,0) -- (0.5,0) -- (0,1.5) -- cycle;
    \draw[->,black!60] (-0.06,0) -- (1.75,0) node[right] {$e_1$};
    \draw[->,black!60] (0,-0.06) -- (0,1.75) node[above] {$e_2$};
    \node[below left,font=\scriptsize] at (0,0) {$0$};
    \node[blue!65!black,font=\scriptsize] at (1.03,0.43) {$D(x_1)$};
    \node[orange!80!black,font=\scriptsize] at (0.48,1.10) {$D(x_2)$};
    \ifnum\selectioncase=1
      \draw[blue!65!black,line width=2.2pt] (0,0) -- (1.5,0);
      \draw[orange!80!black,line width=2.2pt] (0,0) -- (0,1.5);
      \fill[blue!65!black] (1.5,0) circle (1.6pt);
      \fill[orange!80!black] (0,1.5) circle (1.6pt);
      \node[below] at (1.5,0) {$1.5$};
      \node[left] at (0,1.5) {$1.5$};
    \else
      \draw[blue!65!black,line width=2.2pt] (0,0) -- (0,0.5);
      \draw[orange!80!black,line width=2.2pt] (0,0) -- (0.5,0);
      \fill[blue!65!black] (0,0.5) circle (1.6pt);
      \fill[orange!80!black] (0.5,0) circle (1.6pt);
      \node[below] at (0.5,0) {$0.5$};
      \node[left] at (0,0.5) {$0.5$};
    \fi
  \end{scope}
  \draw[->,black!60] (8.05,1.12) -- (9.10,1.12);
  \begin{scope}[xshift=10.0cm,x=1.35cm,y=1.35cm]
    \ifnum\selectioncase=1
      \fill[black!7] (0,0) rectangle (1.5,1.5);
      \draw[black!70,line width=0.8pt] (0,0) rectangle (1.5,1.5);
    \else
      \fill[black!7] (0,0) rectangle (0.5,0.5);
      \draw[black!70,line width=0.8pt] (0,0) rectangle (0.5,0.5);
    \fi
    \draw[->,black!60] (-0.06,0) -- (1.75,0) node[right] {$e_1$};
    \draw[->,black!60] (0,-0.06) -- (0,1.75) node[above] {$e_2$};
    \node[below left,font=\scriptsize] at (0,0) {$0$};
    \ifnum\selectioncase=1
      \draw[blue!65!black,line width=2.2pt] (0,0) -- (1.5,0);
      \draw[orange!80!black,line width=2.2pt] (0,0) -- (0,1.5);
      \node at (0.75,0.75) {$9/4$};
      \node[below] at (1.5,0) {$1.5$};
      \node[left] at (0,1.5) {$1.5$};
    \else
      \draw[blue!65!black,line width=2.2pt] (0,0) -- (0,0.5);
      \draw[orange!80!black,line width=2.2pt] (0,0) -- (0.5,0);
      \node[font=\scriptsize] at (0.25,0.25) {$1/4$};
      \node[below] at (0.5,0) {$0.5$};
      \node[left] at (0,0.5) {$0.5$};
    \fi
  \end{scope}
  \end{scope}
}

\coordinate (panel-name-center) at (current bounding box.center);
\node[align=center] at (panel-name-center |- 0,3.20) {(a) $\pi(1)=1,\ \pi(2)=2$};
\node[align=center] at (panel-name-center |- 0,-0.80) {(b) $\pi(1)=2,\ \pi(2)=1$};
\end{tikzpicture}
\caption{From matrix entries to mixed volume}
\label{fig:matrix-volume-contributions}
\begin{minipage}{0.94\linewidth}
\footnotesize
\textit{Note.} We illustrate the mixed volume for $K=2$, with a simplex whose matrix has rows $(3/2,1/2)$ and $(1/2,3/2)$.
These rows give the same coordinate simplices as $A_1,A_2$ in Figure \ref{fig:mixed-area-decomposition}.
In panels (a) and (b), the boxed entries determine the highlighted segments, one from each simplex, whose Minkowski sum is the shaded rectangle.
Blue identifies the contribution from $x_1$, and orange identifies the contribution from $x_2$.
The rectangle areas $9/4$ and $1/4$ are the two permutation products in \eqref{eq:coordinate-mixed-volume}.
\end{minipage}
\end{figure}
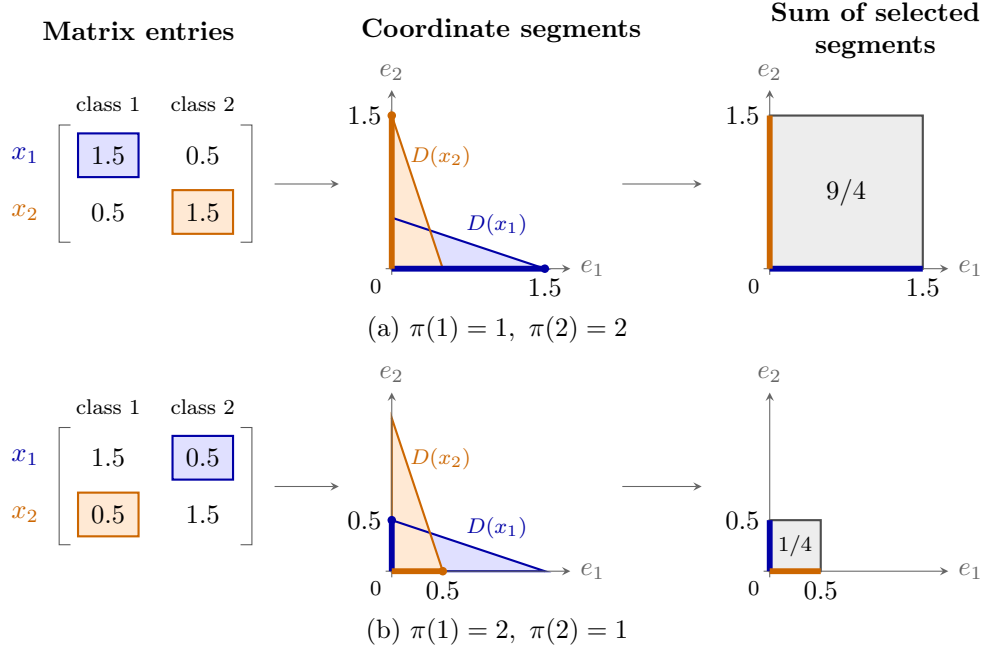
\end{example}


\begin{remark}[Advantages of the method]
\label{rem:roc-volume-method}
The ROC volume formula \eqref{eq:main-volume-formula} in Theorem \ref{thm:volume-formula} follows from three steps.
Equation \eqref{eq:roc-aumann-representation} in Proposition \ref{prop:aumann-roc} expresses the ROC body $\cH$ as the Aumann expectation $\E_A[D(X)]$ of a random coordinate simplex.
Equation \eqref{eq:expected-mixed-volume} in Theorem \ref{prop:polarization} expresses the volume of this expectation as the expected mixed volume of independent copies of $D(X)$.
Equation \eqref{eq:coordinate-mixed-volume} in Theorem \ref{thm:mixed-assignment} evaluates that mixed volume as the largest permutation product divided by $K!$.
This last step introduces the maximum through the geometry of coordinate simplices.
Each permutation selects coordinate segments whose sum is a box, and the largest such box has volume $K!V(D(x_1),\ldots,D(x_K))$, as illustrated in Figures \ref{fig:mixed-area-decomposition} and \ref{fig:matrix-volume-contributions}.

In contrast to our approach, the Lagrangian method first determines the frontier height $\beta_1(t)$ in \eqref{eq:constrained-frontier}.
For $t\in\Lambda$, introduce nonnegative multipliers $\xi=(\xi_2,\ldots,\xi_K)$ for the requirements $T_k(\phi)\geq t_k$, $k=2,\ldots,K$, and define the Lagrangian
\begin{align}\begin{split}\nonumber
\mathcal L(\phi,\xi,t)
\coloneqq T_1(\phi)+\sum_{k=2}^K\xi_k\bigl(T_k(\phi)-t_k\bigr).
\end{split}\end{align}
The frontier height is then given by
\begin{align}\begin{split}\nonumber
\beta_1(t)
=\inf_{\xi\in\R_+^{K-1}}
\left\{\E_P\!\left[\max\{q_1(X),\xi_2q_2(X),\ldots,\xi_Kq_K(X)\}\right]
-\sum_{k=2}^K\xi_kt_k\right\}.
\end{split}\end{align}
For fixed $\xi$, maximizing $\mathcal L(\phi,\xi,t)$ over allocation rules amounts to choosing, at each $x$, a class with the largest weighted density ratio among $q_1(x),\xi_2q_2(x),\ldots,\xi_Kq_K(x)$.
The Lagrangian depends on $\phi$ only through $T(\phi)\in\Gamma$, and $\Gamma$ is compact and convex by Proposition \ref{prop:tpr-region}.
The minimax theorem of \citet{Sion1958} therefore gives
\begin{align}\begin{split}\nonumber
\beta_1(t)=\sup_\phi\inf_{\xi\in\R_+^{K-1}}\mathcal L(\phi,\xi,t)=\inf_{\xi\in\R_+^{K-1}}\sup_\phi\mathcal L(\phi,\xi,t).
\end{split}\end{align}

Computing $\Vol_K(\cH)=\int_\Lambda\beta_1(t)\dd t$ from this representation requires solving the multiplier problem for different $t\in\Lambda$.
In contrast, our method shifts attention from the frontier values $\beta_1(t)$ to the entire ROC body $\cH$.
The Aumann representation of $\cH$ allows us to compute its volume without first reconstructing the frontier.
The remaining optimization is a finite assignment problem using the $K^2$ density-ratio values $q_k(x_r)$ for each tuple $(x_1,\ldots,x_K)$, followed by an ordinary expectation under $P^K$.
\end{remark}

\begin{remark}[Intuitive interpretation]
\label{rem:roc-volume-interpretation}
The ROC volume $\Vol_K(\cH)$ summarizes the information that the observed features $X$ provide about the class label $C$.
Geometrically, it is the probability that uniformly drawn accuracy requirements can be met by an allocation rule based on $X$.
For a requirement vector $U$ drawn uniformly from $[0,1]^K$, \eqref{eq:roc-body-feasibility} shows that $U\in\cH$ precisely when some allocation rule $\phi$ satisfies $T_k(\phi)\geq U_k$ for every $k$.
Definition \ref{def:roc-body} therefore gives $\Vol_K(\cH)=\Pp(U\in\cH)=\int_\Lambda\beta_1(t)\dd t$.
Features that allow more demanding requirements to be met enlarge $\cH$ and increase its volume.

The maximum in \eqref{eq:main-volume-formula} has a statistical interpretation in terms of using the information in the observed features $X$ to recover hidden class labels.
Draw one observation independently from each class, place the observations in uniformly random order, and hide their class labels.
For the observed tuple $(x_1,\ldots,x_K)$, a candidate assignment $\pi$ represents the hypothesis that observation $r$ comes from class $\pi(r)$.
Under this hypothesis, observation $r$ has feature distribution $P_{\pi(r)}$, whose density relative to $P$ is $q_{\pi(r)}$.
Independence of the observations therefore gives $\prod_{r=1}^Kq_{\pi(r)}(x_r)$ as the likelihood of the candidate $\pi$ relative to $P^K$.
The uniform random ordering makes all assignments equally likely before the features are observed, so choosing the largest likelihood maximizes the probability of recovering all labels correctly.
This optimal probability is $1/K!$ times the $P^K$-expectation of the largest likelihood, and therefore equals $\Vol_K(\cH)$ by \eqref{eq:main-volume-formula} in Theorem \ref{thm:volume-formula}.

This interpretation also explains the volume bounds in \eqref{eq:geometric-roc-volume-bounds}.
If the observed features $X$ contain no information about the class label $C$, then $P_1=\cdots=P_K$, and all $K!$ assignments remain equally likely after the features are observed.
No rule can then improve on guessing a fixed assignment, so the optimal success probability and $\Vol_K(\cH)$ both equal $1/K!$.
At the other extreme, if $X$ identifies $C$ perfectly, every class label can be recovered correctly and $\Vol_K(\cH)=1$.
These cases illustrate the endpoints of
\begin{align}\begin{split}\label{eq:assignment-volume-bounds}
\frac1{K!}\leq\Vol_K(\cH)\leq1.
\end{split}\end{align}
The difference $\Vol_K(\cH)-1/K!$ measures how much the observed features $X$ improve the probability of recovering all labels over guessing without using the features.
\end{remark}

\begin{example}[Binary ROC AUC]
\label{ex:binary-roc-volume}
\leavevmode\par
For $K=2$, write $p=p_2$ and $\rho=\rho_2$.
Let $X_1,X_2$ be independent with distribution $P$.
Formula \eqref{eq:main-volume-formula} gives
\begin{align}\begin{split}\nonumber
\Vol_2(\cH)
=\frac{1}{2\rho(1-\rho)}
\E_{P\otimes P}\!\left[
\max\{p(X_1)(1-p(X_2)),p(X_2)(1-p(X_1))\}\right].
\end{split}\end{align}
The larger product is determined by the ordering of the posterior probabilities, since
\begin{align}\begin{split}\nonumber
p(X_1)(1-p(X_2))-p(X_2)(1-p(X_1))=p(X_1)-p(X_2).
\end{split}\end{align}
Splitting the expectation according to $p(X_1)>p(X_2)$, $p(X_1)<p(X_2)$, and $p(X_1)=p(X_2)$ gives
\begin{align}\begin{split}\nonumber
2\rho(1-\rho)\Vol_2(\cH)
&=\E_{P\otimes P}\!\left[p(X_1)(1-p(X_2))\1\{p(X_1)>p(X_2)\}\right]\\
&\quad+\E_{P\otimes P}\!\left[p(X_2)(1-p(X_1))\1\{p(X_1)<p(X_2)\}\right]\\
&\quad+\E_{P\otimes P}\!\left[p(X_1)(1-p(X_2))\1\{p(X_1)=p(X_2)\}\right].
\end{split}\end{align}
Since $X_{1}, X_{2}$ are i.i.d., $(X_1,X_2)$ and $(X_2,X_1)$ have the same distribution, exchanging them in the second term yields
\begin{align}\begin{split}\nonumber
\E_{P\otimes P}\!\left[p(X_2)(1-p(X_1))\1\{p(X_1)<p(X_2)\}\right]
=\E_{P\otimes P}\!\left[p(X_1)(1-p(X_2))\1\{p(X_1)>p(X_2)\}\right].
\end{split}\end{align}
Combining the first two terms and dividing by $2\rho(1-\rho)$ therefore gives
\begin{align}\begin{split}\nonumber
\Vol_2(\cH)=\frac{1}{\rho(1-\rho)}\E_{P\otimes P}\!\Bigl[p(X_1)(1-p(X_2))\Bigl(\1\{p(X_1)>p(X_2)\}+\tfrac12\1\{p(X_1)=p(X_2)\}\Bigr)\Bigr],
\end{split}\end{align}
which is the conventional AUC formula.\footnote{The AUC formula can also be derived from a quantile equivalent form of the ROC curve. See \cite{shapiro2021lectures} section 6.2 for a general discussion in the background of Average Value-at-Risk. 
The tie can not be omitted when there is some discrete component in the distribution of $p\left(X\right)$.}
\end{example}

For the subsequent estimation and inference analysis, write
$x_{1:K}\coloneqq(x_1,\ldots,x_K)$ and $X_{1:K}\coloneqq(X_1,\ldots,X_K)$, where $X_1,\ldots,X_K$ are independent with distribution $P$.
For a given feature tuple $x_{1:K}$, let $A_\pi(q;x_{1:K})$ denote the product selected by $\pi$ from the density-ratio matrix in \eqref{eq:roc-assignment-matrix}, and let $A^*(q;x_{1:K})$ denote the largest such product.
\begin{align}\begin{split}\nonumber
A_\pi(q;x_{1:K})\coloneqq\prod_{r=1}^Kq_{\pi(r)}(x_r),
\quad
A^*(q;x_{1:K})\coloneqq\max_{\pi\in\Sn_K}A_\pi(q;x_{1:K}).
\end{split}\end{align}
Here $q=(q_1,\ldots,q_K)$, and replacing $q$ by $p=(p_1,\ldots,p_K)$ defines $A_\pi(p;x_{1:K})$ and $A^*(p;x_{1:K})$.
The same definitions apply to candidate and estimated probability vectors.
Since each permutation uses every class exactly once,
\begin{align}\begin{split}\nonumber
A_\pi(q;x_{1:K})=\frac{A_\pi(p;x_{1:K})}{\prod_{k=1}^K\rho_k},
\quad
A^*(q;x_{1:K})=\frac{A^*(p;x_{1:K})}{\prod_{k=1}^K\rho_k}.
\end{split}\end{align}
Theorem \ref{thm:volume-formula} therefore gives the population VUS parameter used in the next section,
\begin{align}\begin{split}\label{eq:theta}
\theta_0\coloneqq \text{VUS} =\Vol_K(\cH)
&=\frac1{K!}\E_{P^K}\!\left[A^*(q;X_{1:K})\right]\\
&=\frac{\E_{P^K}[A^*(p;X_{1:K})]}{K!\prod_{k=1}^K\rho_k}.
\end{split}\end{align}
The geometric bounds in \eqref{eq:geometric-roc-volume-bounds} give $1/K!\leq\theta_0\leq1$.
When comparing models or features on the same population, omitting the common positive factor $\prod_{k=1}^K\rho_k$ does not change their ranking.

\section{Double/Debiased Estimation and Inference of ROC-VUS}\label{sample VUS section}


\subsection{Neyman orthogonality}\label{neyman-orthogonality}

The population VUS parameter $\theta_0$ can in principle be estimated with sample analogs by replacing $p(\cdot)$  with 
parametric or machine learning estimate of $\hat p(\cdot)$, $\rho_r$ with their empirical counterparts, and $\E_{P^K}$ with a $K$th order U-statistic based on the observed data: 
\begin{align}\begin{split}\nonumber
  \binom{n}{K}^{-1}\sum_{1\le i_1<\cdots<i_K\le n} 
  \frac1{K!}\sum_{\pi\in\Sn_K}
  \frac1{K!} A^*(\hat p;X_{i_{\pi(1)}},\ldots,X_{i_{\pi(K)}}) \Big/ \prod_{r=1}^K \hat \rho_r
\end{split}\end{align}
However, providing theoretic guarantee for such a model based estimator is challenging due to the complex structure when $\hat p(\cdot)$ is obtained using machine learning methods. Instead,
this section develops double/debiased estimation and inference procedures for the population VUS parameter $\theta_0$. Section \ref{neyman-orthogonality} verifies a Neyman orthogonal estimation form, and section \ref{sample-splitting} combines sample splitting with Neyman orthogonality to construct an asymptotically efficient double/debiased estimator for the VUS parameter $\theta_0$. 

When the distribution of $p(x)=\{p_1(x),\ldots,p_K(x)\}$ contains discrete components, the maximizing permutation may not be unique both in population and in sample. In other words, 
both 
\begin{align}\begin{split}\nonumber
\arg\max_{\pi\in\Sn_K}A_\pi(p;x_{1:K}) \quad \text{and} \quad
\arg\max_{\pi\in\Sn_K}A_\pi(\hat p;x_{1:K})
\end{split}\end{align}
can be non-singleton sets. 
It is possible to define a double debiased estimator without specifying a consistent tie-breaking rule to select a unique maximizing permutation from these sets. 
However, this comes at the cost of substantial additional notation complexity in both the theoretical analysis and the practical implementation.
Therefore we begin with a discussion of tie breaking and assignment representation of the optimizing permutation which we used to define a double/debiased machine learning estimator. 

\subsubsection{Double/debiased estimator}

Let $M \coloneqq K!$ and fix once and for all an ordering
\begin{align}\begin{split}\nonumber
\pi_1\prec\pi_2\prec\cdots\prec\pi_M
\end{split}\end{align}
of the permutations. For a candidate score vector $\bar p=(\bar p_1,\ldots,\bar p_K)$, abbreviate
\begin{align}\begin{split}\nonumber
A_m(\bar p;x) \coloneqq A_{\pi_m}(\bar p;x),
\quad m=1,\ldots,M.
\end{split}\end{align}
For every $m \prec \ell$, define the weak pairwise comparison
\begin{align}\begin{split}\nonumber
J_{m\ell}(\bar p;x)
\coloneqq \1\{A_m(\bar p;x)\geq A_\ell(\bar p;x)\}.
\end{split}\end{align}
The indicator that $\pi_m$ is the smallest-index maximizer is
\begin{align}\begin{split}\label{eq:chi}
\chi_m(\bar p;x)
\coloneqq 
\prod_{\ell \prec m}\{1-J_{\ell m}(\bar p;x)\}
\prod_{\ell \succ m}J_{m\ell}(\bar p;x).
\end{split}\end{align}
Exactly one of $\chi_1(\bar p;x),\ldots,\chi_M(\bar p;x)$ equals one. Hence
\begin{align}\begin{split}\nonumber
\pi_{\bar p}^{\dagger}(x)
\coloneqq \sum_{m=1}^M\pi_m\chi_m(\bar p;x)
=\min_{\prec}\argmax_{\pi\in\Sn_K}A_\pi(\bar p;x)
\end{split}\end{align}
is single-valued for every $x$ and every candidate score. 
Equation \eqref{eq:chi} is the assignment counterpart of the strict/weak tie-breaking rule in \cite{feng2026statistical}: an alternative with a smaller index must be beaten strictly, while an alternative with a larger index need only be beaten weakly. 
The maximum value remains unchanged under the tie breaking assignment rule.


Next we define the Neyman orthogonal score for each $[K]$-tuple of observations in the U-statistic summation form of the double/debiased estimator.
For the $r$th observation in each $[K]$-tuple, define the one-hot encoding class indicator
$Z_{rk} \coloneqq 1\{C_r=k\}$,
$k\in[K]$.
Then
\begin{align}\begin{split}\label{eq:class-signal-mean}
\E[Z_{rk}\mid X_r]=p_k(X_r).
\end{split}\end{align}
For an arbitrary candidate score $\bar p$, define the ordered kernel
\begin{align}\begin{split}\nonumber
\psi_{\bar p}(O_{1:K})
\coloneqq \frac1{K!}
\sum_{m=1}^M\chi_m(\bar p;X_{1:K})
\prod_{r=1}^K Z_{r,\pi_m(r)}.
\end{split}\end{align}
Equivalently, the sum contains only the term selected by $\pi_{\bar p}^{\dagger}$. Conditional independence and \eqref{eq:class-signal-mean} give
\begin{align}\begin{split}\label{eq:conditional-evaluation}
\E[\psi_{\bar p}(O_{1:K})\mid X_{1:K}]
=\frac1{K!}
 A_{\pi_{\bar p}^{\dagger}(X_{1:K})}(p;X_{1:K}).
\end{split}\end{align}
Thus $\bar p$ is used only to choose the permutation. 
The observed class labels evaluate that selected permutation under the true conditional distribution.

The smallest-index rule need not be equivariant with respect to a relabeling of the rows on an exact tie. As conventional in U-statistics estimation, consider a symmetrized version of the ordered kernel:
\begin{align}\begin{split}\nonumber
(\Sym f)(o_{1:K})
\coloneqq \frac1{K!}\sum_{\sigma\in\Sn_K}
 f(o_{\sigma(1)},\ldots,o_{\sigma(K)}),
\quad
h_{\bar p} \coloneqq \Sym\psi_{\bar p}.
\end{split}\end{align}
Symmetrization leaves expectations unchanged. With $h_0 \coloneqq h_p$,
$\E[h_0(O_{1:K})]=\theta_0\prod_{k=1}^K\rho_k$.

Let $\Tcal_n$ be the sigma-field generated by a training sample 
independent of the evaluation observations $\left\{O_i\right\}_{i = 1}^{n}$, and let $\widehat p$ be $\Tcal_n$-measurable. 
Also let
$\zeta(O_{1:K}) = \prod_{r=1}^K Z_{r,r}$ and $g = \Sym \zeta$, that is, 
\begin{align}\begin{split}\nonumber
g(O_{1:K}) = \frac{1}{K!}\sum_{\sigma\in\Sn_K} \prod_{r = 1}^K Z_{\sigma(r),r}.
\end{split}\end{align}
The sample estimator is defined by the estimating equation
\begin{align}\begin{split}\label{basic sample splitting sample VUS}
\Un h_{\widehat p} - \widehat\theta_n \Un g = 0, \quad \text{or} \quad
\widehat\theta_n = \frac{\Un h_{\widehat p}}{\Un g},
\end{split}\end{align}
where
\begin{align}\begin{split}\nonumber
\Un h_{\widehat p}
& \coloneqq \binom{n}{K}^{-1}
  \sum_{1\leq i_1<\cdots<i_K\leq n}
  h_{\widehat p}(O_{i_1},\ldots,O_{i_K}),\\
\Un g
& \coloneqq \binom{n}{K}^{-1}
  \sum_{1\leq i_1<\cdots<i_K\leq n}
  g(O_{i_1},\ldots,O_{i_K}).
\end{split}\end{align}
Conditionally on $\Tcal_n$, both $\Un h_{\widehat p}$ and $\Un g$ that are formed to define the estimating equation are ordinary fixed-kernel U-statistics of degree $K$.  The estimating equation for $\hat\theta_n$ can be rewritten as
\begin{align}\begin{split}\nonumber
\(\Un g\)  \sqrt{n}\(\widehat\theta_n - \theta_0\) = \sqrt{n}\(\Un h_{\widehat p} - \theta_0 \Un g\) = \sqrt{n}\(\Un \(h_{\widehat p} - \theta_0 g\)\).
\end{split}\end{align}
The rest of this section will show that $\Un g \overset{p}{\rightarrow} \E[g(O_{1:K})] = \E \prod_{r=1}^K Z_{r,r} = \prod_{r=1}^K \rho_r$, and that 
\begin{align}\begin{split}\nonumber
   \sqrt{n}\(\Un \(h_{\widehat p} - \theta_0 g\)\) \rightsquigarrow \mathcal{N}(0, \sigma_0^2),
   \quad \text{for some } \sigma_0^2 > 0.
\end{split}\end{align}
It then follows from Slutsky's lemma that
\begin{align}\begin{split}\nonumber
  \sqrt{n}\(\widehat\theta_n - \theta_0\) \rightsquigarrow \mathcal{N}\(0, \frac{\sigma_0^2}{\prod_{r=1}^K \rho_r^2}\).
\end{split}\end{align}
Asymptotic normality of $\sqrt{n}\(\Un \(h_{\widehat p} - \theta_0 g\)\)$  is obtained by separately verifying an asymptotic linear representation in terms of the influence function of the U-statistic for $\Un h_{\widehat p}$ and 
$\Un g$. We break the analysis into several steps.

\subsubsection{Assignment margin assumptions}\label{refined margin assumption}

For $m<\ell$, define the true and estimated pairwise assignment contrasts
\begin{align}\begin{split}\nonumber
\Delta_{m\ell}(x)
\coloneqq A_m(p;x)-A_\ell(p;x),
\quad
\widehat\Delta_{m\ell}(x)
\coloneqq A_m(\widehat p;x)-A_\ell(\widehat p;x).
\end{split}\end{align}
Then
\begin{align}\begin{split}\nonumber
J_{m\ell}(p;x)=\1\{\Delta_{m\ell}(x)\geq0\},
\quad
J_{m\ell}(\widehat p;x)=\1\{\widehat\Delta_{m\ell}(x)\geq0\}.
\end{split}\end{align}
The following assumption parallels the three parts of Assumption 9 in \cite{escanciano2026debiased}, with pairwise assignment contrasts replacing their pairwise score difference.

\begin{assumption}[Assignment comparison and margin]
\label{ass:boundary}
Let
\begin{align}\begin{split}\nonumber
e_n(x) \coloneqq \max_{k\in[K]}\abs{\widehat p_k(x)-p_k(x)},
\quad
\delta_{s,n} \coloneqq \norm{e_n}_{L^s(P)}.
\end{split}\end{align}
The following conditions hold.

\begin{enumerate}[label=(\roman*)]
\item \emph{First-step consistency:}
\begin{align}\begin{split}\nonumber
\delta_{1,n}=o_{\Pp}(1).
\end{split}\end{align}

\item \emph{Stability of exact ties:} with
\begin{align}\begin{split}\label{eq:tie-stability-index}
\zeta_n
\coloneqq \sum_{1\leq m<\ell\leq M}
\E\!\left[
 \1\{\Delta_{m\ell}(X_{1:K})=0\}
 \abs{J_{m\ell}(\widehat p;X_{1:K})-J_{m\ell}(p;X_{1:K})}
 \middle|\Tcal_n
\right],
\end{split}\end{align}
one has $\zeta_n=o_{\Pp}(1)$.

\item \emph{Pairwise assignment margin:} there are constants $\beta>0$, $C_\Delta<\infty$, and $t_0>0$ such that, for every $m<\ell$ and every $0<t\leq t_0$,
\begin{align}\begin{split}\label{eq:pairwise-margin}
\Pp^K\!\left(
0<\abs{\Delta_{m\ell}(X_{1:K})}\leq t
\right)
\leq C_\Delta t^\beta.
\end{split}\end{align}
\end{enumerate}
\end{assumption}

\begin{remark}
Assumption \ref{ass:boundary}(i) is concerned with the convergence of the first step estimator for the conditional probabilities $p_k(x)$.  Parametric, nonparametric and machine learning techniques can be employed to estimate these probabilities. Verifying the convergence property required by 
Assumption \ref{ass:boundary}(i)  depends on the specific form of the estimator for $\hat p_k\(x\)$'s.

Condition \eqref{eq:tie-stability-index} is not a no-tie assumption. The set $\{\Delta_{m\ell}=0\}$ may have positive probability. It only requires the estimated weak comparison to 
preserve the deterministic tie convention asymptotically on that set. 
It holds automatically when exact ties are structural identities that persist after replacing $q$ by $\widehat q$. 
This is the assignment analogue of the treatment of observations with identical covariates in \cite{escanciano2026debiased}.
Their true score difference and fitted score difference are 
both exactly zero on the diagonal. 
If all nontrivial assignment ties have probability zero, condition (ii) is automatic.

On the one hand, \eqref{eq:tie-stability-index} is identically zero under Assumption 9(ii) on page 20 of \cite{escanciano2026debiased}, which holds in turn when all $p_k(X)$'s either have an absolutely continuous distribution or are finitely discrete and injective. 
The finitely discrete propensity scores are injective when each support point corresponds to a unique value of the covariates $X$, ensuring that no two different covariate values map to the same propensity score. 
If there exist distinct covariate values that lead to the same propensity score, injectivity fails and Condition 9(ii) of \cite{escanciano2026debiased} may not hold. 

On the other hand, Condition \eqref{eq:tie-stability-index} is more general than Assumption 9(ii) of \cite{escanciano2026debiased} and can hold even when injectivity fails. 
Note that in \eqref{eq:tie-stability-index}, $J_{m\ell}(p;X_{1:K})=1$ when $\Delta_{m\ell}=0$ just by definition. 
Equation \eqref{eq:tie-stability-index} can then be rewritten as
\begin{align}\begin{split}\nonumber
\zeta_n
\coloneqq \sum_{1\leq m<\ell\leq M}
\E\!\left[
 \1\{\Delta_{m\ell}(X_{1:K})=0\}
 \abs{J_{m\ell}(\widehat p;X_{1:K})-1}
 \middle|\Tcal_n
\right].
\end{split}\end{align}
For example, when the distribution of the propensity score functions $p_k(X)$ has only discrete components, a straightforward Oracle model selection procedure that respects 
the deterministic tie-breaking rule with probability converging to one can produce a post-model selection estimator of 
$\hat p_k(X)$ 's to ensure that Condition \eqref{eq:tie-stability-index} can still hold even if injectivity fails. When the distribution of $p_k(X)$ has both 
an absolutely continuous component and finitely discrete components, 
a density based clustering approach that will be discussed later can also generate a post-selection estimator of 
$\hat p_k(X)$ 's that satisfy  Condition \eqref{eq:tie-stability-index}. 

Assumption \ref{ass:boundary}(iii) can be expected to hold generally. It is difficult to come up with an example where Assumption \ref{ass:boundary}(iii)  can fail. Assumption \ref{ass:boundary}(iii) holds, 
for example, when the distribution of the propensity 
score functions $p_k(X)$ as random variables is a mixture between a smooth absolutely continuous component and a finitely discrete component. It might fail to hold when 
the distribution of $p_k(X)$ has a component that admits a discrete distribution with countably infinite support points. An example of such a component is the inverse of a Poisson random variable. 
\end{remark}

It is possible to adapt a more high-level margin assumption, define the set of true maximizers
\begin{align}\begin{split}\nonumber
\Mcal^*(x) \coloneqq \argmax_{\pi\in\Sn_K}A_\pi(p;x)
\end{split}\end{align}
and the strict suboptimality gap
\begin{align}\begin{split}\nonumber
G_A(x)
\coloneqq \min_{\pi\notin\Mcal^*(x)}
\{A^*(p;x)-A_\pi(p;x)\},
\quad \min\varnothing \coloneqq +\infty.
\end{split}\end{align}
Exact co-optimal ties are omitted from $G_A$. Since there are only finitely many permutations,
\begin{align}\begin{split}\nonumber
\{0<G_A\leq t\}
\subseteq
\bigcup_{m<\ell}
\{0<\abs{\Delta_{m\ell}}\leq t\}.
\end{split}\end{align}
Hence \eqref{eq:pairwise-margin} implies the reader-friendly permutation margin
\begin{align}\begin{split}\nonumber
\Pp^K(0<G_A\leq t)
\leq \binom{M}{2}C_\Delta t^\beta.
\end{split}\end{align}
Alternatively, we can also directly assume that 
\begin{align}\begin{split}\nonumber
\Pp^K(0<G_A\leq t) \leq C_{\Delta}t^{\beta}.
\end{split}\end{align} 
The pairwise form is used below because it permits an almost literal adaptation of Appendix I of \cite{escanciano2026debiased}.


Assume throughout that, for some deterministic $B<\infty$,
\begin{align}\begin{split}\label{eq:bounded-scores}
0\leq p_k(x),\widehat p_k(x)\leq B,
\quad k\in[K].
\end{split}\end{align}
The true propensity scores satisfy this condition with $B=1$. The estimated ratios may be clipped to the same bounded range.

For each permutation set ratios, write 
\begin{align}\begin{split}\nonumber
E_m(x)
\coloneqq \abs{A_m(\widehat p;x)-A_m(p;x)}.
\end{split}\end{align}
For $m \prec \ell$, define the contrast perturbation
\begin{align}\begin{split}\nonumber
B_{m\ell,n}(x)
\coloneqq \abs{\widehat\Delta_{m\ell}(x)-\Delta_{m\ell}(x)}.
\end{split}\end{align}

\begin{lemma}[Product and contrast perturbations]
\label{lem:product}
Under \eqref{eq:bounded-scores},
\begin{align}\begin{split}\label{eq:contrast-perturbation}
E_m(x_{1:K})
\leq B^{K-1}\sum_{r=1}^K e_n(x_r), \quad
B_{m\ell,n}(x_{1:K})
\leq2B^{K-1}\sum_{r=1}^K e_n(x_r).
\end{split}\end{align}
Consequently, for every $s\in[1,\infty]$,
\begin{align}\begin{split}\label{eq:Bml-norm}
\norm{B_{m\ell,n}}_{L^s(P^K)}
\leq L_K\delta_{s,n},
\quad
L_K \coloneqq 2KB^{K-1}.
\end{split}\end{align}
\end{lemma}

\begin{proof}
For $a_r,b_r\in[0,B]$, the telescoping identity
\begin{align}\begin{split}\nonumber
\prod_{r=1}^Ka_r-\prod_{r=1}^Kb_r
=
\sum_{j=1}^K(a_j-b_j)
\left(\prod_{r<j}a_r\right)
\left(\prod_{r>j}b_r\right)
\end{split}\end{align}
implies
\begin{align}\begin{split}\nonumber
\abs{\prod_ra_r-\prod_rb_r}\leq B^{K-1}\sum_j\abs{a_j-b_j}.
\end{split}\end{align}
Apply this to $A_m(\widehat q)-A_m(q)$, 
then the triangle inequality gives \eqref{eq:contrast-perturbation}, and the Minkowski inequality gives \eqref{eq:Bml-norm}.
\end{proof}


\begin{proposition}
\label{prop:selector-consistency}
Under Assumption \ref{ass:boundary} and \eqref{eq:bounded-scores},
\begin{align}\begin{split}\label{eq:comparison-consistency}
\sum_{m<\ell}
\E\!\left[
 \abs{J_{m\ell}(\widehat p;X_{1:K})-J_{m\ell}(p;X_{1:K})}
 \middle|\Tcal_n
\right]
=o_{\Pp}(1).
\end{split}\end{align}
Consequently,
\begin{align}\begin{split}\label{eq:selector-consistency}
\Pp^K\!\left(
\pi_{\widehat p}^{\dagger}(X_{1:K})
\neq
\pi_p^{\dagger}(X_{1:K})
\middle|\Tcal_n
\right)=o_{\Pp}(1).
\end{split}\end{align}
\end{proposition}

\begin{proof}
Fix $m \prec \ell$. If the two weak comparisons disagree and $\Delta_{m\ell}\neq0$, then the true and estimated contrasts lie on opposite sides of zero, or one is zero. Hence
$0<\abs{\Delta_{m\ell}}\leq B_{m\ell,n}$.
For any deterministic proof threshold $t\in(0,t_0]$,
\begin{align}\begin{split}\label{eq:comparison-bound}
&\E\!\left[
 \abs{J_{m\ell}(\widehat p)-J_{m\ell}(p)}
 \middle|\Tcal_n
\right] \\
&\quad\leq
\Pp^K(0<\abs{\Delta_{m\ell}}\leq t)
+\Pp^K(B_{m\ell,n}>t\mid\Tcal_n) 
+ \E\!\left[
 \1\{\Delta_{m\ell}=0\}
 \abs{J_{m\ell}(\widehat p)-J_{m\ell}(p)}
 \middle|\Tcal_n
\right] \\
&\quad\leq
C_\Delta t^\beta
+\frac{L_K\delta_{1,n}}{t}
+\zeta_{m\ell,n},
\end{split}\end{align}
where $\zeta_{m\ell,n}$ is the corresponding summand in \eqref{eq:tie-stability-index}. Choosing, for example,
$t=(L_K\delta_{1,n})^{1/(1+\beta)}$
when this quantity is below $t_0$ gives a bound of order
$\delta_{1,n}^{\beta/(1+\beta)}+\zeta_{m\ell,n}=o_{\Pp}(1)$.
Summing over the finitely many pairs proves \eqref{eq:comparison-consistency}.

If every pairwise comparison is unchanged, the smallest-index maximizer is unchanged. Therefore
\begin{align}\begin{split}\nonumber
\{\pi_{\widehat p}^{\dagger}\neq\pi_p^{\dagger}\}
\subseteq
\bigcup_{m<\ell}
\{J_{m\ell}(\widehat p)\neq J_{m\ell}(p)\}.
\end{split}\end{align}
The union bound and \eqref{eq:comparison-consistency} prove \eqref{eq:selector-consistency}.
\end{proof}

The threshold $t$ in \eqref{eq:comparison-bound} is used only to split the proof. It does not enter the estimator or the assignment algorithm. This is exactly the role played by the sequence $t_n\downarrow0$ in the sign-consistency proof in Appendix I of \cite{escanciano2026debiased}.

\subsubsection{Assignment regret}

Write
\begin{align}\begin{split}\nonumber
\pi_0(x) \coloneqq \pi_p^\dagger(x), 
\quad \widehat\pi(x) \coloneqq \pi_{\widehat p}^\dagger(x),
\end{split}\end{align}
and define the population regret of the estimated assignment by
\begin{align}\begin{split}\nonumber
R_n(x)
\coloneqq A_{\pi_0(x)}(p;x)-A_{\widehat\pi(x)}(p;x)
\geq0.
\end{split}\end{align}
Switching from the tie-broken oracle to another true co-maximizer gives $R_n=0$, so exact ties do not contribute to the population residual.

\begin{proposition}
\label{prop:margin-bound}
Suppose Assumption \ref{ass:boundary}(iii) and \eqref{eq:bounded-scores} hold. For $s\in(1,\infty)$,
\begin{align}\begin{split}\label{eq:regret-rate}
\E[R_n(X_{1:K})\mid\Tcal_n]
\leq
C\,\delta_{s,n}^{\,s(1+\beta)/(s+\beta)}.
\end{split}\end{align}
For the supremum norm,
\begin{align}\begin{split}\label{eq:regret-sup}
\E[R_n(X_{1:K})\mid\Tcal_n]
\leq
C\,\delta_{\infty,n}^{\,1+\beta}.
\end{split}\end{align}
The constants depend only on $K,B,C_\Delta,\beta$, and $s$.
\end{proposition}

\begin{proof}
Condition on $\Tcal_n$, so $\widehat p$ is fixed. Suppose that the oracle selects $\pi_m$, the estimated rule selects $\pi_\ell$, and the regret is positive. Then
\begin{align}\begin{split}\nonumber
A_m(p)>A_\ell(p),
\quad
A_\ell(\widehat p)\geq A_m(\widehat p).
\end{split}\end{align}
Consequently,
$0<\abs{\Delta_{m\ell}} \leq B_{m\ell,n}$,
where for $m \succ \ell$ we interpret $\Delta_{m\ell} \coloneqq -\Delta_{\ell m}$ and $B_{m\ell,n} \coloneqq B_{\ell m,n}$. 
Since only one oracle and one estimated permutation are selected,
\begin{align}\begin{split}\label{eq:regret-pairwise-sum}
R_n
\leq
\sum_{m<\ell}
B_{m\ell,n}
\1\{0<\abs{\Delta_{m\ell}}\leq B_{m\ell,n}\}.
\end{split}\end{align}

Fix one pair and abbreviate $B_n \coloneqq B_{m\ell,n}$ and $\Delta \coloneqq \Delta_{m\ell}$. For an arbitrary $t\in(0,t_0]$, use the same small-margin/large-margin decomposition as in Appendix I of \cite{escanciano2026debiased}:
\begin{align}\begin{split}\nonumber
\E\!\left[
B_n\1\{0<\abs{\Delta}\leq B_n\}
\middle|\Tcal_n
\right] 
\leq
\E[B_n\1\{0<\abs{\Delta}\leq t\}\mid\Tcal_n]
+
\E[B_n\1\{B_n>t\}\mid\Tcal_n].
\end{split}\end{align}
For the first term, H\"older inequality and \eqref{eq:pairwise-margin} give
\begin{align}\begin{split}\label{eq:small-margin}
\E[B_n\1\{0<\abs{\Delta}\leq t\}\mid\Tcal_n]
&\leq
\norm{B_n}_{L^s(P^K)}
\Pp^K(0<\abs{\Delta}\leq t)^{(s-1)/s} \\
&\leq
C\norm{B_n}_{L^s(P^K)}
 t^{\beta(s-1)/s}.
\end{split}\end{align}
For the second term, the tail-moment inequality yields
\begin{align}\begin{split}\label{eq:large-margin}
\E[B_n\1\{B_n>t\}\mid\Tcal_n]
\leq
\frac{\norm{B_n}_{L^s(P^K)}^s}{t^{s-1}}.
\end{split}\end{align}
Combining \eqref{eq:small-margin} and \eqref{eq:large-margin}, then choosing
$t\asymp\norm{B_n}_{L^s(P^K)}^{s/(s+\beta)}$
gives
\begin{align}\begin{split}\nonumber
\E\!\left[
B_n\1\{0<\abs{\Delta}\leq B_n\}
\middle|\Tcal_n
\right]
\leq
C\norm{B_n}_{L^s(P^K)}^{s(1+\beta)/(s+\beta)}.
\end{split}\end{align}
Use \eqref{eq:Bml-norm}, sum over the finitely many pairs in \eqref{eq:regret-pairwise-sum}, and absorb their number into the constant. 
This proves \eqref{eq:regret-rate}.

For $s=\infty$, \eqref{eq:Bml-norm} gives $B_{m\ell,n}\leq L_K\delta_{\infty,n}$. Hence
\begin{align}\begin{split}\nonumber
\E\!\left[
B_{m\ell,n}
\1\{0<\abs{\Delta_{m\ell}}\leq B_{m\ell,n}\}
\middle|\Tcal_n
\right] & \leq
L_K\delta_{\infty,n}
\Pp^K\!\left(
0<\abs{\Delta_{m\ell}} \leq L_K\delta_{\infty,n}
\right) \\ 
& \leq C\delta_{\infty,n}^{1+\beta}.
\end{split}\end{align}
Summing over pairs proves \eqref{eq:regret-sup}.
\end{proof}

The exponent in \eqref{eq:regret-rate},
$\frac{s(1+\beta)}{s+\beta}$,
is exactly the smoothness exponent obtained by \cite{escanciano2026debiased} for a pairwise sign residual. 
The only change is that their score difference is replaced by each pairwise assignment contrast and then summed over the finite set of permutation pairs.

\begin{corollary}
\label{cor:rootn-regret}
For some $s\in(1,\infty)$, suppose
\begin{align}\begin{split}\label{eq:first-step-rate}
\delta_{s,n}
=o_{\Pp}\!\left(
 n^{-(s+\beta)/\{2s(1+\beta)\}}
\right).
\end{split}\end{align}
Then
\begin{align}\begin{split}\nonumber
\sqrt n\,\E[R_n(X_{1:K})\mid\Tcal_n]
=o_{\Pp}(1).
\end{split}\end{align}
For $s=\infty$, it is enough that
\begin{align}\begin{split}\label{eq:first-step-sup}
\delta_{\infty,n}
=o_{\Pp}\!\left(n^{-1/\{2(1+\beta)\}}\right).
\end{split}\end{align}
\end{corollary}

For example, $s=2$ and $\beta=1$ require $\delta_{2,n}=o_{\Pp}(n^{-3/8})$, while a hard margin, formally $\beta=\infty$, yields the familiar $o_{\Pp}(n^{-1/4})$ $L^2$ requirement.

\subsubsection{Asymptotic properties and variance estimation}\label{sample VUS aymptotic properties}

Define the estimated-kernel residual
\begin{align}\begin{split}\nonumber
r_n(O_{1:K})
\coloneqq h_{\widehat p}(O_{1:K})-h_0(O_{1:K}).
\end{split}\end{align}

\begin{lemma}
\label{lem:kernel-residual}
Under the preceding definitions,
\begin{align}\begin{split}\label{eq:mean-residual}
\E[r_n(O_{1:K})\mid\Tcal_n]
=-\frac1{K!}\E[R_n(X_{1:K})\mid\Tcal_n].
\end{split}\end{align}
Under Assumption \ref{ass:boundary} and \eqref{eq:bounded-scores},
\begin{align}\begin{split}\label{eq:l2-residual}
\E[r_n(O_{1:K})^2\mid\Tcal_n]
\leq
\Pp^K(\widehat\pi\neq\pi_0\mid\Tcal_n)
=o_{\Pp}(1).
\end{split}\end{align}
\end{lemma}

\begin{proof}
For the ordered kernel, \eqref{eq:conditional-evaluation} gives
\begin{align}\begin{split}\nonumber
\E[\psi_{\widehat p}-\psi_p\mid X_{1:K},\Tcal_n]
=-\frac1{K!}R_n(X_{1:K}).
\end{split}\end{align}
Taking expectations and using exchangeability to remove the symmetrization proves \eqref{eq:mean-residual}.

If $\widehat\pi=\pi_0$, the ordered kernels agree. Otherwise their absolute difference is at most $1$. Hence
\begin{align}\begin{split}\nonumber
(\psi_{\widehat p}-\psi_p)^2
\leq 
\1\{\widehat\pi\neq\pi_0\}.
\end{split}\end{align}
Jensen inequality for $\Sym$, exchangeability, and Proposition \ref{prop:selector-consistency} prove \eqref{eq:l2-residual}.
\end{proof}


\begin{lemma}
\label{lem:uvar}
Let $f_n$ be a square-integrable symmetric kernel of fixed degree $K$, measurable with respect to $\Tcal_n$ and independent of the evaluation sample. 
Then
\begin{align}\begin{split}\nonumber
\E\!\left[
 n\{\Un f_n-\E[f_n\mid\Tcal_n]\}^2
 \middle|\Tcal_n
\right]
\leq C_K\E[f_n^2\mid\Tcal_n],
\end{split}\end{align}
where $C_K<\infty$ depends only on $K$.
\end{lemma}

\begin{proof}
Condition on $\Tcal_n$. In the covariance expansion of $\Var(\Un f_n)$, terms indexed by disjoint $K$-subsets are independent. The number of ordered pairs of $K$-subsets having a nonempty intersection is $O_K(n^{2K-1})$, while $\binom nK^2$ is of order $n^{2K}$. Cauchy--Schwarz bounds each remaining covariance by a constant multiple of $\E[f_n^2\mid\Tcal_n]$.
\end{proof}

\begin{proposition}
\label{prop:oracle-equivalence}
Suppose Assumption \ref{ass:boundary}, \eqref{eq:bounded-scores}, and the root-$n$ rate \eqref{eq:first-step-rate} hold for some $s\in(1,\infty)$, or suppose \eqref{eq:first-step-sup} holds. Then
\begin{align}\begin{split}\nonumber
\sqrt n\{\Un h_{\widehat p}-\Un h_0\}
=o_{\Pp}(1).
\end{split}\end{align}
\end{proposition}

\begin{proof}
Decompose
\begin{align}\begin{split}\nonumber
\sqrt n\Un r_n = \sqrt n\,\E[r_n\mid\Tcal_n]
+ \sqrt n\{\Un r_n-\E[r_n\mid\Tcal_n]\}.
\end{split}\end{align}
The first term is $o_{\Pp}(1)$ by \eqref{eq:mean-residual} and Corollary \ref{cor:rootn-regret}. 
Lemmas \ref{lem:uvar} and \ref{lem:kernel-residual} imply that the conditional second moment of the second term converges to zero in probability. 
The conditional Markov inequality completes the proof.
\end{proof}

This proof has the same two components as Appendix I of \cite{escanciano2026debiased}. 
Comparison consistency makes the conditional second moment of the estimated-kernel difference vanish. 
The margin bound makes the conditional population residual negligible on the root-$n$ scale.


Define the first Hoeffding projection of the oracle kernel by
\begin{align}\begin{split}\nonumber
g_0(o)
\coloneqq \E[h_0(o,O_2,\ldots,O_K) - g(o,O_2,\ldots,O_K) \theta_0],
\quad
\E[g_0(O)]= 0, 
\end{split}\end{align}
and let
\begin{align}\begin{split}\nonumber
\sigma_0^2 \coloneqq K^2\Var\{g_0(O)\}.
\end{split}\end{align}

\begin{assumption}[Nondegenerate first projection]
\label{ass:nondegenerate}
$0<\Var\{g_0(O)\}<\infty$.
\end{assumption}
The upper bound is automatic because $h_0$ is bounded. Strict positivity excludes a first-order degenerate oracle U-statistic.

\begin{theorem}
\label{thm:clt}
Suppose the conditions of Proposition \ref{prop:oracle-equivalence} and Assumption \ref{ass:nondegenerate} hold. Then
\begin{align}\begin{split}\nonumber
\sqrt n(\widehat\theta_n-\theta_0)
=
\frac{1}{\prod_{r=1}^K \rho_r}
\frac K{\sqrt n}\sum_{i=1}^n
 g_0(O_i)
+o_{\Pp}(1).
\end{split}\end{align}
Consequently,
\begin{align}\begin{split}\label{eq:clt}
\sqrt n(\widehat\theta_n-\theta_0)
\rightsquigarrow N\left(0,
\frac{\sigma_0^2}{\prod_{r=1}^K \rho_r^2}
\right).
\end{split}\end{align}
\end{theorem}

\begin{proof}
The Hoeffding decomposition for the bounded fixed-order oracle kernel gives
\begin{align}\begin{split}\nonumber
\sqrt n(\Un (h_0-g\theta_0))
=\frac K{\sqrt n}\sum_{i=1}^n g_0(O_i)
+o_{\Pp}(1).
\end{split}\end{align}
The terms of Hoeffding order two and higher are $o_{\Pp}(1)$ after multiplication by $\sqrt n$. Proposition \ref{prop:oracle-equivalence} permits replacing $\Un h_0$ by $\Un h_{\widehat p}$, so that 
\begin{align}\begin{split}\nonumber
\sqrt n(\Un (h_{\widehat p}-\theta_0 g))
=\frac K{\sqrt n}\sum_{i=1}^n g_0(O_i)
+o_{\Pp}(1).
\end{split}\end{align}
Next, $g=\Sym\zeta$ is a measurable symmetric kernel of fixed degree $K$
with $0\leq g\leq1$, so $\E|g(O_{1:K})|<\infty$. Since the evaluation
observations are i.i.d., the Hoeffding strong law of large numbers for
U-statistics \citep{hoeffding1961strong} gives
$\Un g \to \E[g(O_{1:K})]$ almost surely.
Symmetrization preserves expectation, and independence of the observations yields
\begin{align}\begin{split}\nonumber
\E[g(O_{1:K})]
=\E\left[\prod_{r=1}^K Z_{r,r}\right]
=\prod_{r=1}^K\E[Z_{r,r}]
=\prod_{r=1}^K\rho_r.
\end{split}\end{align}
Consequently,
\begin{align}\begin{split}\nonumber
\Un g=\prod_{r=1}^K\rho_r+o_{\Pp}(1), \quad 
\frac{1}{\Un g} = \frac{1}{\prod_{r=1}^K\rho_r} + o_{\Pp}(1).
\end{split}\end{align}
Recall that 
$\(\Un g\)  \sqrt{n}\(\widehat\theta_n - \theta_0\) 
= \sqrt{n}\(\Un \(h_{\widehat p} - \theta_0 g\)\)$.
Then we can write
\begin{align}\begin{split}\nonumber
\sqrt{n}\(\widehat\theta_n - \theta_0\) 
& = \frac{\sqrt{n}\(\Un \(h_{\widehat p} - \theta_0 g\)\)}{\Un g} \\
& = \left(\frac{1}{\prod_{r=1}^K\rho_r} + o_{\Pp}(1)\right) \sqrt{n}\(\Un \(h_{\widehat p} - \theta_0 g\)\)\\
& = 
\left(\frac{1}{\prod_{r=1}^K\rho_r} + o_{\Pp}(1)\right) 
\left(\frac K{\sqrt n}\sum_{i=1}^n g_0(O_i)+o_{\Pp}(1)\right)\\
& = \frac{1}{\sqrt n}\sum_{i=1}^n 
\frac{1}{\prod_{r=1}^K\rho_r} 
K g_0(O_i) + o_{\Pp}(1).
\end{split}\end{align}
The central limit theorem 
then gives the asymptotic normality result \eqref{eq:clt}.
\end{proof}

For variance estimation, define the leave-one-out projection estimate
\begin{align}\begin{split}\nonumber
\widehat g_i
\coloneqq \binom{n-1}{K-1}^{-1}
\sum_{\substack{J\subseteq[n]\setminus\{i\}\\|J|=K-1}}
 h_{\widehat p}(O_i,O_J) - \widehat\theta_n g(O_i, O_J)
\end{split}\end{align}
and
\begin{align}\begin{split}\nonumber
\widehat\sigma_n^2
\coloneqq \frac{K^2}{n}
\sum_{i=1}^n(\widehat g_i)^2.
\end{split}\end{align}

\begin{proposition}
\label{prop:variance}
Under Assumption \ref{ass:boundary}, \eqref{eq:bounded-scores}, $\delta_{1,n}=o_{\Pp}(1)$, and Assumption \ref{ass:nondegenerate},
\begin{align}\begin{split}\label{eq:variance-consistency}
\widehat\sigma_n^2 \to_{\Pp}\sigma_0^2.
\end{split}\end{align}
If the root-$n$ rate condition of Theorem \ref{thm:clt} also holds, then
\begin{align}\begin{split}\label{eq:studentized}
\left(\prod_{r=1}^K\hat\rho_r\right)
\frac{\sqrt n(\widehat\theta_n-\theta_0)}{\widehat\sigma_n}
\rightsquigarrow N(0,1).
\end{split}\end{align}
\end{proposition}

\begin{proof}
Define the oracle leave-one-out estimate
\begin{align}\begin{split}\nonumber
\widetilde g_i
\coloneqq \binom{n-1}{K-1}^{-1}
\sum_{\substack{J\subseteq[n]\setminus\{i\}\\|J|=K-1}}
 (h_0(O_i,O_J) - \theta_0 g(O_i, O_J)). 
\end{split}\end{align}
Cauchy-Schwarz inequality gives
\begin{align}\begin{split}\label{eq:leave-one-kernel-bound}
\frac1n\sum_{i=1}^n(\widehat g_i-\widetilde g_i)^2
& \leq \Un \left(\left(h_{\widehat p} - h_0 - (\hat\theta_n - \theta_0) g\right)^2\right) \\ 
& \leq 4 \Un(r_n^2) + 4 (\hat\theta_n - \theta_0)^2 \Un(g^2).
\end{split}\end{align}
The conditional expectation of the right-hand side is $\E[r_n^2\mid\Tcal_n]=o_{\Pp}(1)$ by Lemma \ref{lem:kernel-residual}. 
Furthermore, $\Un g^2 = \E[g^2] + o_{\Pp}(1)$ and  $(\hat\theta_n - \theta_0) = o_{\Pp}(1)$.  Hence \eqref{eq:leave-one-kernel-bound} is $o_{\Pp}(1)$.

Conditional on $O_i$, $\widetilde g_i$ is a degree-$(K-1)$ U-statistic estimating $g_0(O_i)$. 
Boundedness and the fixed-order variance bound imply
\begin{align}\begin{split}\nonumber
\E[(\widetilde g_i-g_0(O_i))^2]\leq C_K/n.
\end{split}\end{align}
Therefore
\begin{align}\begin{split}\label{eq:projection-l2}
\frac1n\sum_{i=1}^n
(\widehat g_i-g_0(O_i))^2=o_{\Pp}(1).
\end{split}\end{align}
The law of large numbers and \eqref{eq:projection-l2}
now yield
\begin{align}\begin{split}\nonumber
\frac1n\sum_{i=1}^n
(\widehat g_i)^2
\to_{\Pp}
\E[(g_0(O))^2].
\end{split}\end{align}
Multiplication by $K^2$ proves \eqref{eq:variance-consistency}. 
Then \eqref{eq:studentized} follows by 
\begin{align}\begin{split}\nonumber
\left(\prod_{r=1}^K\hat\rho_r\right) = \left(\prod_{r=1}^K\rho_r\right) + o_{\Pp}(1),
\end{split}\end{align}
and Slutsky's lemma. 
\end{proof}

To summarize the main results of this section,
the exact-tie-broken robust estimator is
\begin{align}\begin{split}\nonumber
\widehat\theta_n=\frac{\Un h_{\widehat p}}{\Un g},
\quad
h_{\widehat p}
=\Sym\left[
\frac1{K!}
\prod_{r=1}^K
  \1\{C_r=\pi_{\widehat p}^{\dagger}(X_{1:K})(r)\}
\right], \quad
g = S \prod_{r=1}^K Z_{r,r}.
\end{split}\end{align}
There is no assignment tolerance. 
The first-step score enters only through the exact maximum-weight permutation.

The pairwise margin assumption
\begin{align}\begin{split}\nonumber
\Pp^K(0<\abs{\Delta_{m\ell}}\leq t)
\leq Ct^\beta
\end{split}\end{align}
and the product perturbation bound
\begin{align}\begin{split}\nonumber
\norm{\widehat\Delta_{m\ell}-\Delta_{m\ell}}_s
\lesssim
\norm{\max_k\abs{\widehat p_k-p_k}}_{L^s(P)}
\end{split}\end{align}
lead to the Appendix-I-type residual estimate
\begin{align}\begin{split}\nonumber
\E[R_n\mid\Tcal_n]
\lesssim
\delta_{s,n}^{s(1+\beta)/(s+\beta)}.
\end{split}\end{align}
Exact co-optimal ties contribute zero regret. 
They are permitted in the sample theory provided the deterministic strict/weak tie convention is asymptotically stable on the exact-tie set. 
This replaces a no-tie assumption and is the direct analogue of preserving the diagonal ties in \cite{escanciano2026debiased}.

The residual bound controls the conditional mean of the estimated kernel, while comparison consistency controls its conditional $L^2$ distance from the oracle kernel. 
Together they imply
\begin{align}\begin{split}\nonumber
\sqrt n\{\Un h_{\widehat p}-\Un h_0\}=o_{\Pp}(1),
\end{split}\end{align}
so standard fixed-degree U-statistic theory gives
\begin{align}\begin{split}\nonumber
\sqrt n(\widehat\theta_n-\theta_0) = 
\frac{K}{\prod_{r=1}^K \rho_r} \frac{1}{\sqrt n}\sum_{i=1}^n
\{g_0(O_i)\}+o_{\Pp}(1).
\end{split}\end{align}

\subsection{Sample splitting}\label{sample-splitting}

This subsection provides a simple $K$th-order U-statistics sample splitting scheme. 
This scheme is not balanced. 
Let $\left(O_1,\ldots,O_n\right)$ be the sample in its original index order, with $[n]=\{1,\ldots,n\}$. 
Choose $2 \leq K < T \leq n$. 
First suppose $n=Tm$ for an integer $m$ and partition the indices into ordered intervals
\begin{align}\begin{split}\nonumber
G_s \coloneqq (b_s, \ldots, b_s+m-1), \quad b_s \coloneqq 1+(s-1)m, \quad s=1,\ldots,T.
\end{split}\end{align}
An upper-triangular cell has block labels $\vs=(s_1, \ldots, s_K)$ with $s_1 \leq \cdots \leq s_K$. Its starts, lengths, and endpoints are
\begin{align}\begin{split}\nonumber
\va \coloneqq (b_{s_r})_{r=1}^K, \quad
\vl \coloneqq (|G_{s_r}|)_{r=1}^K, \quad e_r \coloneqq a_r+\ell_r-1.
\end{split}\end{align}
Its test tuples, occupied union, and training complement are
\begin{align}
\mathcal{T}(\va,\vl) & \coloneqq \{(i_1,\ldots,i_K):a_r\le i_r\le e_r,\ i_1<\cdots<i_K\},\label{eq:cell}\\
\mathcal{V}(\va,\vl) & \coloneqq \bigcup_{r=1}^K\{a_r,\ldots,e_r\}, \quad \mathcal{R}(\va,\vl) \coloneqq [n]\setminus V(\va,\vl).\label{eq:train}
\end{align}

Algorithm \ref{alg:loops} below first selects a cell, then collects its increasing index tuples. 
The outer for-loops permit repeated block labels, while the inner for-loops exclude repeated sample indices. 
Each layer has $K$ nested loops. 
For $K=2$, omit the intermediate loops. 
Write $\vi \coloneqq (i_{1}, \ldots, i_{K})$. 
For each cell, use $(O_{\vi})$, $\vi \in \mathcal{T}$ for testing and $(O_{\vi})$, $\vi \in \mathcal{R}$ for training. 

Generally, write $n=Tm+q$, with $m=\lfloor n/T\rfloor$ and $0 \leq q<T$. 
Use
\begin{align}\begin{split}\nonumber
m_s=m+\mathbf1\{s\le q\},\quad
b_s=1+(s-1)m+\min(s-1,q),\quad
G_s=(b_s, \ldots, b_s+m_s-1).
\end{split}\end{align}
The Algorithm \ref{alg:loops} keeps unchanged. 

\begin{algorithm}[htbp]
\caption{For-loops version sample splitting scheme}\label{alg:loops}
\small
\begin{algorithmic}[1]
\Require Ordered intervals $G_1,\ldots,G_T$ with starts $b_s$, kernel order $K$.
\For{$s_1=1,\ldots,T$} 
 \Statex $\vdots$\quad $\mathbf{for}\ s_r=s_{r-1},\ldots,T$, nested for $r=2,\ldots,K-1$
 \For{$s_K=s_{K-1},\ldots,T$}
  \State Set $\va \gets (b_{s_r})_{r=1}^K$, $\vl\gets(|G_{s_r}|)_{r=1}^K$, and $e_r\gets a_r+\ell_r-1$.
  \State Fix $\mathcal{R} \gets [n]\setminus\bigcup_{r=1}^K G_{s_r}$, and initialize $\mathcal{T} \gets [\,]$.
  \For{$i_1=a_1,\ldots,e_1$} 
   \Statex \hspace{3em}$\vdots$\quad $\mathbf{for}\ i_r=\max\{a_r, i_{r-1}+1\}, \ldots, e_r$, nested for $r=2, \ldots, K-1$
   \For{$i_K=\max\{a_K,i_{K-1}+1\},\ldots,e_K$}
    \State Append $(i_1, \ldots, i_K)$ to $\mathcal{T}$.
   \EndFor
   \Statex \hspace{3em}$\vdots$\quad close the intermediate index loops
  \EndFor
  \If{$\mathcal{T} \ne [\,]$}
   \State Output this cell's test indices $\mathcal{T}$ and training indices $\mathcal{R}$.
  \EndIf
 \EndFor
 \Statex $\vdots$\quad close the intermediate cell loops
\EndFor
\end{algorithmic}
\end{algorithm}

Algorithm \ref{alg:cell} replaces the outer and inner $K$ nested for-loops by enumerations. 
It visits exactly the same nonempty cells and test tuples, in the same order, with the same training complements. 
For each cell and its parameter $\va = (a_{1}, \ldots, a_{K})$, precompute the coordinate bounds
\begin{align}\begin{split}\label{eq:bounds}
v_r=\max_{1\le q\le r}\{a_q+r-q\}, \quad
u_r=\min_{r\le q\le K}\{e_q-(q-r)\}, \quad 
r = 1, \ldots, K.
\end{split}\end{align}
Alternatively, we can also compute recursively by 
\begin{align}\begin{split}\nonumber
v_1 & = a_1,\quad v_r=\max\{a_r, v_{r-1}+1\} \\ 
u_K & = e_K,\quad u_r=\min\{e_r, u_{r+1}-1\}.
\end{split}\end{align}

Here $\vv=(v_r)_{r=1}^K$ and $\vu=(u_r)_{r=1}^K$ account for strict increase. 
The cell is empty if some $v_r>u_r$, otherwise $\vv$ and $\vu$ are its first and last tuples.
\begin{algorithm}[htbp]
\caption{Successor-rule version sample splitting scheme}\label{alg:cell}
\small
\begin{algorithmic}[1]
\Require Ordered intervals $G_1,\ldots,G_T$ with starts $b_s$, kernel order $K$.
\State Initialize $\vs\gets(1, \ldots, 1)$.
\While{true} 
 \State Set $\va\gets(b_{s_r})_{r=1}^K$, $\vl\gets(|G_{s_r}|)_{r=1}^K$, and compute $\vv,\vu$ by \eqref{eq:bounds}.
 \If{$v_r\le u_r$ for every $r$}
  \State Fix $\mathcal R\gets[n]\setminus\bigcup_{r=1}^K G_{s_r}$, set $\mathcal T\gets[\,]$ and $\vi\gets\vv$.
  \While{true} 
   \State Append $\vi$ to $\mathcal T$.
   \State \textbf{if} $\vi=\vu$, \textbf{break} the inner loop.
   \State Set $j\gets\max\{r:i_r<u_r\}$, then $c\gets i_j+1$.
   \State Replace $(i_j,\ldots,i_K)$ by $(\max\{v_r,c+r-j\})_{r=j}^K$.
  \EndWhile
  \State Output this cell's test indices $\mathcal{T}$ and training indices $\mathcal{R}$.
 \EndIf
 \State \textbf{if} $\vs=(T,\ldots,T)$, \textbf{stop}.
 \State Set $j\gets\max\{r:s_r<T\}$, then $c\gets s_j+1$.
 \State Replace $(s_j,\ldots,s_K)$ by $(c,\ldots,c)$.
\EndWhile
\end{algorithmic}
\end{algorithm}

In the inner successor loop, $j$ is the rightmost coordinate that can increase and $c=i_j+1$. 
For $r\ge j$, the new value must be at least $c+r-j$ to preserve strict increase and at least $v_r$ to respect the cell. 
Their maximum gives the smallest admissible suffix. 
In the outer loop, the rightmost coordinate is computed in the same way as above. 
The suffix is then replaced accordingly keeping the nondecreasing order cell index.

Using the notations adapted in this subsection, we obtain an efficient cross validation VUS estimator by taking a weighted average over individual estimates defined on each fold of testing cells:
\begin{align}\begin{split}\nonumber
  \widehat{\theta}_{cv, n} \coloneqq \frac{1}{\sum_{(\va, \vl) \in \mathcal{F}} 
\vert\mathcal T(\va,\vl)\vert
  }\sum_{(\va, \vl) \in \mathcal{F}} \widehat{\theta}_{(\va, \vl)}\vert\mathcal T(\va,\vl)\vert, 
\end{split}\end{align}
where $\mathcal{F}$ consists of $(\va, \vl)$ such that $\mathcal{T}(\va, \vl) \neq \emptyset$ and $\mathcal{R}(\va, \vl) \neq \emptyset$, and 
$\widehat{\theta}_{(\va, \vl)}$ is calculated by \eqref{basic sample splitting sample VUS} with test set $\mathcal{T}(\va, \vl)$ and training set $\mathcal{R}(\va, \vl)$,  
where the sample size $n$ in \eqref{basic sample splitting sample VUS}  is now replaced by the size of each test cell $\vert\mathcal T(\va,\vl)\vert$.
By following the same arguments as in \cite{chernozhukov2018double}, it is not hard to see that all properties derived in 
section \ref{sample VUS aymptotic properties} apply to $\widehat{\theta}_{cv, n}$ with a larger effective sample size $n$. Therefore, Theorem \ref{thm:clt} holds 
with $\widehat{\theta}_{cv, n}$ replacing $\widehat\theta_n$. 
Taking weighted average over cell estimates is needed to account for the different
number of $K$-tuples over different cells. 

As also suggested in \cite{chernozhukov2018double}, instead of averaging individual estimates defined over each 
fold, an alternative double/debiased cross-validation estimator can be obtained by constructing a single moment condition
throughout the $K$th-order U-statistic using the entire sample. In particular, analogous to \eqref{basic sample splitting sample VUS}, we can define the double/debiased cross-validation estimator as
\begin{align}\begin{split}\nonumber
\Un h_{\widehat p} - \widehat\theta_{cv,n} \Un g = 0, \quad \text{or} \quad
\widehat\theta_{cv,n} = \frac{\Un h_{\widehat p}}{\Un g},
\end{split}\end{align}
where we now define instead
\begin{align}\begin{split}\nonumber
\Un h_{\widehat p}
& \coloneqq \binom{n}{K}^{-1}
\sum_{(\va, \vl)\in \mathcal F} \sum_{(i_1,\ldots,i_K)\in \mathcal T(\va, \vl)}
  h_{\widehat p_{\mathcal R(\va, \vl)}}(O_{i_1},\ldots,O_{i_K}),\\
\Un g
& \coloneqq \binom{n}{K}^{-1}
  \sum_{1\leq i_1<\cdots<i_K\leq n}
  g(O_{i_1},\ldots,O_{i_K}).
\end{split}\end{align}
In the above, $\mathcal R(\va, \vl)$ denotes the training set corresponding to the fold $(\va, \vl)$, and $\mathcal T(\va, \vl)$ denotes the test set corresponding to the same fold. The function $h_{\widehat p_{\mathcal R(\va, \vl)}}$ is evaluated using the estimator $\widehat p$ trained on $\mathcal R(\va, \vl)$.

\begin{remark}
Note that in Algorithm \ref{alg:loops} and Algorithm \ref{alg:cell}, the training step depends only 
on $\mathcal{R} \gets [n]\setminus\bigcup_{r=1}^K G_{s_r}$. It is possible that different cell interval collections 
of $\{G_{s_r}, r=1,\ldots,K\}$ may lead to the same union $\bigcup_{r=1}^K G_{s_r}$ and the same
training set $\mathcal{R}$. Computation time can be reduced by avoiding redundant training on the same $\mathcal{R}$. This can be implemented by storing previously trained models, and checking against them.   A previously trained model can be
retrieved from the storage if the same training set $\mathcal{R}$ has been encountered before. Otherwise a newly trained model is stored for future reference.
\end{remark}

\section{Further Applications}\label{more examples section}

\subsection{Aumann expectation of random segments}\label{subsec:zonoid-determinants}
Suppose $X\sim P$ and $a=(a_1,\ldots,a_K)^\top:\mathcal X\to\R^K$ is measurable with $\E_P[\|a(X)\|]<\infty$.
Define the random segment
\begin{align}\begin{split}\nonumber
D_a(X)\coloneqq[0,a(X)]=\{t a(X):0\leq t\leq1\}.
\end{split}\end{align}
Let $X_1,\ldots,X_K$ be independent copies of $X$.
\citet[Theorem 4.1]{vitale1991expectation} and \citet[Theorem 3.2]{vitale1991determinants} establish the following identity:
\begin{align}\begin{split}\label{eq:zonoid-preview}
\Vol_K\bigl(\E_A[D_a(X)]\bigr)
=\frac{1}{K!}\E_{P^K}\left[
\left|\det\bigl(a(X_1),\ldots,a(X_K)\bigr)\right|
\right].
\end{split}\end{align}
On the left, $\E_A[D_a(X)]$ is the Aumann expectation of the random segment $D_a(X)$, as defined in Definition~\ref{def:aumann-expectation}.
On the right, the matrix $\bigl(a_k(X_r)\bigr)_{k,r=1}^K$ has columns $a(X_1),\ldots,a(X_K)$, which are independent copies of $a(X)$.
The absolute determinant is the $K$-dimensional volume of the parallelepiped generated by these columns.
Equation \eqref{eq:zonoid-preview} is the segment counterpart of the ROC volume formula \eqref{eq:main-volume-formula} in Theorem~\ref{thm:volume-formula}. 

\begin{definition}[Zonotope and zonoid]
A set $Z \subseteq \mathbb{R}^d$ is called a \emph{zonotope}
if it is a finite Minkowski sum of line segments. Equivalently,
there exist $m \geq 1$ and vectors $c,v_1,\ldots,v_m \in \mathbb{R}^d$
such that
\begin{align}\begin{split}\nonumber
Z = c + \sum_{i=1}^{m}[-v_i,v_i]
  = \left\{
      c + \sum_{i=1}^{m} t_i v_i :
      t_i \in [-1,1],\ i=1,\ldots,m
    \right\}.
\end{split}\end{align}
A nonempty compact convex set $K \subseteq \mathbb{R}^d$
is called a \emph{zonoid} if there exists a sequence of zonotopes
$(Z_n)_{n \geq 1}$ such that $d_H(Z_n,K) \to 0$. 
\end{definition}

\citet[Theorem 3.2, p. 295]{vitale1991determinants} derives this identity from the zonotope volume formula and strong laws for random sets and $U$-statistics.
We next prove this identity using the Aumann-expectation and mixed-volume method developed in Section~\ref{sec:roc-geometric-representation}.

Here the set whose volume we seek is already given as $\E_A[D_a(X)]$.
We first establish the selection characterization for segments, following Lemma~\ref{lem:coordinate-simplex-selections}, and a stability bound using the paired selections from the proof of Lemma~\ref{lem:finite-simplex-approximation}.
The allocation rule $\phi=(\phi_1,\phi_2)$ assigns probabilities to the endpoints $0$ and $a(x)$.

\begin{lemma}
\label{lem:zonoid-aumann}
The integrable measurable selections of $x\mapsto D_a(x)$ are precisely the functions $f_\phi(x)=\phi_2(x)a(x)$, up to $P$-almost-sure equality, where $\phi:\mathcal X\to\Delta^1$ is a measurable allocation rule.
The Aumann expectation $\E_A[D_a(X)]$ is compact and convex, and
\begin{align}\begin{split}\label{eq:zonoid-aumann}
\E_A[D_a(X)]
&=\int_{\mathcal X}D_a(x)\dd P(x)\\
&=\left\{\E_P[\phi_2(X)a(X)]:\phi:\mathcal X\to\Delta^1
\text{ is measurable}\right\}.
\end{split}\end{align}
For any measurable map $b:\mathcal X\to\R^K$ with $\E_P[\|b(X)\|]<\infty$, write $D_b(x)=[0,b(x)]$. Then
\begin{align}\begin{split}\label{eq:zonoid-stability}
d_H(\E_A[D_a(X)],\E_A[D_b(X)])
\leq\E_P[\|a(X)-b(X)\|].
\end{split}\end{align}
\end{lemma}

\begin{proof}\ 
The map $v\mapsto[0,v]$ is continuous in Hausdorff distance, since
$d_H([0,v],[0,w])\leq\|v-w\|$.
Thus $D_a$ is measurable, and it is integrably bounded because
$\sup_{z\in D_a(x)}\|z\|=\|a(x)\|$.
Its values are nonempty, compact, and convex.
Every such allocation rule defines a measurable selection $f_\phi=\phi_2a$, and $\|f_\phi(x)\|\leq\|a(x)\|$ makes it integrable.
Conversely, an integrable measurable selection $f(x)\in[0,a(x)]$ has the measurable representation
\begin{align}\begin{split}\nonumber
\phi_2(x)=
\begin{cases}
\langle f(x),a(x)\rangle/\|a(x)\|^2,&a(x)\neq0,\\
0,&a(x)=0,
\end{cases}
\quad f(x)=\phi_2(x)a(x),\quad0\leq\phi_2(x)\leq1,
\end{split}\end{align}
after assigning zero values on any exceptional null set.
Setting $\phi_1=1-\phi_2$ gives the corresponding measurable allocation rule $\phi=(\phi_1,\phi_2)$.
Definition \ref{def:aumann-expectation} and Theorem \ref{thm:aumann-compactness} give \eqref{eq:zonoid-aumann}, compactness, and convexity.
For stability, pairing the same allocation rule for $a$ and $b$ gives
\begin{align}\begin{split}\nonumber
\left\|\E_P[\phi_2(X)a(X)]-\E_P[\phi_2(X)b(X)]\right\|
\leq\E_P[\|a(X)-b(X)\|].
\end{split}\end{align}
Every point in either Aumann expectation has such a paired point in the other, which proves the Hausdorff bound.
\end{proof}

Using the selection representation in Lemma~\ref{lem:zonoid-aumann}, we next compute the volume of $\E_A[D_a(X)]$ when $a(X)$ takes finitely many values.
As in Lemma~\ref{lem:finite-coordinate-simplices}, we write this expectation as a weighted Minkowski sum and apply the volume polynomial in Theorem~\ref{thm:minkowski-volume} to express its volume as an expected mixed volume.

\begin{lemma}
\label{lem:zonoid-finite}
If $a(X)$ takes the distinct values $a_1,\ldots,a_M$ with probabilities $\omega_1,\ldots,\omega_M>0$.
Then
\begin{align}\begin{split}\label{eq:zonoid-finite-minkowski}
\E_A[D_a(X)]=\sum_{m=1}^M\omega_m[0,a_m],
\end{split}\end{align}
and
\begin{align}\begin{split}\nonumber
\Vol_K(\E_A[D_a(X)])
&=\E_{P^K}[V(D_a(X_1),\ldots,D_a(X_K))].
\end{split}\end{align}
\end{lemma}

\begin{proof}\ 
Take any $z\in\E_A[D_a(X)]$.
By Lemma~\ref{lem:zonoid-aumann}, there is a measurable allocation rule $\phi$ such that $z=\E_P[\phi_2(X)a(X)]$.
For this rule, put 
\begin{align}\begin{split}\nonumber
t_m=\E_P[\phi_2(X)\mid a(X)=a_m]\in[0,1]. 
\end{split}\end{align}
Conditioning on the values of $a(X)$ gives
\begin{align}\begin{split}\nonumber
z=\E_P[\phi_2(X)a(X)]=\sum_{m=1}^M\omega_m t_m a_m
\in\sum_{m=1}^M\omega_m[0,a_m].
\end{split}\end{align}
Conversely, take any $z\in\sum_{m=1}^M\omega_m[0,a_m]$.
There exist $t_1,\ldots,t_M\in[0,1]$ such that
\begin{align}\begin{split}\nonumber
z=\sum_{m=1}^M\omega_m t_m a_m.
\end{split}\end{align}
Define
\begin{align}\begin{split}\nonumber
\phi_2(x)=\sum_{m=1}^M t_m\1\{a(x)=a_m\},
\quad \phi_1(x)=1-\phi_2(x).
\end{split}\end{align}
The events $\{a(x)=a_m\}$ are measurable and disjoint, so $\phi=(\phi_1,\phi_2)$ is a measurable allocation rule.
Its expected selection satisfies
\begin{align}\begin{split}\nonumber
\E_P[\phi_2(X)a(X)]
=\sum_{m=1}^M\omega_m t_m a_m=z.
\end{split}\end{align}
Thus $z\in\E_A[D_a(X)]$ by Lemma~\ref{lem:zonoid-aumann}.
This proves the reverse inclusion and hence \eqref{eq:zonoid-finite-minkowski}. 
Taking volumes in \eqref{eq:zonoid-finite-minkowski} gives 
\begin{align}\begin{split}\nonumber
\Vol_K(\E_A[D_a(X)])
&=\Vol_K\left(\sum_{m=1}^M\omega_m[0,a_m]\right)\\
&=\sum_{i_1,\ldots,i_K=1}^M
\left(\prod_{r=1}^K\omega_{i_r}\right)
V([0,a_{i_1}],\ldots,[0,a_{i_K}])\\
&=\E_{P^K}[V(D_a(X_1),\ldots,D_a(X_K))].
\end{split}\end{align}
Note that applying the volume polynomial \eqref{eq:mixed-volume-polynomial} to this weighted sum gives the second equality.
Independence makes the probability of each ordered tuple $(a(X_1),\ldots,a(X_K))=(a_{i_1},\ldots,a_{i_K})$ equal to $\prod_r\omega_{i_r}$, giving the third equality.
\end{proof}

We next repeat the approximation step of Lemma~\ref{lem:finite-simplex-approximation}.
Here $a$ may be unbounded and have signed coordinates, so we truncate and round toward zero.
The same Hausdorff-continuity and dominated-convergence argument then applies with an integrable bound.

\begin{lemma}
\label{lem:zonoid-approximation}
There exist finite-valued measurable maps $a^{(n)}:\mathcal X\to\R^K$ such that
\begin{align}\begin{split}\label{eq:zonoid-l1-approximation}
\|a^{(n)}(x)\|\leq\|a(x)\|,
\quad
\E_P[\|a^{(n)}(X)-a(X)\|] \to 0.
\end{split}\end{align}
Writing $D_{a^{(n)}}(x)=[0,a^{(n)}(x)]$, we have
\begin{align}
\Vol_K(\E_A[D_{a^{(n)}}(X)])& \to \Vol_K(\E_A[D_a(X)]),
\label{eq:zonoid-volume-limit}\\
\E_{P^K}[V(D_{a^{(n)}}(X_1),\ldots,D_{a^{(n)}}(X_K))]
& \to \E_{P^K}[V(D_a(X_1),\ldots,D_a(X_K))].
\label{eq:zonoid-mixed-limit}
\end{align}
\end{lemma}

\begin{proof} 
Truncate each coordinate and round toward zero: for $k=1,\ldots,K$, set
\begin{align}\begin{split}\nonumber
a_k^{(n)}(x)
\coloneqq\frac{\operatorname{sgn}(a_k(x))}{n}
\left\lfloor n\min\{|a_k(x)|,n\}\right\rfloor.
\end{split}\end{align}
Each coordinate belongs to the finite grid $\{j/n:j=-n^2,\ldots,n^2\}$.
Moreover, $a^{(n)}(x)\to a(x)$ and $|a_k^{(n)}(x)|\leq|a_k(x)|$, with the same sign whenever nonzero.
Thus $\|a^{(n)}(x)\|\leq\|a(x)\|$ and $\|a^{(n)}(x)-a(x)\|\leq\|a(x)\|$.
Dominated convergence proves \eqref{eq:zonoid-l1-approximation}.
Taking $b=a^{(n)}$ in the stability bound \eqref{eq:zonoid-stability} gives
\begin{align}\begin{split}\nonumber
d_H(\E_A[D_{a^{(n)}}(X)],\E_A[D_a(X)])
\leq\E_P[\|a^{(n)}(X)-a(X)\|] \to 0,
\end{split}\end{align}
so Proposition \ref{prop:hausdorff-continuity} proves \eqref{eq:zonoid-volume-limit}.

For each $x$, 
\begin{align}\begin{split}\nonumber
d_H(D_{a^{(n)}}(x),D_a(x))\leq\|a^{(n)}(x)-a(x)\| \to 0. 
\end{split}\end{align}
Proposition \ref{prop:mixed-volume-continuity} therefore gives pointwise convergence of the mixed volumes in \eqref{eq:zonoid-mixed-limit}.
The approximating mixed volumes are measurable and finite-valued, so their pointwise limit is measurable.
Let $B_K$ be the closed Euclidean unit ball in $\R^K$.
Both $D_{a^{(n)}}(x)$ and $D_a(x)$ lie in $\|a(x)\|B_K$.
Monotonicity, homogeneity, and normalization of mixed volume give the common bound
\begin{align}\begin{split}\label{eq:zonoid-mixed-bound}
0\leq V(D_{a^{(n)}}(X_1),\ldots,D_{a^{(n)}}(X_K))
\leq\Vol_K(B_K)\prod_{r=1}^K\|a(X_r)\|.
\end{split}\end{align}
Independence and $\E_P[\|a(X)\|]<\infty$ imply
\begin{align}\begin{split}\nonumber
\E_{P^K}\left[\prod_{r=1}^K\|a(X_r)\|\right]
=\left(\E_P[\|a(X)\|]\right)^K<\infty.
\end{split}\end{align}
Dominated convergence now proves \eqref{eq:zonoid-mixed-limit}.
\end{proof}

Combining Lemmas~\ref{lem:zonoid-finite} and \ref{lem:zonoid-approximation} now repeats the limiting step in the proof of Theorem~\ref{prop:polarization}:
\begin{align}\begin{split}\label{eq:zonoid-expected-mixed}
\Vol_K\bigl(\E_A[D_a(X)]\bigr)
&=\lim_{n\to\infty}\Vol_K\bigl(\E_A[D_{a^{(n)}}(X)]\bigr)\\
&=\lim_{n\to\infty}\E_{P^K}\!\left[V(D_{a^{(n)}}(X_1),\ldots,D_{a^{(n)}}(X_K))\right]\\
&=\E_{P^K}\!\left[V(D_a(X_1),\ldots,D_a(X_K))\right].
\end{split}\end{align}
Thus the volume calculation has reached the same expected-mixed-volume representation as in Section~\ref{sec:roc-geometric-representation}.
Each approximating set $\E_A[D_{a^{(n)}}(X)]$ is a finite Minkowski sum of segments, or a zonotope, so its Hausdorff limit $\E_A[D_a(X)]$ is a zonoid.

To evaluate the integrand in \eqref{eq:zonoid-expected-mixed}, fix $a_1,\ldots,a_K\in\R^K$.
Every Minkowski sum of fewer than $K$ of the segments $[0,a_r]$ has zero $K$-dimensional volume.
Consequently, only the term $I=\{1,\ldots,K\}$ remains in the polarization formula \eqref{eq:polytopal-mixed-volume} of Definition~\ref{def:mixed-volume}, giving
\begin{align}\begin{split}\label{eq:zonoid-segment-mixed}
V([0,a_1],\ldots,[0,a_K])
&=\frac{1}{K!}\Vol_K\left(\sum_{r=1}^K[0,a_r]\right)\\
&=\frac{1}{K!}|\det(a_1,\ldots,a_K)|.
\end{split}\end{align}
For the last equality, write $A=(a_1,\ldots,a_K)$ with these vectors as its columns.
The full sum is $A[0,1]^K$, whose volume is $|\det A|$ by the linear change-of-variables formula.\footnote{
For an invertible $K\times K$ matrix $A$ and a measurable set $B\subseteq\R^K$, the formula gives $\Vol_K(AB)=\int_B|\det A|\dd x=|\det A|\Vol_K(B)$, where $AB=\{Ax:x\in B\}$.
Taking $B=[0,1]^K$, whose volume is one, yields the stated equality.
If $A$ is singular, $A[0,1]^K$ lies in a proper linear subspace, so both its $K$-dimensional volume and $|\det A|$ are zero.}
Applying \eqref{eq:zonoid-segment-mixed} to $a(X_1),\ldots,a(X_K)$ and taking expectations gives
\begin{align}\begin{split}\nonumber
\E_{P^K}\!\left[V(D_a(X_1),\ldots,D_a(X_K))\right]
=\frac{1}{K!}\E_{P^K}\!\left[|\det(a(X_1),\ldots,a(X_K))|\right].
\end{split}\end{align}
The bound in \eqref{eq:zonoid-mixed-bound} ensures that these expectations are finite.
Together with \eqref{eq:zonoid-expected-mixed}, this gives the following proposition.

\begin{proposition}
\label{prop:zonoid-vitale}
Under the assumptions above,
\begin{align}\begin{split}\label{eq:zonoid-vitale}
\Vol_K\bigl(\E_A[D_a(X)]\bigr)
&=\E_{P^K}[V(D_a(X_1),\ldots,D_a(X_K))]\\
&=\frac{1}{K!}\E_{P^K}\left[|\det(a(X_1),\ldots,a(X_K))|\right].
\end{split}\end{align}
All expectations are finite, including when $\E_A[D_a(X)]$ has zero $K$-dimensional volume.
\end{proposition}

\citet[p.~877]{koshevoy1996lorenz} use \eqref{eq:zonoid-vitale} to compute Lorenz-zonoid volumes in their study of multivariate inequality.
\citet[Sections 4 and 5.5]{koshevoy1998lift} apply \eqref{eq:zonoid-vitale} to lift zonoids to bound expected random convex hull volumes and compare expected absolute determinants under the lift-zonoid order.

\subsection{Feasible error set in fairness-accuracy algorithm design}
\label{subsec:fa-application}

\citet[Sections 2.2--2.3]{liang2026algorithm} study preferences over the group errors produced by an algorithm.
Here an algorithm is a randomized allocation rule that uses the observed features to assign probabilities to the available decisions, as in Definition~\ref{def:allocation-rule}.
Each rule produces a group-error vector, whose coordinates are the expected losses within the respective groups.
Accuracy concerns reducing these losses, while fairness concerns reducing their disparity across groups.
They define fairness--accuracy (FA) preferences through the following comparison of two algorithms.
If the first has no higher error for any group and no greater disparity than the second, and at least one group error or the disparity is strictly lower, it must be strictly preferred.
The designer chooses a rule whose group-error vector is ranked highest under the given FA preference.

In this part, we consider binary decisions and use the associated losses to characterize the attainable group-loss set and compute its volume. 

Retain $X\sim P$ from Section~\ref{subsec:allocation-roc-body}, and let $Y$ denote the outcome in a measurable space.
With $K=2$ decisions, the allocation rule $\phi=(\phi_1,\phi_2)$ has measurable components satisfying $\phi_1(X),\phi_2(X)\geq0$ and $\phi_1(X)+\phi_2(X)=1$ almost surely.
Here $\phi_j(X)$ is the probability of decision $j\in\{1,2\}$ conditional on $X$.
Let $G\geq2$ denote the number of groups, each with positive population share.
As in Section~\ref{subsec:allocation-roc-body}, for each $g\in[G]$, let $P_g$ denote the feature distribution within group $g$ and define the group density ratio $q_g\coloneqq\dd P_g/\dd P$.
Let $\ell$ be a bounded measurable loss function, with $\ell(j,Y)$ denoting the loss from decision $j$ when the outcome is $Y$.
Write $\ell_g(j,x)$ for the conditional expected loss from decision $j$ for an individual in group $g$ with features $X=x$.\footnote{Set $\ell_g(j,x)=0$ when $q_g(x)=0$, since these feature values do not affect group $g$'s expected loss.}
The expected loss in group $g$ is
\begin{align}\begin{split}
T_g(\phi)
\coloneqq\E_{P_g}\!\left[\phi_1(X)\ell_g(1,X)+\phi_2(X)\ell_g(2,X)\right].
\end{split}\end{align}
The resulting group-error vector is $T(\phi)=(T_1(\phi),\ldots,T_G(\phi))^\top\in\R^G$, and its attainable values form the feasible error set
\begin{align}\begin{split}
\Gamma\coloneqq\{T(\phi):\phi\text{ is an allocation rule}\}.
\end{split}\end{align}

\noindent\textit{Notation.}
In this subsection, $T(\phi)\in\R^G$ records expected losses for $G$ groups, so smaller coordinates mean better accuracy for the corresponding groups.
In Section~\ref{subsec:allocation-roc-body}, $T(\phi)\in\R^K$ records true-positive rates for $K$ classes, so larger coordinates mean better accuracy.
The set $\Gamma$ retains its role as the collection of vectors attained by allocation rules, as in Definition~\ref{def:true-positive-rate-region}.
Here we compute $\Vol_G(\Gamma)$.
The VUS in Definition~\ref{def:roc-body} is $\Vol_K(\cH)$, where $\cH$ is the downward closure of the attainable true-positive-rate region in the nonnegative orthant.



Using the change of measure in \eqref{eq:class-change-of-measure} gives
\begin{align}\begin{split}\label{eq:fa-error-integral}
T_g(\phi)
&=\int\bigl[\phi_1(x)\ell_g(1,x)+\phi_2(x)\ell_g(2,x)\bigr]\dd P_g(x)\\
&=\E_P\!\left[q_g(X)\phi_1(X)\ell_g(1,X)+q_g(X)\phi_2(X)\ell_g(2,X)\right].
\end{split}\end{align}
Intuitively, the rule averages the conditional losses using its decision probabilities.
Collecting these expressions across groups gives the counterpart of \eqref{eq:weighted-allocation-selection}.
For each group, define the loss difference $\gamma_g(x)\coloneqq q_g(x)[\ell_g(1,x)-\ell_g(2,x)]$ and write $\gamma(x)=(\gamma_1(x),\ldots,\gamma_G(x))^\top$.
Let $T_0\coloneqq T(1,0)$ be the error vector from always choosing decision 1.
Substituting $\phi_1=1-\phi_2$ in \eqref{eq:fa-error-integral} and collecting the group coordinates gives
\begin{align}\begin{split}\label{eq:fa-binary-representation}
T_g(\phi)
&=\E_P\!\left[q_g(X)\ell_g(1,X)-\phi_2(X)q_g(X)\{\ell_g(1,X)-\ell_g(2,X)\}\right],\\
T(\phi)&=T_0-\E_P[\phi_2(X)\gamma(X)].
\end{split}\end{align}
For each $x$, varying $\phi_2(x)$ over $[0,1]$ makes $\phi_2(x)\gamma(x)$ range over the segment $[0,\gamma(x)]$.
This is the segment $D_a(x)$ from Section~\ref{subsec:zonoid-determinants}, with $a=\gamma$ and $K=G$.
Equation \eqref{eq:zonoid-aumann} in Lemma~\ref{lem:zonoid-aumann} therefore gives $T_0-\Gamma=\E_A[D_a(X)]=\E_A[[0,\gamma(X)]]$.
Applying \eqref{eq:zonoid-vitale} in Proposition~\ref{prop:zonoid-vitale} yields the following volume formula.

\begin{proposition}
\label{prop:fa-aumann}\label{prop:fa-area}
The feasible error set can be represented as 
\begin{align}\begin{split}\label{eq:fa-aumann}
T_0-\Gamma=\E_A[[0,\gamma(X)]].
\end{split}\end{align}
For independent $X_1,\ldots,X_G\sim P$,
\begin{align}\begin{split}\label{eq:fa-segment-volume}
\Vol_G(\Gamma)
&=\Vol_G\bigl(\E_A[[0,\gamma(X)]]\bigr)\\
&=\E_{P^G}\!\left[V([0,\gamma(X_1)],\ldots,[0,\gamma(X_G)])\right]\\
&=\frac{1}{G!}\E_{P^G}\!\left[
\left|\det\bigl(\gamma(X_1),\ldots,\gamma(X_G)\bigr)\right|\right].
\end{split}\end{align}
All expectations are finite, including when $\Gamma$ has zero $G$-dimensional volume.
\end{proposition}

\begin{proof}\ 
Equation \eqref{eq:fa-binary-representation} gives
\begin{align}\begin{split}\nonumber
T_0-\Gamma
=\left\{\E_P[\phi_2(X)\gamma(X)]:\phi\text{ is an allocation rule}\right\}.
\end{split}\end{align}
The density ratios $q_g$ are bounded because each group has positive population share, so boundedness of the loss implies boundedness of $\gamma$.
Lemma~\ref{lem:zonoid-aumann} identifies this set with $\E_A[[0,\gamma(X)]]$, proving \eqref{eq:fa-aumann}.
The translation and reflection operations preserve volume, so Proposition \ref{prop:zonoid-vitale}, with $a=\gamma$ in dimension $G$, gives \eqref{eq:fa-segment-volume}.
\end{proof}

\begin{example}[Two groups]
\label{ex:fa-two-group-area}
\leavevmode\par
For two groups, label them $r,b$ and write $\gamma=(\gamma_r,\gamma_b)^\top$.
Taking $G=2$ in \eqref{eq:fa-segment-volume} gives the area
\begin{align}\begin{split}
\Vol_2(\Gamma)
&=\frac12\E_{P^2}\!\left[
\left|\det\begin{pmatrix}
\gamma_r(X_1)&\gamma_r(X_2)\\
\gamma_b(X_1)&\gamma_b(X_2)
\end{pmatrix}\right|\right]\\
&=\frac12\E_{P^2}\!\left[
\left|\gamma_r(X_1)\gamma_b(X_2)-\gamma_b(X_1)\gamma_r(X_2)\right|\right].
\end{split}\end{align}
\end{example}

\subsection{Gini coefficients and Lorenz volumes}
\label{subsec:gini-application}

The Lorenz curve $L(p)$ records the share of total income held by the poorest proportion $p\in[0,1]$ of the population \citep{lorenz1905wealth}.
The Gini coefficient compares the area between this curve and the line of equality with the area below that line \citep{gini1921measurement}.
For income $w$, this gives
\begin{align}\begin{split}\label{eq:gini-lorenz-index}
G(w)
\coloneqq\frac{\int_0^1[p-L(p)]\dd p}{\int_0^1p\dd p}
=1-2\int_0^1L(p)\dd p.
\end{split}\end{align}
The denominator is the area $1/2$ of the triangle below the line of equality, so $G(w)$ is twice the area between that line and the Lorenz curve, as illustrated in Figure~\ref{fig:lorenz-gini-area}.

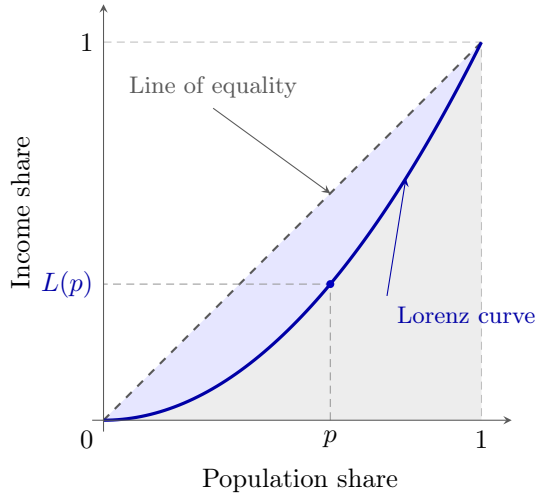
\begin{figure}[h]
\centering
\begin{tikzpicture}[x=1cm,y=1cm,>=stealth,font=\small]
\begin{scope}[x=5cm,y=5cm]
  \fill[black!7] (0,0)
    -- plot[domain=0:1,samples=81] (\x,{\x*\x})
    -- (1,0) -- cycle;
  \fill[blue!10] (0,0) -- (1,1)
    -- plot[domain=1:0,samples=81] (\x,{\x*\x}) -- cycle;
  \draw[black!25,densely dashed] (0,1) -- (1,1) -- (1,0);
  \draw[black!65,dashed,line width=0.8pt] (0,0) -- (1,1);
  \draw[blue!65!black,very thick]
    plot[domain=0:1,samples=81] (\x,{\x*\x});
  \draw[black!40,densely dashed] (0.60,0) -- (0.60,0.36) -- (0,0.36);
  \fill[blue!75!black] (0.60,0.36) circle (1.5pt);
  \draw[->,black!65] (-0.03,0) -- (1.08,0);
  \draw[->,black!65] (0,-0.03) -- (0,1.10);
  \node[below left] at (0,0) {$0$};
  \node[below] at (1,0) {$1$};
  \node[left] at (0,1) {$1$};
  \node[below] at (0.60,0) {$p$};
  \node[left,text=blue!65!black] at (0,0.36) {$L(p)$};
  \node at (0.52,-0.16) {Population share};
  \node[rotate=90] at (-0.22,0.55) {Income share};
  \node[text=black!65,font=\footnotesize] (equalitylabel) at (0.29,0.88) {Line of equality};
  \draw[->,black!65,thin] (equalitylabel.south) -- (0.60,0.60);
  \node[text=blue!65!black,font=\footnotesize,anchor=west] (lorenzlabel) at (0.75,0.28) {Lorenz curve};
  \draw[->,blue!65!black,thin] (lorenzlabel.north west) -- (0.80,0.64);
\end{scope}
\end{tikzpicture}
\makeatletter
\patchcmd{\@makecaption}{#1: #2}{#1. #2}{}{}
\makeatother
\caption{The Lorenz curve and the Gini coefficient}
\label{fig:lorenz-gini-area}
\begin{minipage}{0.94\linewidth}
\footnotesize
\textit{Note.} The illustrative Lorenz curve is $L(p)=p^2$.
The blue region lies between the curve and the line of equality, while the blue and gray regions together form the triangle below the line of equality.
This triangle has area $1/2$, so the Gini coefficient is twice the blue area.
The full attainable set $\Gamma$, defined below, also includes an equally sized region above the line of equality, and its area equals the Gini coefficient.
\end{minipage}
\end{figure}

\citet{koshevoy1996lorenz} extend this construction to multivariate welfare distributions by recording the population and welfare shares attainable by selecting a common subpopulation.  Other examples of multivariate extensions include \cite{gajdos2005multidimensional}, \cite{arnold2018analytic}, and \cite{grothe2022multivariate}, but none of them provide a general formula for the volume under the Lorenz surface in $K$ dimensions. 
We apply Lemma~\ref{lem:zonoid-aumann} and Proposition~\ref{prop:zonoid-vitale} to compute Lorenz volumes for $K$ welfare variables and recover the ordinary income Gini coefficient when $K=1$.

In this subsection, we derive the multivariate Gini coefficients under a setting closely related to the Inequality of Opportunity (IOp) measuring problem in \cite{escanciano2026debiased} and references therein. 
Retain $X\sim P$ from Section~\ref{subsec:allocation-roc-body}, now describing a randomly drawn individual.
Let $w_k(X)$, $k\in[K]$, be $K\geq1$ nonnegative measurable welfare variables with positive finite means, and write $w_{k,0}(x)\coloneqq w_k(x)/\E_P[w_k(X)]$ for their normalized values.
For example, $w_1(X)$ may denote income, $w_2(X)$ leisure, and $w_3(X)$ health. 
The $w_{0}(X) = (w_{1, 0}(X), \ldots, w_{K, 0}(X))$ can also be interpreted as the $X$ conditional expectation of a vector of welfare measures, i.e., there exists $Y = (Y_{1}, \ldots, Y_{K})$ such that 
\begin{align}\begin{split}\nonumber
w(X) = \left(\frac{\E\left[Y_{1} \vert X\right]}{\E Y_{1}}, \ldots, \frac{\E\left[Y_{K} \vert X\right]}{\E Y_{K}}\right).
\end{split}\end{align} 
Use a binary allocation rule $\phi=(\phi_1,\phi_2)$, where $\phi_2(X)$ is the probability that the individual is selected.
The selected population share and welfare shares are
\begin{align}\begin{split}
T_0(\phi)&\coloneqq\E_P[\phi_2(X)],\\
T_k(\phi)&\coloneqq\E_P[\phi_2(X)w_{k,0}(X)], \quad k\in[K].
\end{split}\end{align}
Collect these shares in $T(\phi)=(T_0(\phi),T_1(\phi),\ldots,T_K(\phi))^\top$.
For jointly feasible welfare shares $t=(t_1,\ldots,t_K)$, we maximize the selected population share subject to the requirement that the selected individuals hold a share $t_k$ of the total amount of $w_k$ for every $k\in[K]$:
\begin{align}\begin{split}
R(t)\coloneqq\max_{\phi}\quad &T_0(\phi)\\
\text{subject to}\quad &T_k(\phi)=t_k, \quad k\in[K],\\
&\phi\text{ is an allocation rule}.
\end{split}\end{align}
Intuitively, we ask how large a fraction of the population can be selected when the same individuals must collectively hold the prescribed shares of all welfare variables, so the answer depends on the joint welfare distribution.

\begin{example}[Joint welfare shares]
\label{ex:gini-joint-shares}
\leavevmode\par
When the two welfare variables are income and leisure, $R(0.20,0.30)$ is the largest fraction of the population that can be selected so that these same individuals hold both $20\%$ of total income and $30\%$ of total leisure, provided these shares are jointly feasible.
\end{example}

Following Definition 2.2 of \cite{koshevoy1996lorenz}, $R$ is the inverse Lorenz function and its graph is the \emph{Lorenz surface}.
The full set of attainable population and welfare shares is
\begin{align}\begin{split}\label{eq:gini-share-vector}
\Gamma\coloneqq\{T(\phi):\phi\text{ is an allocation rule}\}\subseteq[0,1]^{K+1}.
\end{split}\end{align}
As in Definition~\ref{def:true-positive-rate-region}, $\Gamma$ collects attainable vectors.
Fix feasible welfare shares $t$.
The upper boundary of $\Gamma$ records the largest population share that can hold these welfare shares, namely $R(t)$.
The lower boundary records the smallest population share that can hold the same welfare shares $t$.
If the selected individuals hold shares $t$, the remaining individuals hold shares $\mathbf 1-t$, where $\mathbf 1=(1,\ldots,1)^\top$.
By definition of $R$, the remaining individuals account for at most $R(\mathbf 1-t)$ of the population.
The selected individuals must therefore account for at least $1-R(\mathbf 1-t)$.
To attain this lower bound, take a population attaining $R(\mathbf 1-t)$ and select its complement by swapping $\phi_1$ and $\phi_2$.

\begin{example}[Complementary populations]
\label{ex:gini-complementary-populations}
\leavevmode\par
In the income-and-leisure setting of Example~\ref{ex:gini-joint-shares}, a selected population holding $20\%$ of total income and $30\%$ of total leisure leaves the remaining individuals with $80\%$ of total income and $70\%$ of total leisure.
Suppose $R(0.80,0.70)=0.90$, so the largest population share that can hold these remaining welfare shares is $90\%$.
Taking its complement gives a population share of $10\%$ holding the original welfare shares $(0.20,0.30)$.
No smaller population can hold these original shares, because its complement would exceed $90\%$ of the population while still holding welfare shares $(0.80,0.70)$, contradicting the stated maximum.
Thus the minimum population share at $(0.20,0.30)$ is $1-R(0.80,0.70)=0.10$.
\end{example}

Randomizing between the maximizing and minimizing allocation rules preserves welfare shares $t$ and yields every population share between the two extrema.
Writing $t_0$ for the population share, we therefore obtain
\begin{align}\begin{split}
\Gamma=\{(t_0,t):t\text{ is feasible},\ 1-R(\mathbf 1-t)\leq t_0\leq R(t)\}.
\end{split}\end{align}
The set $\Gamma$ is the Lorenz zonoid of the welfare distribution normalized by its coordinate means \citep[Definition 2.1]{koshevoy1996lorenz}.

To compute the volume of $\Gamma$, write $w_0(x)=(w_{1,0}(x),\ldots,w_{K,0}(x))^\top$ and $a(x)=(1,w_0(x)^\top)^\top$.
Then $T(\phi)=\E_P[\phi_2(X)a(X)]$.
Using \eqref{eq:gini-share-vector} and the representation \eqref{eq:zonoid-aumann} in Lemma~\ref{lem:zonoid-aumann} gives
\begin{align}\begin{split}\label{eq:gini-aumann}
\Gamma
&=\left\{\E_P[\phi_2(X)a(X)]:\phi\text{ is an allocation rule}\right\}\\
&=\int [0,a(x)]\dd P(x) 
=\E_A[[0,a(X)]].
\end{split}\end{align}
For normalization, we also consider the feasible welfare shares obtained by dropping the population coordinate from $\Gamma$.
This allows $T_0(\phi)$ to vary while retaining $(T_1(\phi),\ldots,T_K(\phi))^\top$.
Applying \eqref{eq:zonoid-aumann} to $w_0$ gives
\begin{align}\begin{split}\label{eq:gini-welfare-aumann}
&\left\{t\in[0,1]^K:(t_0,t)\in\Gamma\text{ for some }t_0\in[0,1]\right\}\\
& =\left\{\E_P[\phi_2(X)w_0(X)]:\phi\text{ is an allocation rule}\right\}\\
& =\int[0,w_0(x)]\dd P(x) 
=\E_A[[0,w_0(X)]].
\end{split}\end{align}
When this set has positive $K$-dimensional volume, define the multivariate Gini coefficient as the normalized Lorenz volume by
\begin{align}\begin{split}\label{eq:gini-normalized-volume}
G(w_1,\ldots,w_K)
\coloneqq\frac{\Vol_{K+1}(\Gamma)}{\Vol_K\bigl(\E_A[[0,w_0(X)]]\bigr)}.
\end{split}\end{align}
Proposition~\ref{prop:zonoid-vitale} evaluates both volumes, using ambient dimensions $K+1$ and $K$, respectively.

\begin{proposition}[Lorenz volumes]
\label{prop:gini-lorenz-volumes}
Under the assumptions above, 
\begin{align}\begin{split}\label{eq:gini-determinant-volumes}
\Vol_{K+1}(\Gamma)
=\frac{1}{(K+1)!}\E_{P^{K+1}}\left|
\det\bigl(a(X_1),\ldots,a(X_{K+1})\bigr)\right|.
\end{split}\end{align}
The feasible welfare-share set has volume
\begin{align}\begin{split}\label{eq:gini-welfare-volume}
\Vol_K\bigl(\E_A[[0,w_0(X)]]\bigr)
=\frac{1}{K!}\E_{P^K}\left|\det\bigl(w_0(X_1),\ldots,w_0(X_K)\bigr)\right|.
\end{split}\end{align}
When this volume is positive, the multivariate Gini coefficient in \eqref{eq:gini-normalized-volume} is
\begin{align}\begin{split}\label{eq:gini-determinant-ratio}
G(w_1,\ldots,w_K)
=\frac{\E_{P^{K+1}}\left|\det\bigl(a(X_1),\ldots,a(X_{K+1})\bigr)\right|}
{(K+1)\E_{P^K}\left|\det\bigl(w_0(X_1),\ldots,w_0(X_K)\bigr)\right|}.
\end{split}\end{align}
The draws $X_1,\ldots,X_{K+1}$ are independent, while variables within each individual may be dependent.
All expectations are finite, and \eqref{eq:gini-determinant-volumes} also holds when the volume is zero.
\end{proposition}

\begin{proof}\ 
Normalization gives $\E_P[\|w_0(X)\|]\leq K$ and $\E_P[\|a(X)\|]\leq K+1$, so both maps are integrable.
Lemma~\ref{lem:zonoid-aumann} therefore gives compactness, convexity, and \eqref{eq:gini-aumann}.
Proposition~\ref{prop:zonoid-vitale}, with ambient dimension $K+1$, yields \eqref{eq:gini-determinant-volumes}.

By \eqref{eq:gini-welfare-aumann}, the feasible welfare-share set is $\E_A[[0,w_0(X)]]$.
Applying the same proposition to $w_0$ in ambient dimension $K$ gives \eqref{eq:gini-welfare-volume}.
Dividing \eqref{eq:gini-determinant-volumes} by \eqref{eq:gini-welfare-volume} proves \eqref{eq:gini-determinant-ratio}.
\end{proof}

For fixed feasible welfare shares $t$, the largest and smallest attainable population shares differ by $R(t)+R(\mathbf 1-t)-1$.
Equation \eqref{eq:gini-determinant-ratio} averages this difference uniformly over the feasible welfare-share set with respect to $K$-dimensional volume.
Holding the same positive welfare shares with fewer individuals means higher average welfare in every coordinate.
The index therefore captures an aspect of inequality: the extent to which the same aggregate welfare shares can be held by populations of different sizes.
A larger index indicates a wider average gap between these population shares and hence greater inequality in this sense.
For income alone, this is precisely the ordinary Gini coefficient, as shown in Example~\ref{ex:gini-income}.
With multiple welfare variables, it measures this aspect of joint welfare concentration.

\begin{example}[Income inequality of opportunity]
\label{ex:gini-income}
For income alone, take $K=1$.
The determinant in \eqref{eq:gini-welfare-volume} is then simply $w_{1,0}(X)$.
Since $w_{1,0}(X)\geq0$ and $w_{1,0}(X)=w_1(X)/\E_P[w_1(X)]$, we obtain
\begin{align}\begin{split}\nonumber
\Vol_1\bigl(\E_A[[0,w_{1,0}(X)]]\bigr)
&=\E_P[|w_{1,0}(X)|]=\E_P[w_{1,0}(X)]=1.
\end{split}\end{align}
Expanding the determinant in \eqref{eq:gini-determinant-ratio} and substituting the normalized incomes gives
\begin{align}\begin{split}
G(w_1)=\Vol_2(\Gamma)
&=\frac12\E_{P^2}\left|\det\begin{pmatrix}
1&1\\
w_{1,0}(X_1)&w_{1,0}(X_2)
\end{pmatrix}\right|\\
&=\frac12\E_{P^2}\left|w_{1,0}(X_2)-w_{1,0}(X_1)\right|\\
&=\frac{\E_{P^2}|w_1(X_2)-w_1(X_1)|}{2\E_P[w_1(X)]}.
\end{split}\end{align}

To connect the expected income difference directly to the Lorenz curve, let $Q$ be the quantile function of normalized income $w_{1,0}(X)$.
The poorest population share $p$ holds income share $L(p)=\int_0^p Q(u)\dd u$, with $L(1)=1$.
For a random variable $U \sim U(0, 1)$, the uniform distribution on $(0, 1)$, we have that $Q(U)$ has the same distribution as $w_{1, 0}(X)$. 
Therefore, the independence of $X_{1}$ and $X_{2}$ gives 
\begin{align}\begin{split}\nonumber
\frac{\E_{P^2}|w_1(X_2)-w_1(X_1)|}{2\E_P[w_1(X)]}
&=\frac12\int_0^1\int_0^1|Q(v)-Q(u)|\dd u\dd v\\
&=\int_0^1\int_0^v[Q(v)-Q(u)]\dd u\dd v\\
&=\int_0^1[pQ(p)-L(p)]\dd p\\
&=\bigl[pL(p)\bigr]_0^1-2\int_0^1L(p)\dd p 
=1-2\int_0^1L(p)\dd p.
\end{split}\end{align}
The second line restricts the integral to $u\leq v$, where $Q(v)\geq Q(u)$, and cancels the factor $1/2$ by symmetry.
The third line evaluates the inner integral, and the fourth uses integration by parts with $L'(p)=Q(p)$ almost everywhere.
This argument also allows income distributions with atoms.

The same identity can also be read from the geometry of $\Gamma$.
Fix a population share $p$.
The smallest attainable income share is $L(p)$, while the largest is $1-L(1-p)$, obtained by taking the complement of the poorest fraction $1-p$.
Integrating the difference between these boundaries gives
\begin{align}\begin{split}\nonumber
G(w_1)=\Vol_2(\Gamma)
&=\int_0^1\bigl[1-L(1-p)-L(p)\bigr]\dd p\\
&=1-\int_0^1L(1-p)\dd p-\int_0^1L(p)\dd p 
=1-2\int_0^1L(p)\dd p.
\end{split}\end{align}
In the last line, the substitution $u=1-p$ gives $\int_0^1L(1-p)\dd p=\int_0^1L(u)\dd u$.
This is the ordinary Gini coefficient in \eqref{eq:gini-lorenz-index}.
\end{example}

\section{Conclusion}\label{conclusion section}

While the area under the curve (AUC) is one of the most prominent measures of binary prediction performance,
there is no widely accepted generalization regarding its multi-class correspondence. 
This paper provides a novel random polytope method to calculate the population volume under the surface (VUS). 
Our method builds on the theory of Minkowski mixed volumes in convex geometry, and can work for any underlying data generating process, including those with atoms or singular components. 
The population VUS takes a form of expectation of symmetric U-statistics kernel. 
We then propose a double/debiased machine learning U-statistic estimator of the VUS, derive its asymptotic properties, and develop an asymptotically valid inference procedure.
The volume calculation method can also be applied to other problems in which the sets of interest admit representations as projections of critical function sets.
Examples include the volume of the feasible error set across pre-defined groups and a generalized multivariate Gini coefficient.

\addcontentsline{toc}{section}{References}
\bibliography{overall}
\bibliographystyle{aer}

\end{document}